\documentclass[10pt,twocolumn,letterpaper]{article}
\usepackage{usenix2019_v3}
\usepackage{url}
\usepackage{adjustbox}
\usepackage{amssymb}
\usepackage{amsmath}
\usepackage{amsthm}
\usepackage[notcomma,notperiod,notquote,notexcl,notcolon,notscolon]{hanging}
\usepackage{booktabs}
\usepackage[font={small}]{caption}
\usepackage[sort]{cite}
\usepackage{pifont}
\usepackage{color}
\usepackage{colortbl}
\usepackage{enumitem}
\usepackage{framed}
\PassOptionsToPackage{hyphens}{url}
\usepackage[pdfstartview=FitH,hidelinks,
    pdfpagelabels,bookmarksdepth=3,bookmarksopen=true,          bookmarksnumbered
]{hyperref}
\usepackage{nicefrac}
\usepackage{mathtools}
\usepackage{microtype}
\usepackage{multicol}
\usepackage{multirow}
\usepackage{scalefnt}
\usepackage{setspace}
\usepackage{thm-restate}
\usepackage{tikz}
\usetikzlibrary{decorations.pathreplacing}
\usetikzlibrary{calc}
\usetikzlibrary{hobby}
\usetikzlibrary{decorations.markings, decorations.pathmorphing}
\usetikzlibrary{shapes}
\usetikzlibrary{positioning,fit}
\usepackage{xcolor}
\usepackage{xspace}
\usepackage{titlesec}
\usepackage{wrapfig}
\usepackage{siunitx}
\usepackage[capitalise,noabbrev]{cleveref}
\usepackage{algpseudocode}
\usepackage{subfigure}
\usepackage[subtle,wordspacing=normal,mathspacing=normal]{savetrees}
\usepackage[available,functional,reproduced]{usenixbadges}

\newcommand{\N}{\mathbb{N}}
\newcommand{\Z}{\mathbb{Z}}

\newcommand{\G}{\mathbb{G}}

\newcommand{\Hash}{\mathsf{Hash}}

\newcommand{\Commit}{\mathsf{Commit}}

\newcommand{\todo}[1]{}

\newcommand{\calA}{\ensuremath{\mathcal{A}}}
\newcommand{\calB}{\ensuremath{\mathcal{B}}}

\newcommand{\calE}{\ensuremath{\mathcal{E}}}
\newcommand{\calF}{\ensuremath{\mathcal{F}}}

\newcommand{\calM}{\ensuremath{\mathcal{M}}}

\newcommand{\zo}{\ensuremath{\{0,1\}}}

\theoremstyle{plain}
\newtheorem{theorem}{Theorem}

\newtheorem{lemma}[theorem]{Lemma}
\newtheorem{claim}[theorem]{Claim}

\theoremstyle{definition}

\theoremstyle{remark}

\newcommand{\esm}[1]{\ensuremath{#1}}
\newcommand{\ms}[1]{\esm{\mathsf{#1}}}

\newcommand{\prf}{\ms{PRF}}

\newcommand{\getsr}{\rgets}
\newcommand{\rgets}{\mathrel{\mathpalette\rgetscmd\relax}}
\newcommand{\rgetscmd}{\ooalign{$\leftarrow$\cr
    \hidewidth\raisebox{1.2\height}{\scalebox{0.5}{\ \rm R}}\hidewidth\cr}}

\newcommand{\ct}{\ms{ct}}

\newcommand{\View}{\ms{View}}
\newcommand{\Real}{\ms{Real}}
\newcommand{\Ideal}{\ms{Ideal}}
\newcommand{\Sim}{\ms{Sim}}

\newcommand{\presigctr}{{\mathsf{ctr}_{\mathsf{presig}}}}
\newcommand{\authctr}{{\mathsf{ctr}_{\mathsf{auth}}}}
\newcommand{\tweak}{\omega}
\newcommand{\sk}{\ms{sk}}
\newcommand{\pk}{\ms{pk}}
\newcommand{\id}{\ms{id}}

\newcommand{\Sign}{\mathsf{Sign}}
\newcommand{\Verify}{\mathsf{Verify}}

\newcommand{\Enc}{\mathsf{Enc}}

\newcommand{\Dec}{\mathsf{Dec}}

\newcommand{\chal}{\mathsf{chal}}

\newcommand{\PreSign}{\mathsf{PreSign}}

\newcommand{\ClientRegister}{\mathsf{ClientRegister}}
\newcommand{\LogRegister}{\mathsf{LogRegister}}
\newcommand{\FinishRegister}{\mathsf{FinishRegister}}

\newcommand{\ClientAuth}{\mathsf{ClientAuth}}
\newcommand{\LogAuth}{\mathsf{LogAuth}}

\newcommand{\FinishAuth}{\mathsf{FinishAuth}}

\newcommand{\LarchPW}{\mathsf{Larch}_\mathsf{PW}}

\newcommand{\ECDSAModAdv}{\mathsf{ECDSAAdv}}

\newcommand{\KeyGen}{\mathsf{Gen}}

\newcommand{\LogKeyGen}{\mathsf{LogKeyGen}}
\newcommand{\ClientKeyGen}{\mathsf{ClientKeyGen}}
\newcommand{\pw}{\mathsf{pw}}

\newcommand{\kid}{\mathsf{k}_\mathsf{id}}
\newcommand{\klog}{\mathsf{klog}}
\newcommand{\kclient}{\mathsf{kclient}}

\newcommand{\HMAC}{\mathsf{HMAC}}

\newcommand{\presig}{\mathsf{presig}}
\newcommand{\Prove}{\mathsf{Prove}}

\newcommand{\cm}{\mathsf{cm}}
\newcommand{\idx}{\mathsf{idx}}

\newcommand{\ECDSAAdv}{\mathsf{GSECDSAAdv}}

\newcommand{\PiMul}{\Pi_\mathsf{HalfMul}}
\newcommand{\PiPreSign}{\Pi_\mathsf{PreSign}}
\newcommand{\PiGen}{\Pi_\mathsf{Gen}}
\newcommand{\PiHalfMul}{\Pi_\mathsf{HalfMul}}

\newcommand{\digest}{\mathsf{dgst}}

\newcommand{\PiOpen}{\Pi_\mathsf{Open}}

\newcommand{\FOpen}{\calF_\mathsf{Open}}
\newcommand{\FECDSA}{\calF_\mathsf{ECDSA}}

\newcommand{\SimOpen}{\Sim_\mathsf{Open}}
\newcommand{\Adv}{\calA}
\newcommand{\sksiglog}{\mathsf{sk}_0}
\newcommand{\pksiglog}{\mathsf{pk}_0}
\newcommand{\sksigclient}{\mathsf{sk}_1}
\newcommand{\initdone}{\mathsf{initdone}}
\newcommand{\presigdone}{\mathsf{presigdone}}

\newcommand{\ClientGen}{\mathsf{ClientGen}}
\newcommand{\LogGen}{\mathsf{LogGen}}

\newcommand{\DLProof}{\mathsf{DLProof}}

\newcounter{ExperimentCount}
\newcommand{\Experiment}[1]{\refstepcounter{ExperimentCount}\textbf{Experiment~\arabic{ExperimentCount}: #1.}}
\newcounter{PropertyCount}

\newif\iffull
\fulltrue

\makeatletter
\g@addto@macro{\UrlBreaks}{\UrlOrds}
\makeatother

\newlength{\defhangindent}
\usepackage[T1]{fontenc}
\usepackage{newtxmath}
\usepackage{newtxtext}

\usepackage[scaled=.8]{beramono}

\DeclareMathAlphabet\mathcal{OMS}{cmsy}{m}{n}
\SetMathAlphabet\mathcal{bold}{OMS}{cmsy}{b}{n}

\SetMathAlphabet{\mathsf}{normal}{OT1}{cmss}{m}{n}
\SetMathAlphabet{\mathsf}{bold}{OT1}{cmss}{bx}{n}

\DeclareSymbolFont{AMSb}{U}{msb}{m}{n}
\DeclareSymbolFontAlphabet{\mathbb}{AMSb}

\DeclareSymbolFont{numbers}{T1}{ptm}{m}{n}
\SetSymbolFont{numbers}{bold}{T1}{ptm}{bx}{n}
\DeclareMathSymbol{0}\mathalpha{numbers}{"30}
\DeclareMathSymbol{1}\mathalpha{numbers}{"31}
\DeclareMathSymbol{2}\mathalpha{numbers}{"32}
\DeclareMathSymbol{3}\mathalpha{numbers}{"33}
\DeclareMathSymbol{4}\mathalpha{numbers}{"34}
\DeclareMathSymbol{5}\mathalpha{numbers}{"35}
\DeclareMathSymbol{6}\mathalpha{numbers}{"36}
\DeclareMathSymbol{7}\mathalpha{numbers}{"37}
\DeclareMathSymbol{8}\mathalpha{numbers}{"38}
\DeclareMathSymbol{9}\mathalpha{numbers}{"39}

\SetSymbolFont{numbers}{bold}{T1}{ptm}{b}{n}
\makeatletter
\renewcommand{\operator@font}{\mathgroup\symnumbers}
\makeatother

\makeatletter

\definecolor{LightCyan}{rgb}{0.88,1,1}

\long\def\com#1{}

\long\def\xxx#1{}
\long\def\emma#1{}
\long\def\hcg#1{}

\newcommand{\itpara}[1]{\smallskip\noindent\textit{#1}}

\renewcommand{\paragraph}[1]{\smallskip\noindent\textbf{#1}}

\makeatletter

\let\c@table\c@figure
\makeatother 

\setlist[description]{leftmargin=\parindent,topsep=0ex,itemsep=0ex,partopsep=1ex,parsep=1ex}

\LetLtxMacro{\oldtextsc}{\textsc}
\renewcommand{\textsc}[1]{\oldtextsc{\scalefont{1.2}#1}}

\setitemize{noitemsep,topsep=2pt,parsep=2pt,partopsep=2pt,leftmargin=4ex}
\setenumerate{noitemsep,topsep=2pt,parsep=2pt,partopsep=2pt,leftmargin=4ex}

\newcolumntype{R}[2]{>{\adjustbox{angle=#1,lap=\width-(#2)}\bgroup}c<{\egroup}}

\begin{document}

\def\sys{larch\xspace}
\def\Sys{Larch\xspace}

\title{Accountable authentication with privacy protection: \\The \Sys system for universal login}
\author{Emma Dauterman\\
\emph{UC Berkeley}
\and
Danny Lin\\
\emph{Woodinville High School}
\and
Henry Corrigan-Gibbs\\
\emph{MIT}
\and
David Mazi\`eres\\
\emph{Stanford}
}
\date{}
\maketitle

\frenchspacing
\paragraph{Abstract.}
Credential compromise is hard to detect and hard to
mitigate. To address this problem, we present \sys, an
accountable authentication framework with strong security and
privacy properties. \Sys protects user privacy while
ensuring that the \sys log server correctly
records every authentication. 
Specifically, an attacker who compromises
a user's device cannot authenticate without creating evidence
in the log, and the log cannot learn which web service
(relying party) the user is authenticating to.
To enable fast adoption, \sys is backwards-compatible with
relying parties that support FIDO2, TOTP, and password-based
login. Furthermore, \sys does not degrade the security and
privacy a user already expects: the log server cannot
authenticate on behalf of a user, and \sys does not allow
relying parties to link a user across accounts.
We implement \sys for FIDO2, TOTP, and password-based login.
Given a client with four cores and a log server with eight cores,
an authentication with \sys takes 150ms for FIDO2, 91ms for
TOTP, and 74ms for passwords (excluding preprocessing, which
takes 1.23s for TOTP).
 
\section{Introduction}

Account security is a perennial weak link in computer systems.  Even
well-engineered systems with few bugs become vulnerable once
human users are involved.  With poorly engineered or configured
systems, account compromise is often the first of several cascading
failures.
In general, 82\% of data breaches
involve a human element, with the most common methods including
use of stolen credentials (40\%) and
phishing (20\%)~\cite{dbir}.

When users and administrators identify stolen credentials, it is
challenging to determine the extent of the damage.  Not knowing what an
attacker accessed can lead to either inadequate or overly extensive
recovery.  LastPass suffered a breach in November 2022 because they
didn't fully recover from a compromise the previous
August~\cite{roth:lastpass}.  Conversely, Okta feared 366
organizations might have been accessed when an attacker gained remote
desktop access at one of their vendors.  It took a three-month
investigation to determine that, in fact, only two organizations, not
366, had really been victims of the breach~\cite{faife:okta}.

Single sign-on schemes, such as OpenID~\cite{sakimura:openid} and ``Sign in with Google,''
can keep an authentication log and thereby 
determine the extent of a credential compromise.
However, these centralized systems represent a security
and privacy risk:
they give a third party access to all of a user's accounts 
and to a trace of their authentication activity.

An ideal solution would give the benefits of universal 
authentication logging without the security and privacy
drawbacks of single-sign-on systems.
For security, the logging service shouldn't be able to authenticate
on behalf of a user. For privacy, the logging service should learn
no information about a user's authentication history: the log service
should not even learn if the user is authenticating to the same
web service twice or to two separate web services.

In this paper, we propose \sys (``login archive''), an accountable
authentication framework with strong security and privacy properties.
Authentication takes place between a user and a service, which we call
the \emph{relying party}.  In \sys, we add a third party: a
user-chosen \sys log service.  The \sys log service provides the user
with a complete, comprehensive history of her authentication
activity, which helps users detect and recover from compromises.  Once
an account is registered with \sys, even an attacker who controls the
user's client cannot authenticate to the account without the \sys log
service storing a record that allows the user to recover the time and
relying-party name.

The key challenge in \sys is allowing the log service to maintain a
complete authentication history \emph{without becoming a single point
of security or privacy failure}.
A malicious \sys log service cannot access users' accounts
and learns no information about users' authentication histories.
Only users can decrypt their own log records.

\Sys works with any relying party that supports one of three standard
user authentication schemes:  FIDO2~\cite{fido2} (popularized by
Yubikeys and Passkeys~\cite{passkeys}), TOTP~\cite{totp-rfc}
(popularized by Google Authenticator), and password-based login.
FIDO2 is the most secure but least widely deployed of the three
options.

A \sys deployment consists of two components: a browser add-on, which
manages the user's authentication secrets, and one or more \sys log
services, which store authentication logs on behalf of a set of users.
At a high level, \sys provides four operations.  (1)~Upon deciding to
use \sys, a user performs a one-time \emph{enrollment} with a log
service.  (2)~For each account to use with \sys, the user
runs \emph{registration}.  To relying parties, registration looks like
adding a FIDO2 security key, adding an authenticator app, or setting a
password.  (3)~The user then performs
\emph{authentication} with \sys as necessary to access
registered accounts.  Finally, (4) at any point the user
can \emph{audit} login activity by downloading and decrypting the
complete history of authentication events to all accounts.
The client can use auditing
for intrusion detection or to evaluate the extent of the
damage after a client has been compromised.

All authentication mechanisms require generating an authentication
credential based on some secret.  In FIDO2, the secret is a signature
key and the credential is a digital signature; the signed payload
depends on the name of the relying party and a fresh challenge,
preventing both phishing and credential reuse.  With TOTP, the secret
is an HMAC key and the credential an HMAC of the current time, which
prevents credential reuse in the future.
With passwords, the credential
is simply the password, which has the disadvantage that it can be
reused once a malicious client obtains it.

\Sys splits the authentication secret between the client and log
service so that both parties must participate in authentication.  We
introduce split-secret authentication protocols for FIDO2, TOTP, and password-based
login.  At the end of each protocol, the log service holds an encrypted
authentication log record and the client holds a credential.  
\Sys ensures that if the client obtains a valid credential, the log
service also obtains a well-formed log record, even if the client is
compromised and behaves maliciously. At the same time, the log service learns
no information about the relying parties that the user authenticates to.

We design \sys to achieve the following (informal) security and privacy goals:
\begin{itemize}
\item \emph{Log enforcement against a malicious client:} An attacker that compromises a client cannot authenticate
to an account that the client created before compromise
without the log obtaining a well-formed, encrypted log
record.
\item \emph{Client privacy and security against a malicious log:}
A malicious log service cannot authenticate to the 
user's accounts or learn any information about the relying parties to
which the user has
authenticated, including whether two authentications
are for the same account or different accounts.
\item \emph{Client privacy against a malicious relying party:}
Colluding malicious relying parties cannot link a
user across accounts.
\end{itemize}

\Sys's FIDO2 protocol uses zero-knowledge proofs~\cite{GMR85} to
convince the log that an encrypted authentication log record generated
by the client is well-formed relative to the digest of a FIDO2
payload.  If it is, the client and log service sign the digest with a new,
lightweight two-party ECDSA signing protocol tailored to our setting.
For TOTP, \sys executes an authentication circuit using an existing
garbled-circuit-based multiparty computation
protocol~\cite{Yao82,WRK17}.  For password-based login, the client
privately swaps a ciphertext encrypting the relying party's identity
for the log's share of the corresponding password using
a discrete-log-based protocol~\cite{GK14}.

In the event that a user's device is compromised, a user can revoke access to
all accounts---even accounts she may have forgotten
about---by interacting only with the log service.  At the same time,
involving the log service in every authentication could pose a
reliability risk (just as relying on OpenID does).
We show how
to split trust across multiple log service providers to strengthen
availability guarantees, making \sys strictly better than OpenID for
all three of security, privacy, and availability.

We expect users to perform many password-based authentications, some
FIDO2 authentications, and a comparatively small number of TOTP
authentications.  Given a client with four cores and a log server with eight
cores, an authentication with \sys takes 150ms for FIDO2, 91ms for
TOTP, and 74ms for passwords (excluding preprocessing, which takes
1.23s for TOTP). One authentication requires 1.73MiB of
communication for FIDO2, 65.2MiB for TOTP, and 3.25KiB for
passwords.  TOTP communication costs are
comparatively high because we use garbled circuits~\cite{WRK17};
however, all but 202KiB of the communication can be moved into a
preprocessing step.

\Sys shows that it is possible to achieve privacy-preserving
authentication logging that is backwards compatible with existing standards.
Moreover, \sys provides new paths for FIDO2 adoption, as \sys
users can authenticate using FIDO2 without dedicated hardware
tokens, which could motivate more relying parties to deploy FIDO2.
Users who do own hardware tokens can use them to authenticate to
the \sys log service, providing strong security guarantees for
relying parties that do not yet support FIDO2 (albeit without
the anti-phishing protection).
We also suggest small changes to the FIDO standard that would
substantially reduce the overheads of \sys while providing the
same security and privacy properties.
 
 \section{Design overview}
\label{sec:overview}

We now give an overview of \sys.

\subsection{Entities}
\label{sec:overview:entities}
A \sys deployment 
involves the following entities:

\paragraph{Users.}
We envision a deployment with millions of users, each of which has
hundreds of accounts at different online services---shopping websites,
financial institutions, news sites, and so on.  Each user has an
account at a \sys log service, secured by a strong, unique password
and optionally (but ideally) strong second-factor authentication such
as a FIDO2 hardware security key. (In \cref{sec:mal}, we describe how
a user can create accounts with multiple log services in order to
protect against faulty logs.) A user also has a set of devices
(e.g. laptop, phone, tablet) running \sys client software and storing
\sys secrets, including cryptographic keys and passwords.

\paragraph{Relying parties.}
A relying party is any website that a user authenticates to (e.g.,
a shopping website or bank).  \Sys is compatible with any relying
party that supports authentication via FIDO2
(U2F)~\cite{fido2,webauthn}, time-based one-time passwords
(TOTP)~\cite{totp-rfc}, or standard passwords.
The strength of \sys's security guarantees depends on the strength
of the underlying authentication method.

\paragraph{Log service.}
Whenever the user authenticates to a relying party, the client
must communicate with the log service. We envision a major
service provider (e.g. Google or Apple) deploying this service
on behalf of their customers.
The log service:
\begin{itemize}
\item keeps an encrypted record of the user's
      authentication history, but
\item learns no information about which relying party the user
  authenticates to.
\end{itemize}
At any time, a client can fetch this authentication
record from the log service and decrypt it to see the user's
authentication history. 
That is, if an attacker compromises one of
Alice's devices and authenticates to \texttt{github.com} as Alice, the
attacker will leave an indelible trace of this authentication in the
\sys log.  At the same time, to protect Alice's privacy, the log
service learns no information about which relying parties Alice has
authenticated to.  A production log service should consist of
multiple, georeplicated servers to ensure high availability.

\subsection{Protocol flow}
\label{sec:overview:protocol}
\label{sec:overview:flow}

\paragraph{Background.} 
We use two-out-of-two \emph{additive secret sharing}~\cite{S79}: to
secret-share a value $x \in \{0, \dots, p-1\}$, choose random values
$x_1,x_2 \in \{0, \dots, p-1\}$ such that $x_1 + x_2 = x \bmod p$.
Neither $x_1$ nor $x_2$ individually reveals any information
about~$x$.  We also use a cryptographic \emph{commitment scheme}:  to
commit to a value $x \in \zo^*$, choose a random value $r \in
\zo^{256}$ (the commitment \emph{opening}) and output the hash of $(x
\| r)$ using a cryptographic hash function such as SHA-256.  For
computationally bounded parties, the commitment reveals no information
about $x$, but makes it impractical to convince another party that the
commitment opens to a value $x' \neq x$.

\smallskip

The client's interaction with the log service consists of four
operations.

\paragraph{Step 1: Enrollment with a log service.}
To use \sys, a user must first \emph{enroll} with a \sys log service
by creating an account.  In addition to configuring traditional
account authentication (i.e., setting a password and optionally
registering FIDO2 keys), the user's client generates a secret \emph{archive
key} for each authentication method supported.  For FIDO2 and TOTP,
the archive key is a symmetric encryption key, and the client sends
the log service a commitment to this key.  For passwords, the archive
key is an ElGamal private encryption key, so the client sends the log
service the corresponding public key.  The client subsequently
encrypts log records using these archive keys, while the log service
verifies these log records are well-formed using the corresponding
commitment or public key.

\paragraph{Step 2: Registration with relying parties.}
After the user has enrolled with a log service, she can create
accounts at relying parties (e.g., \texttt{github.com}) using
\sys-protected credentials.  We call this process \emph{registration}.
Registration works differently depending on which authentication
mechanism the relying party uses:  FIDO2 public-key authentication,
TOTP codes, or standard passwords.  All generally follow the same
pattern where at the conclusion of the registration protocol:
\begin{itemize}
  \item the log service holds an encryption of the relying
        party's identity under a key that only the client knows,
  \item the log service and client jointly hold
    the account's authentication secret using
    two-out-of-two \emph{secret sharing}~\cite{S79},
  \item the relying party is unaware of \sys and holds the
    usual information necessary to verify account access:
        an ECDSA public key (for FIDO2),
        an HMAC secret key (for TOTP), or  
        a password hash (for password-based login), and
  \item the log service learns nothing about the identity
        of the relying party.
\end{itemize}
By splitting the user's authentication secret between the client and
the log, we ensure that the log service participates in all of the
user's authentication attempts, which allows the log service to
guarantee that every authentication attempt is correctly logged.

The underlying authentication mechanisms (FIDO2, TOTP, and
password-based login) only provide security for a given relying party
if the user's device was uncompromised at the time of registration; 
\sys provides the same guarantees.

\paragraph{Step 3: Authentication to a relying party.}
Registering with a relying party lets the user later authenticate to
that relying party (\cref{fig:system-arch}).  At the conclusion of an
authentication operation, \sys must ensure that:
\begin{itemize}
  \item authentication succeeds at the relying party,
  \item the log service holds a record of the authentication attempt
    that \emph{includes the name of the relying party}, encrypted
    under the archive key known only to the client, and
  \item the log service learns \emph{no information} about the
    identity of the relying party involved.
\end{itemize}
The technical challenge here is guaranteeing that a compromised client
cannot successfully authenticate to a relying party without creating a
valid log record. In particular, the log service must verify that the
log record contains a valid encryption of the relying party's name 
under the archive key \emph{without} learning anything about the relying
party's identity.

To achieve these goals, we design \emph{split-secret authentication
protocols} that allow the client and log to use their split
authentication secrets to jointly produce an authentication
credential.  Our split-secret authentication protocols are essentially
special-purpose two-party computation protocols~\cite{Yao86}.  In a
two-party computation, each party holds a secret input, and the
protocol allows the parties to jointly compute a function on their
inputs while keeping each party's input secret from the other.  Our
split-secret authentication protocols follow a general pattern,
although the specifics depend on the underlying authentication
mechanism in use (FIDO2, TOTP, or password-based login):
\begin{itemize}
    \item The client algorithm takes as input the identity of the
      relying party, the client's share of the corresponding
      authentication secret, the archive key, and the opening for the
      log service's commitment to the archive key.
    \item The log algorithm takes as input its shares of
      authentication secrets and the client's commitment to the
      archive key (which it received at enrollment).
    \item The client algorithm outputs an authentication credential:
      a signature (for FIDO2), an HMAC code (for TOTP), or a password
      (for password-based login).
    \item The log algorithm outputs an encryption of the relying party
      identifier under the archive key.
\end{itemize}
In this way, the client and log service jointly generate
authentication credentials while guaranteeing that every successful
authentication is correctly logged.  The client and log do not learn
any information beyond the outputs of the computation.  We use this
general pattern to construct split-secret authentication protocols for
FIDO2 (\cref{sec:fido2}), TOTP (\cref{sec:totp}), and password-based
login (\cref{sec:pw}).

\begin{figure}[t]
    \centering
    \includegraphics[width=\columnwidth]{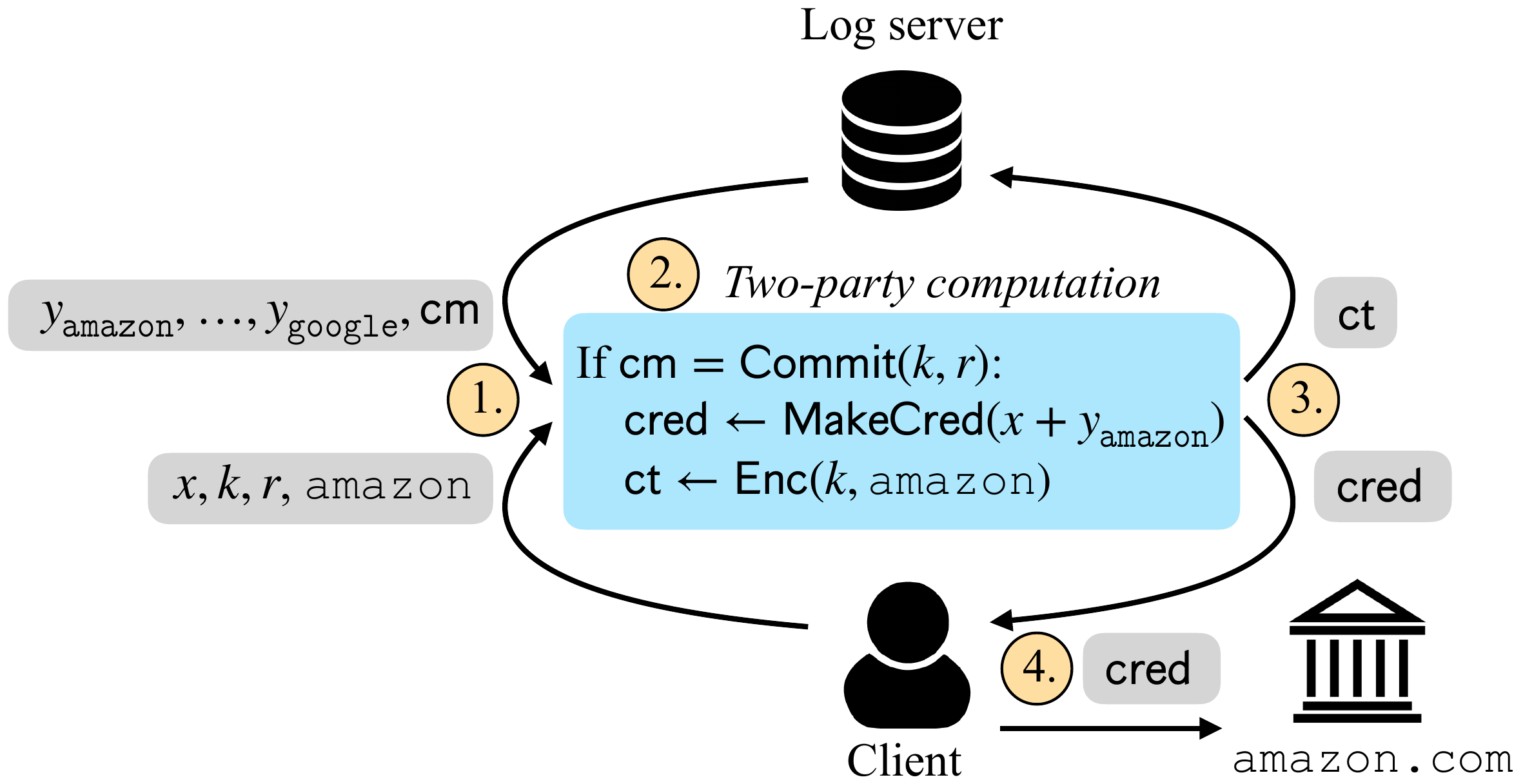}
    \caption{The client and log service run split-secret
      authentication where the client obtains the credential for
      \texttt{amazon.com} and the log service obtains an encryption of
      \texttt{amazon.com} under the client's key.  The client's inputs
      are its share $x$ of the authentication secret, the archive key
      $k$, a random nonce $r$, and the string \texttt{amazon.com}.
      The log's inputs are its shares $y_\texttt{amazon}, \dots,
      y_\texttt{google}$ of all the client's authentication secrets
      and the commitment $\cm$ to the archive key generated at
      enrollment. The $\mathsf{MakeCred}$ function takes extra inputs
      for FIDO2 and TOTP.}
    \label{fig:system-arch}
\end{figure}

\paragraph{Step 4: Auditing with the log.}
Finally, at any time, the user can ask the log service for its
collection of log entries encrypted under the archive key.  A user
could do this when she suspects that an attacker has compromised her
credentials.  The user's client could also perform this auditing in
the background and notify the user if it ever detects anomalous
behavior.  The client uses the encryption key it generated during
enrollment to decrypt log entries.

\subsection{System goals}
\label{sec:overview:sec-goals}

We now describe the security goals of \sys (\cref{fig:sec-goals}).

\paragraph{Goal 1: Log enforcement against a malicious client.}
Say that an honest client enrolls with an honest log service and then
registers with a set of relying parties.  Later on, an attacker
compromises the client's secrets (e.g., by compromising one of the
user's devices and causing it to behave maliciously).  Every
successful authentication attempt that the attacker makes using
credentials managed by \sys will appear in the client's authentication
log stored at the \sys log service.  Furthermore, the honest client
can decrypt these log entries using its secret key.

\paragraph{Goal 2: Client privacy and security against a malicious log.}
Even if the log service deviates arbitrarily from the prescribed
protocol, it learns no information about
(a) the client's authentication secrets (meaning that the log
service cannot authenticate on behalf of the client) or 
(b) which relying parties a client has interacted with.

\paragraph{Goal 3: Client privacy against a malicious relying party.}
A set of colluding malicious relying parties learn no information
about which registered accounts belong to the same client.
That is, relying parties cannot link a client across multiple relying parties
using information they learn during registration or authentication.

\begin{figure*}[t]
    \centering
    \includegraphics[width=\textwidth]{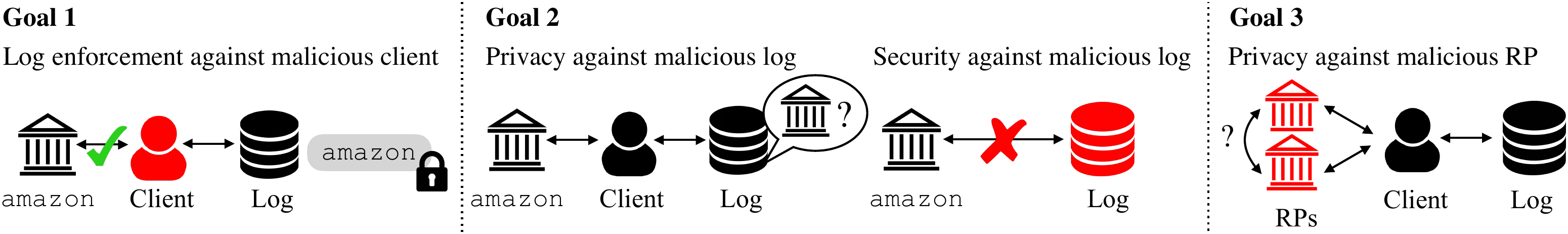}
    \caption{\Sys security goals.}
    \label{fig:sec-goals}
\end{figure*}

\label{sec:overview:functionality-goals}
\medskip

To be usable in practice, 
\sys should additionally achieve the following functionality
goal:

\paragraph{Goal 4: No changes to the relying party.}
Relying parties that support FIDO2 (U2F), TOTP, or password
authentication do not need to be aware of \sys.  Clients can
unilaterally register authentication credentials such that all future
authentications are logged in \sys.

\subsection{Non-goals and extensions}
\label{sec:overview:non}

\itpara{Availability against a compromised log service.}
\Sys does not provide availability if the log service
refuses to provide service.
We discuss defenses against availability attacks in
\cref{sec:mal}.

\itpara{Privacy against colluding log and relying party.}  If the log
service colludes with a relying party, they can always use timing
information to map log entries to authentication requests.  Therefore,
\sys makes no effort to obscure the relationship between private
messages seen by the two parties and only guarantees privacy when the
relying party and log service do not collude.

\itpara{Limitations of underlying authentication schemes.}  \Sys
provides security guarantees that match the security of the underlying
authentication schemes.  FIDO2 provides the strongest security,
followed by TOTP, and then followed by passwords.  For TOTP and
password-based login, \sys provides no protection against
\emph{credential breaches}: if an attacker steals users'
authentication secrets (MAC keys or passwords) from the relying party,
the attacker can use those secrets to authenticate without those
authentications appearing in the log.  FIDO2 defends against
credential breaches because the relying party only ever sees the
client's public key.

\Sys does protect against \emph{device compromise} for all three
authentication mechanisms: even if an attacker gains control of a
user's device, generating any of the user's \sys-protected credentials
requires communicating with the log service and results in an archived
log record.  If the user discovers the device break-in later on, she
can recover from the log a list of authentications and take steps to
remediate the effects of compromise (contacting the affected relying parties, etc.).

An attacker who compromises an account can often disable two-factor
authentication or add its own credentials to a compromised account.
Therefore, only an attacker's first successful access to a given relying
party is guaranteed to be archived in \sys.  That said, many relying
parties send out notifications, require step-up authentication, or
revoke access to logged in clients on credential updates, all of which
could complicate an attack or alert legitimate users to a problem.
Hence, it is valuable to ensure that all accesses with the original
account credentials are logged.
\Sys can make this guarantee for FIDO2, where every authentication
requires a unique two-party signature.  It does not provide this guarantee
with passwords, as the attacker learns the password as part of the
authentication process: only the attacker's first authentication to
a given relying party will be logged. With TOTP, each generated code produces a
\sys log record.  Some relying parties implement a TOTP replay cache,
in which case one code allows one login.  Other relying parties allow
a single TOTP code to be used for arbitrarily many authentications in
a short time period (generally about a minute).

Fortunately, when recovering from compromise, a user is most interested
in learning whether an attacker has accessed an account zero times
or more than zero times. For \sys-generated credentials, users will
always be able to learn this information from the \sys log.
However, if users import passwords that are not unique into \sys,
this guarantee does not hold. By default, the \sys client software
generates a unique random password for every relying party, but it
also allows user to import existing legacy passwords, which might
not be unique. In the event of password reuse, the attacker can
generate a single log record to obtain the password and then use it
to authenticate to all affected relying parties.

 \section{Logging for FIDO2}
\label{sec:fido2}

\subsection{Background}
\label{sec:fido2:tools}

\paragraph{FIDO2 protocol.}
The FIDO2 protocol~\cite{fido2,webauthn} allows a client to
authenticate using cryptographic keys
stored on a device (e.g., a Yubikey hardware token or a
Google passkey).
To register with a relying party (e.g., \texttt{github.com}),
the client generates an ECDSA keypair, stores the secret key,
and sends the public key to the relying party.
When the client subsequently wants to authenticate to
relying party \texttt{github.com}, Github's server sends the client
a random challenge.
The client then signs the hash of the string \texttt{github.com} 
and the Github-chosen challenge using the secret key the client generated for
\texttt{github.com} at registration.
If the signature is valid, the Github server authorizes the client.
Because the message signed by the client is bound to the name
\texttt{github.com},
FIDO2 provides a strong defense against phishing attacks.
The FIDO2 protocol supports passwordless, second-factor, and multi-factor
authentication.

\paragraph{Zero-knowledge arguments.}
Informally, zero-knowledge arguments allow a prover to convince
a verifier that a statement is true without revealing \emph{why} the
statement is true~\cite{GMR85}.
More precisely, we consider non-interactive zero-knowledge argument systems~\cite{BFM88,FS86} 
in the random-oracle model~\cite{ROM}.
Both the prover and verifier hold the description of a computation~$C$
and a public input~$x$.
The prover's goal is to produce a proof~$\pi$ that convinces 
the verifier that there exists a witness $w$ that causes $C(x,w)=1$,
without revealing the witness~$w$ to the verifier.
We require the standard notions of completeness,
soundness, and zero knowledge~\cite{BFM88, GMR85}.
Throughout the paper, we will refer to this type
of argument system as a ``zero-knowledge proof.''

We use the ZKBoo protocol~\cite{IKOS07,zkboo,zkb++} 
for proving statements about computations expressed as Boolean circuits.
Our system could also be instantiated with
succinct non-interactive arguments of knowledge, which
would decrease proof size and verification time, but at the cost of
increasing proving time and requiring large parameters generated
via a separate setup algorithm 
\cite{BCCT12,GGPR13,Groth16,pinocchio}.

\paragraph{Threshold signatures.}
A two-party threshold signature scheme~\cite{Desmedt87,DF89} is a set of protocols
that allow two parties to jointly generate a single public key
along with two shares of the corresponding secret key and then
jointly sign messages using their secret key shares such that
the signature verifies under the joint public key.
Informally, no malicious party should be able to subvert the protocols to extract
another party's share of the secret key or forge a signature on a
message other than the honest party's message.
We would ideally instantiate our system using BLS multisignatures~\cite{bonehbls}.
Unfortunately, the predominant signing algorithms for FIDO2 are ECDSA
and RSASSA~\cite{RSA-PKCS,fido2-interop,fido2-server-reqs}.
For backwards-compatibility, we present a
construction for two-party ECDSA signing with preprocessing tailored to
our setting in \cref{sec:two-ecdsa}.

\subsection{Split-secret authentication}
\label{sec:fido2:construction}

We now describe our split-secret authentication protocol
for FIDO2 where the authentication secret is split between
the \sys client software and the log service.
The key challenge is achieving log enforcement and log privacy
simultaneously: every successful authentication should result
in a valid log entry encrypting the identity of the relying
party, but the log should not learn the identity of the relying
party.

We use threshold signing to ensure that both the client and
log participate in every successful authentication.  A natural way to
use threshold signing would be to have the client and log each
generate a new threshold signing keypair at every registration.
Unfortunately, if the log service used a different key share for each
relying party, it would know which authentication requests correspond
to the same relying party, violating Goal~2 (privacy against a malicious
log).  Instead, we have the log
use the \emph{same signing-key share for all relying parties}.  The
client still uses a different signing-key share per party, ensuring
the public keys are unlinkable across relying parties.  To
authenticate to a relying party with identifier $\id$ and challenge
$\chal$, the client computes a digest
$\digest = \Hash(\id,\chal)$ that hides $\id$.
The client and log then jointly sign $\digest$.

We also need to ensure that the log service obtains a correct record
of every authentication.  In particular, the log should only
participate in threshold signing if it obtains a valid encryption
$\ct$ of the relying-party identifier $\id$~\cite{Sta96}.

To be valid, a ciphertext $\ct$ must
(1) decrypt to $\id$ under the archive key $k$ established
for that client, and (2) be correctly related to the digest $\digest$ that the log
will sign (i.e., $\Dec(k, \ct) = \id$ and $\digest = \Hash(\id, \chal)$).
To allow the log service to check that the client is using the right archive key
without learning the key, we use a commitment scheme.
During enrollment, the client generates a commitment $\cm$ to the archive key
$k$ using random nonce~$r$ and sends $\cm$ to the log service.
During authentication, the client uses a zero-knowledge proof to prove to the
log that it knows a key $k$, randomness $r$, relying-party identifier $\id$,
and authentication challenge $\chal$ such that
ciphertext $\ct$, digest $\digest$, and commitment $\cm$ from
enrollment meet the following conditions:
\begin{itemize}
    \item[(a)] $\cm = \Commit(k,r)$,
    \item[(b)] $\id = \Dec(k, \ct)$, and
    \item[(c)] $\digest = \Hash(\id, \chal)$.
\end{itemize}
The public inputs are the ciphertext $\ct$, digest $\digest$, and
commitment $\cm$ (known to the client and log); the witness is the
archive key $k$ (known only to the client), commitment
opening $r$, relying-party identifier $\id$, and challenge $\chal$.

\paragraph{Final protocol.}
We now outline our final protocol.

\itpara{Enrollment.}
During enrollment, the client samples a symmetric encryption key $k$ as
the archive key and
commits to it with some random nonce~$r$.
The client sends the commitment $\cm$ to the log, and the log
generates a signing-key share for the user.
The log sends the client the public key corresponding to its signing-key 
share to allow the client to derive future keypairs for relying parties.

\itpara{Registration.}
At registration, the client generates a new signing-key share for
that relying party.
The client then aggregates the log's public key with its new signing-key share
and sends the resulting public key to the relying party.
No interaction with the log service is required.

\itpara{Authentication.}  To authenticate to $\id$ with challenge
$\chal$, the client computes
$\digest\gets\Hash(\id,\chal)$ and $\ct\gets\Enc(k, \id)$. The client then
generates a zero-knowledge proof $\pi$ that it knows 
an archive key $k$, 
commitment nonce $r$, 
relying-party identifier $\id$,
and authentication challenge $\chal$ 
such that $\digest$ and $\ct$ are correctly related relative to
the commitment $\cm$ that the client generated at enrollment.  The client sends
$\digest$, $\ct$, and $\pi$ to the log service.  The log service checks
the proof and, if it verifies, runs its part of the threshold signing
protocol.  The log service stores $\ct$ and returns its signature share
to the client.  The log service also stores the current time and client
IP address with $\ct$, allowing the user to obtain additional metadata
by auditing.  Finally, the client completes the threshold
signature and sends it to the relying party.

\itpara{Auditing.}
To audit the log, the client requests the list of ciphertexts and
metadata from the log service and decrypts all of the relying-party identifiers.

\subsection{Two-party ECDSA with preprocessing}
\label{sec:two-ecdsa}

\cref{sec:fido2:construction} shows how to implement \sys for any
two-of-two threshold signing scheme that cryptographically hashes
input messages.  However, FIDO2 compatibility forces us to use ECDSA,
which is more cumbersome than BLS to threshold.  We present 
a concretely efficient protocol for ECDSA signing between the client and log.

There is a large body of prior work on multi-party ECDSA 
signing~\cite{DKLS18, lindell-ecdsa, dnssec-ecdsa, ecdsa-pcgs, DJN+20,
  CGG+20, GS21, GGN16, GG18, CMP20}.  However, existing protocols are
orders of magnitude more costly than the one we present
here~\cite{lindell-ecdsa, GGN16, GG18, CGG+20, CMP20}.  The efficiency
gain for us comes from the fact that we may assume that the client is
\emph{honest at enrollment time and only later compromised}. 
In contrast, standard schemes for two-party ECDSA signing
must protect against the compromise of either party at any time.
Prior protocols provide this stronger security property at a
computational and communication cost.  In our setting, we need only
ensure that an honest client can run an enrollment procedure with the
log service such that if the client is later compromised, the attacker
cannot subvert the signing protocol.

We leverage the client to split signing into two phases:
\begin{enumerate}
  \item During an \emph{offline phase}, which takes place during enrollment, 
        the client performs some preprocessing to produce a ``presignature.''
        Security only holds if the client is honest during the offline phase.
  \item During an \emph{online phase}, which takes place during authentication, 
        the client and log service use the presignature to perform a lightweight,
        message-dependent computation to produce an ECDSA signature.
        Security holds if either the client or log service is compromised
        during the online phase.
\end{enumerate}
Prior work also splits two-party signing into an offline and
online phase.
However, prior work performs this partitioning to reduce the online time 
at the expense of a more costly offline 
phase~\cite{dnssec-ecdsa, XAXYC21, DJN+20, CGG+20}.
(The offline phase in these schemes is expensive since the protocols
do not assume that both parties are honest during the offline phase.)
We split the signing scheme into an offline and online phase to take advantage
of the fact that we may assume that the client is honest in the offline phase
and so can reduce the total computation time this way.

An additional requirement in our setting is that the log should \emph{not} learn
the public key that the signature is generated under.
Because the public key is specific to a relying party, hiding the public
key is necessary for ensuring that the log cannot distinguish between
relying parties.
The signing algorithm can take
as input a relying-party-specific key share from the client and a
relying-party-independent key share from the log.

\paragraph{Background: ECDSA}.
For a group $\G$ of prime order $q$ with generator $g$,
fixed in the ECDSA standard, an ECDSA secret key is of
the form $\sk \in \Z_q$, where $\Z_q$ denotes the ring of
integers modulo $q$. The corresponding ECDSA public
key is $\pk = g^\sk \in \G$.
ECDSA uses a hash function $\Hash \colon \zo^* \to \Z_q$ 
and a ``conversion'' function $f: \G \to \Z_q$. 
To generate an ECDSA signature on a message $m \in \zo^*$ with secret key $\sk \in \Z_q$,
the signer samples a signing
nonce $r \getsr \Z_q$ and computes 
$$ r^{-1} \cdot (\Hash(m) + f(g^r) \cdot \sk) \quad \in \Z_q.$$

\paragraph{Our construction.}
We now describe our construction for a two-party ECDSA signing protocol
with presignatures.
\iffull
(See \cref{app:two-ecdsa} for technical details.)
\fi
To generate the log keypair, the log samples $x \getsr \Z_q$, sets
its secret key to $x \in \Z_q$, and sets its public key to $X = g^{x} \in \G$.
Then to generate a keypair from the log public key, the client samples
$y \getsr \Z_q$ and sets the relying-party-specific public key to $\pk = X \cdot g^{y} \in \G$.
For each public key of the form $g^{x + y} \in \G$, 
the log has one share $x \in \Z_q$ of the secret key that is the
same for all public keys and the client has the other share $y \in \Z_q$
of the secret key that is different for each public key.

We split the signature-generation process into two parts:
\begin{enumerate}
  \item{}\emph{Offline phase}: a message-independent, key-independent ``presignature'' algorithm
    that the client runs, and
  \item \emph{Online phase}: a message-dependent, key-dependent signing protocol 
      that the log and client run jointly. 
\end{enumerate}
To generate the presignature in the offline phase,
the client samples a signing nonce $r \getsr \Z_q$, computes
$R \gets g^r \in \G$, and splits $r^{-1}$ into additive secret shares: $r^{-1} = r_0 + r_1 \in \Z_q$.
The log's portion of the presignature is $(f(R), r_0) \in \Z_q^2$, and the client's
portion is $(f(R), r_1) \in \Z_q^2$.
Then, to produce a signature on a message in the online phase, the client and
log simply perform
a single secure multiplication to compute
$$r^{-1} \cdot (\Hash(m) + f(R) \cdot \sk) \in \Z_q$$
where $r^{-1} \in \Z_q$ (signing nonce) and $\sk \in \Z_q$ (signing key) are secret-shared between
the client and log.

To perform this multiplication over secret-shared values, we use Beaver triples
\cite{Beaver91}.
A Beaver triple is a set of one-time-use shares of values 
that the log and client can use to efficiently
perform a two-party multiplication on secret-shared values.
Traditionally, generating Beaver triples is one of the expensive portions of
multiparty computation protocols (e.g., in prior work on threshold ECDSA~\cite{dnssec-ecdsa}).
In our setting, the client at enrollment time can
generate a Beaver triple as part of the presignature.
Note that the client and log can use each
signing nonce and Beaver triple exactly once. 
That is, the client and log must use a fresh presignature to generate each signature.

\paragraph{Malicious security.}
By deviating from the protocol, neither the client nor
the log should be able to learn secret information (i.e., the other party's share of the
secret key or signing nonce) or produce a signature for any message
apart from the one that the protocol fixes.
\iffull
We describe how to accomplish this using traditional tools for malicious security
(e.g. information-theoretic MACs~\cite{SPDZ12}) in \cref{app:two-ecdsa}.
\fi

\paragraph{Formalizing and proving security.}
\iffull
We define and prove security in \cref{app:secdefs}
and \cref{app:two-ecdsa}.
\fi

\paragraph{Implications for system design.}
Our preprocessing approach increases the client's work at enrollment:
the client generates some number of presignatures
(e.g., 10K) and sends the log's presignature shares to the log.
To reduce storage burden on the log, the client can store
encryptions of the log's presignature shares.

When the client is close to running out of presignatures, it can
authenticate with the log, generate more presignatures, and send the log's
presignature shares to the log service. If the log service does not receive
an objection after some period of time, it will start using the new
presignatures. An honest client periodically checks the log to see
whether any unexpected presignatures (created by an attacker) appear
in its log. If the client learns that a new batch of presignatures
was generated that the client did not authorize, the client authenticates
to the log service and objects.
This approach provides security as long as an honest client can detect client compromise and object to any adversarially generated presignatures during the objection period. (If a user is concerned about recovering authentication logs from a long period of undetected compromise, the client can set a very long objection period, or generate enough presignatures at enrollment that, with high probability, it will not need to generate more.)

If the client runs out of presignatures and the log service rejects
the client's presignatures, the client and the log can temporarily use a more
expensive signing protocol that does not require
presignatures~\cite{GGN16,lindell-ecdsa,DKLS18,XAXYC21}.
The client could run out of presignatures and be forced to use
the slow multisignature protocol in the following cases:
\begin{enumerate}
    \item The attacker compromised the user's credentials with
        the log service, allowing the attacker to object to the 
        new presignatures. In this case, the attacker could change
        the user's credentials and permanently lock the user
        out of her account.
    \item The honest client was close to running out of presignatures,
        generated new presignatures, and then ran out of presignatures
        while waiting for a possible objection. This scenario only
        occurs when the honest client makes an unexpectedly large
        number of authentications in a short period of time.
        The client only needs to pay the cost of the slow multisignature
        protocol for a short period of time.
\end{enumerate}
An attacker that has compromised the log service can also
deny service, as we discussed in \cref{sec:overview:non}.

\paragraph{Benefits of future support for Schnorr-based signing.}
The FIDO2 standard recommends support for
EdDSA, which, if widely supported in the future, could simplify the
two-party signing protocol and avoid preprocessing altogether.
Adapting a Schnorr-based threshold signing protocol~\cite{CoSI,musig,musig2}
to the setting where only the client knows the message and public key could
potentially improve performance.

 \section{Logging for time-based one-time passwords}\label{sec:totp}
We now show how \sys can support time-based one-time passwords
(TOTP).

\subsection{Background: TOTP}
TOTP is a popular form of
second-factor authentication that authenticator apps 
(Authy, Google Authenticator, and others~\cite{totp-rfc}) implement.
When a client registers for TOTP with a relying party, the
relying party sends the client a secret cryptographic key.
Then, to authenticate, the client and the relying
party both compute a MAC on the current 
time using the secret key from registration.
The client sends the resulting MAC tag to the server.
If the client's submitted tag matches the one
that the server computes, the relying party authorizes the
client.
TOTP uses a hash-based MAC (HMAC).

\subsection{Split-secret authentication for TOTP}
At a high level, in our split-secret authentication protocol for TOTP,
both the client and log service 
have as private input additive secret shares of the
TOTP secret key.
At the conclusion of the split-secret authentication, 
the client holds a TOTP code
and the log service 
holds a ciphertext.
We now give the details of our protocol.

\itpara{Enrollment.}
At enrollment, just as with FIDO2, the client generates and stores a
long-term symmetric-encryption archive key $k$ and random nonce $r$.
Then, the client 
sends the commitment $\cm = \Commit(k, r)$ to the log service.

\itpara{Registration.}
To register a client, a relying party generates and sends the client
a secret MAC key $\kid$ for TOTP\@.
The client samples a random identifier $\id$ for
the relying party and then splits the TOTP secret key $\kid$ into
additive secret shares $\klog_\id$ and $\kclient_\id$.
The client sends $(\id,\klog_\id)$ to the log service and locally stores
$(\id,\kclient_\id)$ alongside a name identifying the relying party (e.g.,
\texttt{user@amazon.com}).

\itpara{Authentication.}
In order to authenticate to the relying party $\id$ at time $t$, the
client needs to compute $\HMAC(\kid, t)$ with the help of the log
service.
Let $n$ be the number of relying parties with which the client
has registered.
To authenticate, the client and log service run a secure two-party computation where:
\begin{itemize}
  \item The client's input is its long-term symmetric archive key $k$
        and commitment opening $r$ from enrollment, the
        relying-party identifier $\id$, and the client's
        share of the TOTP key $\kclient_\id$.
  \item The log service's input is the commitment
        $\cm$ from enrollment, the list of
        relying-party identifiers that the client has registered
        with $(\id_1, \dots, \id_n)$, and the log service's TOTP key 
        shares $(\klog_{\id_1}, \dots, \klog_{\id_n})$---one
        per relying party. 
    \item The client outputs the TOTP code $\HMAC(\kid, t)$.
    \item The log outputs an encrypted log record:
        an encryption of the relying-party identifier $\id$ under 
        the archive key~$k$.
\end{itemize}
We execute this two-party computation using an
off-the-shelf garbled-circuit-based
multiparty computation protocol. 
Garbled circuits allow two parties to jointly execute any Boolean circuit 
on private inputs, where neither party learns 
information about the other's input beyond what they can
infer from the circuit's output~\cite{Yao82}.
We use the protocol from Wang et al.~\cite{WRK17}, which provides malicious
security, meaning that the protocol remains secure even if one corrupted
party deviates arbitrarily from the protocol.
As long as either the client or the log service is honest, the log service
does not learn any information about the client's authentication
secrets, and the client
learn no information about the TOTP secret, apart from
the single TOTP code that the protocol outputs.
Because we use an off-the-shelf garbled-circuit protocol,
the communication overhead is much higher than in the
special-purpose protocols we design for FIDO2 and passwords (\cref{sec:eval}).
TOTP is challenging to design a special-purpose protocol
for because the authentication credential must be generated via the
SHA hash function which, unlike the authentication credentials
for FIDO2 and passwords, does not have structure we can exploit.
Clients can ask the log service to delete registrations for
unused accounts to speed up the two-party computation.

\itpara{Auditing.}
To audit the log, the client simply
requests the list of ciphertexts from the log service.
The client decrypts each ciphertext with its archive key~$k$
and then, using its mapping of $\id$ values
to relying party names, outputs the
resulting list of relying party names.

 \section{Logging for passwords}\label{sec:pw}

We now describe how \sys can support passwords.

\subsection{Protocol overview}
We construct a split-secret authentication protocol that takes 
place between the client and the log service.
In particular, we show how the client can compute the
password to authenticate to a relying party in such 
a way that
(a) the log service does not learn the relying party's
    identity and
(b) the client's authentication attempt is logged.
At the start of the authentication protocol run:
\begin{itemize}
  \item the client holds a secret key, the log service's
        public key, and 
        the identity $\id^*$ of the relying party it wants to
        authenticate to, and 
  \item the log service holds its own secret key, 
        the client's public key, and 
        a list of relying-party identities 
        $(\id_1, \dots, \id_n)$ 
        at which the client has registered.
\end{itemize}
At the end of the authentication protocol run:
\begin{itemize}
  \item the client holds a password derived as 
        a pseudorandom function of the client's secret,
        the log's secret, and the relying party identity $\id^*$, and 
  \item the log service holds a ciphertext encrypting
        the relying party's identity $\id^*$ under
        the client's public key.
\end{itemize}

\paragraph{Limitations inherent to passwords.}
As we discussed in \cref{sec:overview:non},
\sys for passwords does not protect against
credential breaches, but does defend against
device compromise.

\subsection{Split-secret authentication for passwords}

The \sys scheme for password-based authentication
uses a cyclic group $\G$ of prime order $q$ with a fixed
generator $g\in \G$.
Our implementation uses the NIST P-256 elliptic-curve group.

When using password-based authentication in \sys,
the client and log service after registration each hold a
secret share of the password for each relying party.
In particular, the password for a relying party 
with identity $\id \in \zo^*$ is the string
$\mathsf{pw}_\id = k_\id \cdot \Hash(\id)^k \in \G$,
where:
\begin{itemize}
\item $k_\id \in \Z_q$ is a per-relying-party secret share held by the
  client,
\item $\Hash\colon \zo^* \to \G$ is a hash function, and
\item $k \in \Z_q$ is a per-client secret key held by the log service.
\end{itemize}
Thus, computing $\mathsf{pw}_\id$ requires both
the client's per-site key~$k_\id$ and the
log's secret key $k$.

The technical challenge is to construct a protocol
that allows the client to compute the password 
$\mathsf{pw}_\id$ while (a)~hiding $\id$ from the log service and
(b)~ensuring that the log service completes the interaction
holding an encryption of $\id$ under the client's public key.

\paragraph{Protocol.}
We describe the protocol steps:

\itpara{Enrollment.}  The client samples an ElGamal secret key $x \in
\Z_q$ as the archive key and sends the corresponding public key $X =
g^x \in \G$ to the log service.  The log service samples a
Diffie-Hellman secret key $k \in \Z_q$ and sends its public key $K =
g^k \in \G$ to the client.

\itpara{Registration.}  The client samples a per-relying-party random
identifier $\id \getsr \zo^{128}$, saves $\id$ locally alongside the
name of the relying party (e.g., \texttt{user@amazon.com}), and sends
$\id$ to the log service.
The log service saves the string $\Hash(\id)$ and replies with
$\Hash(\id)^k \in \G$.  To generate a new strong password
$\mathsf{pw}_\id$ (the recommended use), the client samples and saves
a random key share $k_\id \getsr \G$ and sets $\mathsf{pw}_\id \gets
k_\id \cdot \Hash(\id)^k \in \G$.  To import a legacy password
$\mathsf{pw}_\id$ (less secure), the client computes and stores
$k_\id\gets\mathsf{pw}_\id \cdot
\left(\Hash(\id)^k)\right)^{-1}\in\G$.  The client then deletes
$\Hash(\id)^k$ and $\pw_\id$. Note that the log server can discard
$\id$, which it only uses to avoid providing $h^k$ for arbitrary $h$.
When the client samples $\id$ and $k_\id$ randomly in the recommended
usage, the password $\pw_\id$ for each relying party is random
and distinct.

\itpara{Authentication.}  During authentication, the client must
recompute the password $\pw_\id$.  To do so, the client first sends
the log service an encryption of $\Hash(\id)$ under the public
ElGamal archive key $g^x$:  the client samples $r \getsr \Z_q^*$ and
computes the ciphertext $(c_1, c_2) = (g^r, \Hash(\id) \cdot g^{xr})
\in \G^2$.  In addition, the client sends a zero-knowledge proof to
the log service attesting to the fact that $(c_1, c_2)$ is an
encryption under the client's public key $X$ of $\Hash(\id)$ for $\id
\in \{\id_1, \id_2, \dots, \id_n\}$---the set of relying-party
identifiers that the client sent to the log service during each of its
registrations so far.  The client executes this proof using the
technique from Groth and Kohlweiss~\cite{GK14}.  The proof size is
$O(\log n)$ and the prover and verifier time are both $O(n)$.
\iffull
(See \cref{app:pw} for implementation details.)
\fi

The log service saves the ciphertext as a log entry, checks the zero-knowledge proof,
and returns the value $h = c_2^k = \Hash(\id)^k \cdot g^{xrk} \in \G$ to the client.
The client can then compute 
\[ \pw_\id = k_\id \cdot h \cdot K^{-xr} = k_\id \cdot \Hash(\id)^k \quad \in \G. \]
Crucially, the client deletes $\pw_\id$ after authentication to
ensure that future authentications must again interact with the log
service.

\itpara{Auditing.}
To audit the log, the client downloads the ElGamal ciphertexts
and can decrypt each ciphertext to recover
a list of hashed identities: $(\Hash(\id_1), \Hash(\id_2), \dots)$.
The client uses its stored mapping of $\id$s to
relying-party identifiers to recover the plaintext names of 
the relying parties in the log.

 \section{Protecting against log misbehavior}
\label{sec:mal}

The \sys log service must participate in each of the user's
authentication attempts.
If the log service goes offline, the user will not be able
to authenticate to any of her \sys-enabled relying parties.
In a real-world deployment, the log service could consist of
multiple servers replicated using standard state-machine
replication techniques to tolerate benign failures~\cite{paxos,raft}.
However, users might also worry about intentional denial-of-service
attacks on the part of the log.

To defend against availability attacks, a user can split trust
across multiple logs.  At enrollment time, the user can enroll with
$n$ logs.  Then at registration, the user can set a threshold $t$ of
logs that must participate in authentication.  Thus, the user can
authenticate to her accounts so long as $t$ logs are online, and she
can audit activity so long as $n-t+1$ logs are available.
We need $n-t+1$ logs to be available for auditing in order to guarantee
that at least one of the $t$ logs that participated in authentication
is online.
To ensure that colluding logs cannot authenticate on behalf of
a client, the user's client can run $n+1$ logical parties, and
$n+t+1$ parties can generate an authentication credential.
In the setting with multiple log services, we need to adapt our
two-party protocols to threshold multi-party protocols.
Although we present our techniques for two parties (the client and
a single log), our techniques generalize to multiple parties in a
straightforward way.

For FIDO2 and passwords, the client now sends
a zero-knowledge proof to each of the $n$ logs.
In the password case, the client can then retrieve $(t,n)$
Shamir shares of the password~\cite{S79}, and in the FIDO2
case, the client can run any existing multi-party threshold
signing protocol that does not take the public key as
input~\cite{ST19,dnssec-ecdsa}.
For TOTP, the client and the $n$ logs can execute the same
circuit using any malicious-secure threshold multi-party
computation protocol~\cite{BGW88}.

Note that for relying parties that support FIDO2,
users can optionally register a backup hardware FIDO2 device to allow 
them to bypass the log. In this case, the user can authenticate
either via \sys or via her backup FIDO2 key. While registering
a backup hardware device protects against availability attacks, if an attacker
obtains this hardware device, they can authenticate as the user without
interacting with the log.

 \section{Implementation}
We implemented \sys for FIDO2, TOTP, and passwords with a single log
service.
We use C/C++ with gRPC and OpenSSL with the P256 curve
(required by the FIDO2 standard).
We wrote approximately 5,700 lines of C/C++ and 50 lines of
Javascript (excluding tests and benchmarks).
Our implementation is available at \url{https://github.com/edauterman/larch}.

For our FIDO2 implementation, we implemented a ZKBoo~\cite{zkboo} library
for arbitrary Boolean circuits.
Our ZKBoo implementation (with optimizations from ZKB++~\cite{zkb++}) uses
emp-toolkit to support arbitrary Boolean circuits in Bristol Fashion~\cite{emp-toolkit}.
To support the parallel repetitions required for soundness error $<2^{-80}$, we use
SIMD instructions with a bitwidth of 32 and run 5 threads in parallel.
For the proof circuit, we use AES in counter mode for encryption and SHA-256 for
commitments (SHA-256 is necessary for backwards compatability with FIDO2).
We built a log service and client that invoke the ZKBoo library,
as well as a Chrome browser extension that interfaces with our
client application and is compatible with existing FIDO2 relying parties.
We built our browser extension on top of an existing extension~\cite{kr-u2f}.

Our TOTP implementation uses a maliciously secure garbled-circuit
construction~\cite{WRK17} implemented in emp-toolkit~\cite{emp-toolkit}.
We generated our circuit using the CBMC-GC compiler~\cite{CBMC-GC}
with ChaCha20 for encryption and SHA-256 for commitments.

For our passwords implementation, we implemented Groth and 
Kohlweiss's proof system~\cite{GK14}.

Our implementation uses a single log server for the log service, does not
encrypt communication between the client and the log service, and does not
require the client to authenticate with the log service. A real-world deployment
would use multiple servers for replication, use TLS between the client and
the log service, and authenticate the client before performing
any operations.

\paragraph{Optimizations.}
We use pseudorandom generators (PRGs) to compress presignatures:
the log stores 6 elements in $\Z_q$ and the client stores 1 element.
Also, instead of running an authenticated encryption scheme (e.g. AES-GCM)
inside the circuit for FIDO2 or TOTP, we run an encryption scheme without authentication
(e.g. AES in counter mode) inside the circuit and then sign the ciphertext
(client has the signing key, log has the verification key).
The log can check the integrity
of the ciphertext by verifying the signature, which is much faster
than checking in a zero-knowledge proof or computing
the ciphertext tag jointly in a two-party computation.
 \section{Evaluation}
\label{sec:eval}

In this section, we evaluate the cost of \sys to
end users and the cost of running a \sys log service.

\paragraph{Experiment setup.}
We run our benchmarks on Amazon AWS EC2 instances.
Unless otherwise specified, we run the log service on a c5.4xlarge
instance with 8 cores and 32GiB of memory
and, for latency benchmarks, the client on a c5.2xlarge instance with 4 cores
and 16GiB of memory, comparable to a commodity laptop (2 hyperthreads per core
throughout).
We configure the network connection between the client and log
service to have a 20ms RTT and a bandwidth of 100 Mbps.

\subsection{End-user cost}
\label{sec:eval:fido2}

We show \sys authentication latency and communication costs for FIDO2,
TOTP, and passwords.

\subsubsection{FIDO2}

\paragraph{Latency.}
The client for our FIDO2 scheme can complete authentication in 
303ms with a single CPU core, or 117ms when using eight cores (\cref{fig:latency}).
Loading a webpage often takes a few seconds because of network latency,
so the client cost of \sys authentication is minor by comparison.
The client's running time during authentication 
is independent of the number of relying parties.
The heaviest part of the client's computation is proving to the log
service that its encrypted log entry is well formed.

At enrollment, the client must generate many ``presignatures,'' 
which it later uses to run our authentication protocol with the log.
Generating 10,000 presignatures for 10,000 future FIDO2 authentications takes 885ms.
When the client runs out of presignatures, it generates new presignatures
it can use after a waiting period (see \cref{sec:two-ecdsa}).

\paragraph{Communication.}
During enrollment, the client must send the log 1.8MiB worth
of presignatures.
Thereafter, each authentication attempt requires 1.73MiB worth
of communication: the bulk of this consists of the client's zero-knowledge
proof of correctness, and 352B of it comes from the signature protocol.
By using a different zero-knowledge proof system, we could 
reduce communication cost at the expense of increasing client computation cost.

\paragraph{Comparison to existing two-party ECDSA.}
For comparison, a state-of-the-art two-party ECDSA protocol~\cite{XAXYC21}
that does not require presignatures from the client and uses Paillier requires
226ms of computation at signing time (the authors' measurements exclude network
latency, which we estimate would add 60ms) and 6.3KiB of per-signature communication.
Using a variant of the same protocol based on oblivious transfer~\cite{XAXYC21}
makes it possible to reduce computation
to 2.8ms (again, excluding network latency, which we estimate would add 60ms)
at the cost of increasing per-signature communication to 90.9KiB.
In contrast, our signing protocol only requires 0.5KiB per-signature communication (including
the log presignature and the signing messages) and takes 61ms time at signing, almost all 
of which is due to network latency and can be run in parallel with proving and
verifying as the computational overhead is minimal (roughly 1ms).

\subsubsection{TOTP}

\paragraph{Latency.}
In \cref{fig:latency} (right), we show how TOTP authentication latency
increases with the number of relying parties the user registered with.
Because we implement TOTP authentication using garbled
circuits~\cite{WRK17}, we can split authentication into 
two phases: an ``offline'', input-independent phase and an ``online'',
input-dependent phase (the log service and client communicate
in both phases).
Both phases are performed once per authentication.
However, the offline phase can be performed in advance of when the user
needs to authenticate to their account, and so it does not affect
the latency that the user experiences.
For 20 relying parties, the online time is 91ms and the offline time
is 1.23s.
For 100 relying parties, the online time is 120ms and the offline time
is 1.39s.

\paragraph{Communication.}
Communication costs for our TOTP authentication scheme are large:
for 20 relying parties, the total communication cost is
65MiB, and for 100 relying parties, the total communication is 93MiB.
The online communication costs are much smaller: for 20 relying
parties, the online communication is 202KiB and for 100 relying parties,
the online communication is 908KiB. 
We envision clients running the offline phase in the background while
they have good connectivity.
While these communication costs are much higher than those associated with
FIDO2 or passwords, we expect users to authenticate with TOTP less
frequently because TOTP is only used for second-factor authentication.

\subsubsection{Passwords}

\paragraph{Latency.}
In \cref{fig:latency} (center), we show how password authentication latency
increases with the number of registered relying parties.
With 16 relying parties, authentication takes 28ms, and
with 512 relying parties, it takes 245ms:
the authentication time grows linearly with the number of relying parties.
The proof system we use requires padding the number of relying parties
to the nearest power of two, meaning that registering at additional
relying parties does not affect the latency or communication until the
number of relying parties reaches the next power of two.

\paragraph{Communication.}
In \cref{fig:pw-comm}, we show how communication increases logarithmically
with the number of relying parties.
This behavior is due to the fact that proof size is logarithmic in the
number of relying parties.
With 16 relying parties, the communication is 1.47KiB, and with 512 relying
parties, it is 4.14KiB.

\begin{figure*}[t]
    \centering
    \begingroup \makeatletter \begin{pgfpicture}\pgfpathrectangle{\pgfpointorigin}{\pgfqpoint{3.067041in}{0.180556in}}\pgfusepath{use as bounding box, clip}\begin{pgfscope}\pgfsetbuttcap \pgfsetmiterjoin \definecolor{currentfill}{rgb}{1.000000,1.000000,1.000000}\pgfsetfillcolor{currentfill}\pgfsetlinewidth{0.000000pt}\definecolor{currentstroke}{rgb}{1.000000,1.000000,1.000000}\pgfsetstrokecolor{currentstroke}\pgfsetdash{}{0pt}\pgfpathmoveto{\pgfqpoint{0.000000in}{0.000000in}}\pgfpathlineto{\pgfqpoint{3.067041in}{0.000000in}}\pgfpathlineto{\pgfqpoint{3.067041in}{0.180556in}}\pgfpathlineto{\pgfqpoint{0.000000in}{0.180556in}}\pgfpathlineto{\pgfqpoint{0.000000in}{0.000000in}}\pgfpathclose \pgfusepath{fill}\end{pgfscope}\begin{pgfscope}\pgfsetbuttcap \pgfsetmiterjoin \definecolor{currentfill}{rgb}{1.000000,1.000000,1.000000}\pgfsetfillcolor{currentfill}\pgfsetfillopacity{0.800000}\pgfsetlinewidth{1.003750pt}\definecolor{currentstroke}{rgb}{0.800000,0.800000,0.800000}\pgfsetstrokecolor{currentstroke}\pgfsetstrokeopacity{0.800000}\pgfsetdash{}{0pt}\pgfpathmoveto{\pgfqpoint{0.020833in}{0.000000in}}\pgfpathlineto{\pgfqpoint{3.046208in}{0.000000in}}\pgfpathquadraticcurveto{\pgfqpoint{3.067041in}{0.000000in}}{\pgfqpoint{3.067041in}{0.020833in}}\pgfpathlineto{\pgfqpoint{3.067041in}{0.159722in}}\pgfpathquadraticcurveto{\pgfqpoint{3.067041in}{0.180556in}}{\pgfqpoint{3.046208in}{0.180556in}}\pgfpathlineto{\pgfqpoint{0.020833in}{0.180556in}}\pgfpathquadraticcurveto{\pgfqpoint{0.000000in}{0.180556in}}{\pgfqpoint{0.000000in}{0.159722in}}\pgfpathlineto{\pgfqpoint{0.000000in}{0.020833in}}\pgfpathquadraticcurveto{\pgfqpoint{0.000000in}{0.000000in}}{\pgfqpoint{0.020833in}{0.000000in}}\pgfpathlineto{\pgfqpoint{0.020833in}{0.000000in}}\pgfpathclose \pgfusepath{stroke,fill}\end{pgfscope}\begin{pgfscope}\pgfsetbuttcap \pgfsetmiterjoin \definecolor{currentfill}{rgb}{0.596078,0.305882,0.639216}\pgfsetfillcolor{currentfill}\pgfsetlinewidth{1.003750pt}\definecolor{currentstroke}{rgb}{0.596078,0.305882,0.639216}\pgfsetstrokecolor{currentstroke}\pgfsetdash{}{0pt}\pgfpathmoveto{\pgfqpoint{0.041667in}{0.065972in}}\pgfpathlineto{\pgfqpoint{0.250000in}{0.065972in}}\pgfpathlineto{\pgfqpoint{0.250000in}{0.138889in}}\pgfpathlineto{\pgfqpoint{0.041667in}{0.138889in}}\pgfpathlineto{\pgfqpoint{0.041667in}{0.065972in}}\pgfpathclose \pgfusepath{stroke,fill}\end{pgfscope}\begin{pgfscope}\definecolor{textcolor}{rgb}{0.000000,0.000000,0.000000}\pgfsetstrokecolor{textcolor}\pgfsetfillcolor{textcolor}\pgftext[x=0.333333in,y=0.065972in,left,base]{\color{textcolor}\rmfamily\fontsize{7.500000}{9.000000}\selectfont Verify (Server)}\end{pgfscope}\begin{pgfscope}\pgfsetbuttcap \pgfsetmiterjoin \definecolor{currentfill}{rgb}{1.000000,0.498039,0.000000}\pgfsetfillcolor{currentfill}\pgfsetlinewidth{1.003750pt}\definecolor{currentstroke}{rgb}{1.000000,0.498039,0.000000}\pgfsetstrokecolor{currentstroke}\pgfsetdash{}{0pt}\pgfpathmoveto{\pgfqpoint{1.258202in}{0.065972in}}\pgfpathlineto{\pgfqpoint{1.466535in}{0.065972in}}\pgfpathlineto{\pgfqpoint{1.466535in}{0.138889in}}\pgfpathlineto{\pgfqpoint{1.258202in}{0.138889in}}\pgfpathlineto{\pgfqpoint{1.258202in}{0.065972in}}\pgfpathclose \pgfusepath{stroke,fill}\end{pgfscope}\begin{pgfscope}\definecolor{textcolor}{rgb}{0.000000,0.000000,0.000000}\pgfsetstrokecolor{textcolor}\pgfsetfillcolor{textcolor}\pgftext[x=1.549869in,y=0.065972in,left,base]{\color{textcolor}\rmfamily\fontsize{7.500000}{9.000000}\selectfont Other}\end{pgfscope}\begin{pgfscope}\pgfsetbuttcap \pgfsetmiterjoin \definecolor{currentfill}{rgb}{0.215686,0.494118,0.721569}\pgfsetfillcolor{currentfill}\pgfsetlinewidth{1.003750pt}\definecolor{currentstroke}{rgb}{0.215686,0.494118,0.721569}\pgfsetstrokecolor{currentstroke}\pgfsetdash{}{0pt}\pgfpathmoveto{\pgfqpoint{2.040997in}{0.065972in}}\pgfpathlineto{\pgfqpoint{2.249330in}{0.065972in}}\pgfpathlineto{\pgfqpoint{2.249330in}{0.138889in}}\pgfpathlineto{\pgfqpoint{2.040997in}{0.138889in}}\pgfpathlineto{\pgfqpoint{2.040997in}{0.065972in}}\pgfpathclose \pgfusepath{stroke,fill}\end{pgfscope}\begin{pgfscope}\definecolor{textcolor}{rgb}{0.000000,0.000000,0.000000}\pgfsetstrokecolor{textcolor}\pgfsetfillcolor{textcolor}\pgftext[x=2.332663in,y=0.065972in,left,base]{\color{textcolor}\rmfamily\fontsize{7.500000}{9.000000}\selectfont Prove (Client)}\end{pgfscope}\end{pgfpicture}\makeatother \endgroup  ~~~~~~~~~~~~~~~~~~~
    \begingroup \makeatletter \begin{pgfpicture}\pgfpathrectangle{\pgfpointorigin}{\pgfqpoint{1.516978in}{0.175154in}}\pgfusepath{use as bounding box, clip}\begin{pgfscope}\pgfsetbuttcap \pgfsetmiterjoin \definecolor{currentfill}{rgb}{1.000000,1.000000,1.000000}\pgfsetfillcolor{currentfill}\pgfsetlinewidth{0.000000pt}\definecolor{currentstroke}{rgb}{1.000000,1.000000,1.000000}\pgfsetstrokecolor{currentstroke}\pgfsetdash{}{0pt}\pgfpathmoveto{\pgfqpoint{0.000000in}{0.000000in}}\pgfpathlineto{\pgfqpoint{1.516978in}{0.000000in}}\pgfpathlineto{\pgfqpoint{1.516978in}{0.175154in}}\pgfpathlineto{\pgfqpoint{0.000000in}{0.175154in}}\pgfpathlineto{\pgfqpoint{0.000000in}{0.000000in}}\pgfpathclose \pgfusepath{fill}\end{pgfscope}\begin{pgfscope}\pgfsetbuttcap \pgfsetmiterjoin \definecolor{currentfill}{rgb}{1.000000,1.000000,1.000000}\pgfsetfillcolor{currentfill}\pgfsetfillopacity{0.800000}\pgfsetlinewidth{1.003750pt}\definecolor{currentstroke}{rgb}{0.800000,0.800000,0.800000}\pgfsetstrokecolor{currentstroke}\pgfsetstrokeopacity{0.800000}\pgfsetdash{}{0pt}\pgfpathmoveto{\pgfqpoint{0.020833in}{0.000000in}}\pgfpathlineto{\pgfqpoint{1.496145in}{0.000000in}}\pgfpathquadraticcurveto{\pgfqpoint{1.516978in}{0.000000in}}{\pgfqpoint{1.516978in}{0.020833in}}\pgfpathlineto{\pgfqpoint{1.516978in}{0.154321in}}\pgfpathquadraticcurveto{\pgfqpoint{1.516978in}{0.175154in}}{\pgfqpoint{1.496145in}{0.175154in}}\pgfpathlineto{\pgfqpoint{0.020833in}{0.175154in}}\pgfpathquadraticcurveto{\pgfqpoint{0.000000in}{0.175154in}}{\pgfqpoint{0.000000in}{0.154321in}}\pgfpathlineto{\pgfqpoint{0.000000in}{0.020833in}}\pgfpathquadraticcurveto{\pgfqpoint{0.000000in}{0.000000in}}{\pgfqpoint{0.020833in}{0.000000in}}\pgfpathlineto{\pgfqpoint{0.020833in}{0.000000in}}\pgfpathclose \pgfusepath{stroke,fill}\end{pgfscope}\begin{pgfscope}\pgfsetbuttcap \pgfsetmiterjoin \definecolor{currentfill}{rgb}{0.301961,0.686275,0.290196}\pgfsetfillcolor{currentfill}\pgfsetlinewidth{1.003750pt}\definecolor{currentstroke}{rgb}{0.301961,0.686275,0.290196}\pgfsetstrokecolor{currentstroke}\pgfsetdash{}{0pt}\pgfpathmoveto{\pgfqpoint{0.041667in}{0.060571in}}\pgfpathlineto{\pgfqpoint{0.250000in}{0.060571in}}\pgfpathlineto{\pgfqpoint{0.250000in}{0.133488in}}\pgfpathlineto{\pgfqpoint{0.041667in}{0.133488in}}\pgfpathlineto{\pgfqpoint{0.041667in}{0.060571in}}\pgfpathclose \pgfusepath{stroke,fill}\end{pgfscope}\begin{pgfscope}\definecolor{textcolor}{rgb}{0.000000,0.000000,0.000000}\pgfsetstrokecolor{textcolor}\pgfsetfillcolor{textcolor}\pgftext[x=0.333333in,y=0.060571in,left,base]{\color{textcolor}\rmfamily\fontsize{7.500000}{9.000000}\selectfont Offline}\end{pgfscope}\begin{pgfscope}\pgfsetbuttcap \pgfsetmiterjoin \definecolor{currentfill}{rgb}{0.901961,0.670588,0.007843}\pgfsetfillcolor{currentfill}\pgfsetlinewidth{1.003750pt}\definecolor{currentstroke}{rgb}{0.901961,0.670588,0.007843}\pgfsetstrokecolor{currentstroke}\pgfsetdash{}{0pt}\pgfpathmoveto{\pgfqpoint{0.863427in}{0.060571in}}\pgfpathlineto{\pgfqpoint{1.071761in}{0.060571in}}\pgfpathlineto{\pgfqpoint{1.071761in}{0.133488in}}\pgfpathlineto{\pgfqpoint{0.863427in}{0.133488in}}\pgfpathlineto{\pgfqpoint{0.863427in}{0.060571in}}\pgfpathclose \pgfusepath{stroke,fill}\end{pgfscope}\begin{pgfscope}\definecolor{textcolor}{rgb}{0.000000,0.000000,0.000000}\pgfsetstrokecolor{textcolor}\pgfsetfillcolor{textcolor}\pgftext[x=1.155094in,y=0.060571in,left,base]{\color{textcolor}\rmfamily\fontsize{7.500000}{9.000000}\selectfont Online}\end{pgfscope}\end{pgfpicture}\makeatother \endgroup  \\
    \begingroup \makeatletter \begin{pgfpicture}\pgfpathrectangle{\pgfpointorigin}{\pgfqpoint{2.100000in}{1.300000in}}\pgfusepath{use as bounding box, clip}\begin{pgfscope}\pgfsetbuttcap \pgfsetmiterjoin \definecolor{currentfill}{rgb}{1.000000,1.000000,1.000000}\pgfsetfillcolor{currentfill}\pgfsetlinewidth{0.000000pt}\definecolor{currentstroke}{rgb}{1.000000,1.000000,1.000000}\pgfsetstrokecolor{currentstroke}\pgfsetdash{}{0pt}\pgfpathmoveto{\pgfqpoint{0.000000in}{0.000000in}}\pgfpathlineto{\pgfqpoint{2.100000in}{0.000000in}}\pgfpathlineto{\pgfqpoint{2.100000in}{1.300000in}}\pgfpathlineto{\pgfqpoint{0.000000in}{1.300000in}}\pgfpathlineto{\pgfqpoint{0.000000in}{0.000000in}}\pgfpathclose \pgfusepath{fill}\end{pgfscope}\begin{pgfscope}\pgfsetbuttcap \pgfsetmiterjoin \definecolor{currentfill}{rgb}{1.000000,1.000000,1.000000}\pgfsetfillcolor{currentfill}\pgfsetlinewidth{0.000000pt}\definecolor{currentstroke}{rgb}{0.000000,0.000000,0.000000}\pgfsetstrokecolor{currentstroke}\pgfsetstrokeopacity{0.000000}\pgfsetdash{}{0pt}\pgfpathmoveto{\pgfqpoint{0.392364in}{0.464507in}}\pgfpathlineto{\pgfqpoint{2.100000in}{0.464507in}}\pgfpathlineto{\pgfqpoint{2.100000in}{1.300000in}}\pgfpathlineto{\pgfqpoint{0.392364in}{1.300000in}}\pgfpathlineto{\pgfqpoint{0.392364in}{0.464507in}}\pgfpathclose \pgfusepath{fill}\end{pgfscope}\begin{pgfscope}\pgfpathrectangle{\pgfqpoint{0.392364in}{0.464507in}}{\pgfqpoint{1.707636in}{0.835493in}}\pgfusepath{clip}\pgfsetbuttcap \pgfsetroundjoin \definecolor{currentfill}{rgb}{0.596078,0.305882,0.639216}\pgfsetfillcolor{currentfill}\pgfsetlinewidth{0.000000pt}\definecolor{currentstroke}{rgb}{0.000000,0.000000,0.000000}\pgfsetstrokecolor{currentstroke}\pgfsetdash{}{0pt}\pgfpathmoveto{\pgfqpoint{0.469983in}{0.556420in}}\pgfpathlineto{\pgfqpoint{0.469983in}{0.464507in}}\pgfpathlineto{\pgfqpoint{0.691754in}{0.464507in}}\pgfpathlineto{\pgfqpoint{1.135296in}{0.464507in}}\pgfpathlineto{\pgfqpoint{2.022380in}{0.464507in}}\pgfpathlineto{\pgfqpoint{2.022380in}{0.556420in}}\pgfpathlineto{\pgfqpoint{2.022380in}{0.556420in}}\pgfpathlineto{\pgfqpoint{1.135296in}{0.556420in}}\pgfpathlineto{\pgfqpoint{0.691754in}{0.556420in}}\pgfpathlineto{\pgfqpoint{0.469983in}{0.556420in}}\pgfpathlineto{\pgfqpoint{0.469983in}{0.556420in}}\pgfpathclose \pgfusepath{fill}\end{pgfscope}\begin{pgfscope}\pgfpathrectangle{\pgfqpoint{0.392364in}{0.464507in}}{\pgfqpoint{1.707636in}{0.835493in}}\pgfusepath{clip}\pgfsetbuttcap \pgfsetroundjoin \definecolor{currentfill}{rgb}{1.000000,0.498039,0.000000}\pgfsetfillcolor{currentfill}\pgfsetlinewidth{0.000000pt}\definecolor{currentstroke}{rgb}{0.000000,0.000000,0.000000}\pgfsetstrokecolor{currentstroke}\pgfsetdash{}{0pt}\pgfpathmoveto{\pgfqpoint{0.469983in}{0.669342in}}\pgfpathlineto{\pgfqpoint{0.469983in}{0.556420in}}\pgfpathlineto{\pgfqpoint{0.691754in}{0.556420in}}\pgfpathlineto{\pgfqpoint{1.135296in}{0.556420in}}\pgfpathlineto{\pgfqpoint{2.022380in}{0.556420in}}\pgfpathlineto{\pgfqpoint{2.022380in}{0.640455in}}\pgfpathlineto{\pgfqpoint{2.022380in}{0.640455in}}\pgfpathlineto{\pgfqpoint{1.135296in}{0.635203in}}\pgfpathlineto{\pgfqpoint{0.691754in}{0.677221in}}\pgfpathlineto{\pgfqpoint{0.469983in}{0.669342in}}\pgfpathlineto{\pgfqpoint{0.469983in}{0.669342in}}\pgfpathclose \pgfusepath{fill}\end{pgfscope}\begin{pgfscope}\pgfpathrectangle{\pgfqpoint{0.392364in}{0.464507in}}{\pgfqpoint{1.707636in}{0.835493in}}\pgfusepath{clip}\pgfsetbuttcap \pgfsetroundjoin \definecolor{currentfill}{rgb}{0.215686,0.494118,0.721569}\pgfsetfillcolor{currentfill}\pgfsetlinewidth{0.000000pt}\definecolor{currentstroke}{rgb}{0.000000,0.000000,0.000000}\pgfsetstrokecolor{currentstroke}\pgfsetdash{}{0pt}\pgfpathmoveto{\pgfqpoint{0.469983in}{1.260215in}}\pgfpathlineto{\pgfqpoint{0.469983in}{0.669342in}}\pgfpathlineto{\pgfqpoint{0.691754in}{0.677221in}}\pgfpathlineto{\pgfqpoint{1.135296in}{0.635203in}}\pgfpathlineto{\pgfqpoint{2.022380in}{0.640455in}}\pgfpathlineto{\pgfqpoint{2.022380in}{0.771760in}}\pgfpathlineto{\pgfqpoint{2.022380in}{0.771760in}}\pgfpathlineto{\pgfqpoint{1.135296in}{0.861048in}}\pgfpathlineto{\pgfqpoint{0.691754in}{0.994979in}}\pgfpathlineto{\pgfqpoint{0.469983in}{1.260215in}}\pgfpathlineto{\pgfqpoint{0.469983in}{1.260215in}}\pgfpathclose \pgfusepath{fill}\end{pgfscope}\begin{pgfscope}\pgfsetrectcap \pgfsetmiterjoin \pgfsetlinewidth{0.501875pt}\definecolor{currentstroke}{rgb}{0.000000,0.000000,0.000000}\pgfsetstrokecolor{currentstroke}\pgfsetdash{}{0pt}\pgfpathmoveto{\pgfqpoint{0.392364in}{0.464507in}}\pgfpathlineto{\pgfqpoint{0.392364in}{1.300000in}}\pgfusepath{stroke}\end{pgfscope}\begin{pgfscope}\pgfsetrectcap \pgfsetmiterjoin \pgfsetlinewidth{0.501875pt}\definecolor{currentstroke}{rgb}{0.000000,0.000000,0.000000}\pgfsetstrokecolor{currentstroke}\pgfsetdash{}{0pt}\pgfpathmoveto{\pgfqpoint{0.392364in}{0.464507in}}\pgfpathlineto{\pgfqpoint{2.100000in}{0.464507in}}\pgfusepath{stroke}\end{pgfscope}\begin{pgfscope}\definecolor{textcolor}{rgb}{0.000000,0.000000,0.000000}\pgfsetstrokecolor{textcolor}\pgfsetfillcolor{textcolor}\pgftext[x=0.469983in,y=0.367284in,,top]{\color{textcolor}\rmfamily\fontsize{8.000000}{9.600000}\selectfont \(\displaystyle {1}\)}\end{pgfscope}\begin{pgfscope}\definecolor{textcolor}{rgb}{0.000000,0.000000,0.000000}\pgfsetstrokecolor{textcolor}\pgfsetfillcolor{textcolor}\pgftext[x=0.691754in,y=0.367284in,,top]{\color{textcolor}\rmfamily\fontsize{8.000000}{9.600000}\selectfont \(\displaystyle {2}\)}\end{pgfscope}\begin{pgfscope}\definecolor{textcolor}{rgb}{0.000000,0.000000,0.000000}\pgfsetstrokecolor{textcolor}\pgfsetfillcolor{textcolor}\pgftext[x=1.135296in,y=0.367284in,,top]{\color{textcolor}\rmfamily\fontsize{8.000000}{9.600000}\selectfont \(\displaystyle {4}\)}\end{pgfscope}\begin{pgfscope}\definecolor{textcolor}{rgb}{0.000000,0.000000,0.000000}\pgfsetstrokecolor{textcolor}\pgfsetfillcolor{textcolor}\pgftext[x=2.022380in,y=0.367284in,,top]{\color{textcolor}\rmfamily\fontsize{8.000000}{9.600000}\selectfont \(\displaystyle {8}\)}\end{pgfscope}\begin{pgfscope}\definecolor{textcolor}{rgb}{0.000000,0.000000,0.000000}\pgfsetstrokecolor{textcolor}\pgfsetfillcolor{textcolor}\pgftext[x=0.922788in, y=0.135803in, left, base]{\color{textcolor}\rmfamily\fontsize{8.000000}{9.600000}\selectfont Client cores }\end{pgfscope}\begin{pgfscope}\definecolor{textcolor}{rgb}{0.000000,0.000000,0.000000}\pgfsetstrokecolor{textcolor}\pgfsetfillcolor{textcolor}\pgftext[x=1.040740in, y=0.021605in, left, base]{\color{textcolor}\rmfamily\fontsize{8.000000}{9.600000}\selectfont  \textbf{FIDO2}}\end{pgfscope}\begin{pgfscope}\pgfpathrectangle{\pgfqpoint{0.392364in}{0.464507in}}{\pgfqpoint{1.707636in}{0.835493in}}\pgfusepath{clip}\pgfsetbuttcap \pgfsetroundjoin \pgfsetlinewidth{0.803000pt}\definecolor{currentstroke}{rgb}{0.900000,0.900000,0.900000}\pgfsetstrokecolor{currentstroke}\pgfsetdash{{0.800000pt}{1.320000pt}}{0.000000pt}\pgfpathmoveto{\pgfqpoint{0.392364in}{0.464507in}}\pgfpathlineto{\pgfqpoint{2.100000in}{0.464507in}}\pgfusepath{stroke}\end{pgfscope}\begin{pgfscope}\definecolor{textcolor}{rgb}{0.000000,0.000000,0.000000}\pgfsetstrokecolor{textcolor}\pgfsetfillcolor{textcolor}\pgftext[x=0.284724in, y=0.425926in, left, base]{\color{textcolor}\rmfamily\fontsize{8.000000}{9.600000}\selectfont \(\displaystyle {0}\)}\end{pgfscope}\begin{pgfscope}\pgfpathrectangle{\pgfqpoint{0.392364in}{0.464507in}}{\pgfqpoint{1.707636in}{0.835493in}}\pgfusepath{clip}\pgfsetbuttcap \pgfsetroundjoin \pgfsetlinewidth{0.803000pt}\definecolor{currentstroke}{rgb}{0.900000,0.900000,0.900000}\pgfsetstrokecolor{currentstroke}\pgfsetdash{{0.800000pt}{1.320000pt}}{0.000000pt}\pgfpathmoveto{\pgfqpoint{0.392364in}{0.727116in}}\pgfpathlineto{\pgfqpoint{2.100000in}{0.727116in}}\pgfusepath{stroke}\end{pgfscope}\begin{pgfscope}\definecolor{textcolor}{rgb}{0.000000,0.000000,0.000000}\pgfsetstrokecolor{textcolor}\pgfsetfillcolor{textcolor}\pgftext[x=0.166667in, y=0.688536in, left, base]{\color{textcolor}\rmfamily\fontsize{8.000000}{9.600000}\selectfont \(\displaystyle {100}\)}\end{pgfscope}\begin{pgfscope}\pgfpathrectangle{\pgfqpoint{0.392364in}{0.464507in}}{\pgfqpoint{1.707636in}{0.835493in}}\pgfusepath{clip}\pgfsetbuttcap \pgfsetroundjoin \pgfsetlinewidth{0.803000pt}\definecolor{currentstroke}{rgb}{0.900000,0.900000,0.900000}\pgfsetstrokecolor{currentstroke}\pgfsetdash{{0.800000pt}{1.320000pt}}{0.000000pt}\pgfpathmoveto{\pgfqpoint{0.392364in}{0.989726in}}\pgfpathlineto{\pgfqpoint{2.100000in}{0.989726in}}\pgfusepath{stroke}\end{pgfscope}\begin{pgfscope}\definecolor{textcolor}{rgb}{0.000000,0.000000,0.000000}\pgfsetstrokecolor{textcolor}\pgfsetfillcolor{textcolor}\pgftext[x=0.166667in, y=0.951146in, left, base]{\color{textcolor}\rmfamily\fontsize{8.000000}{9.600000}\selectfont \(\displaystyle {200}\)}\end{pgfscope}\begin{pgfscope}\pgfpathrectangle{\pgfqpoint{0.392364in}{0.464507in}}{\pgfqpoint{1.707636in}{0.835493in}}\pgfusepath{clip}\pgfsetbuttcap \pgfsetroundjoin \pgfsetlinewidth{0.803000pt}\definecolor{currentstroke}{rgb}{0.900000,0.900000,0.900000}\pgfsetstrokecolor{currentstroke}\pgfsetdash{{0.800000pt}{1.320000pt}}{0.000000pt}\pgfpathmoveto{\pgfqpoint{0.392364in}{1.252336in}}\pgfpathlineto{\pgfqpoint{2.100000in}{1.252336in}}\pgfusepath{stroke}\end{pgfscope}\begin{pgfscope}\definecolor{textcolor}{rgb}{0.000000,0.000000,0.000000}\pgfsetstrokecolor{textcolor}\pgfsetfillcolor{textcolor}\pgftext[x=0.166667in, y=1.213756in, left, base]{\color{textcolor}\rmfamily\fontsize{8.000000}{9.600000}\selectfont \(\displaystyle {300}\)}\end{pgfscope}\begin{pgfscope}\definecolor{textcolor}{rgb}{0.000000,0.000000,0.000000}\pgfsetstrokecolor{textcolor}\pgfsetfillcolor{textcolor}\pgftext[x=0.111111in,y=0.882253in,,bottom,rotate=90.000000]{\color{textcolor}\rmfamily\fontsize{8.000000}{9.600000}\selectfont Auth time (ms)}\end{pgfscope}\end{pgfpicture}\makeatother \endgroup      \begingroup \makeatletter \begin{pgfpicture}\pgfpathrectangle{\pgfpointorigin}{\pgfqpoint{2.103119in}{1.311572in}}\pgfusepath{use as bounding box, clip}\begin{pgfscope}\pgfsetbuttcap \pgfsetmiterjoin \definecolor{currentfill}{rgb}{1.000000,1.000000,1.000000}\pgfsetfillcolor{currentfill}\pgfsetlinewidth{0.000000pt}\definecolor{currentstroke}{rgb}{1.000000,1.000000,1.000000}\pgfsetstrokecolor{currentstroke}\pgfsetdash{}{0pt}\pgfpathmoveto{\pgfqpoint{0.000000in}{0.000000in}}\pgfpathlineto{\pgfqpoint{2.103119in}{0.000000in}}\pgfpathlineto{\pgfqpoint{2.103119in}{1.311572in}}\pgfpathlineto{\pgfqpoint{0.000000in}{1.311572in}}\pgfpathlineto{\pgfqpoint{0.000000in}{0.000000in}}\pgfpathclose \pgfusepath{fill}\end{pgfscope}\begin{pgfscope}\pgfsetbuttcap \pgfsetmiterjoin \definecolor{currentfill}{rgb}{1.000000,1.000000,1.000000}\pgfsetfillcolor{currentfill}\pgfsetlinewidth{0.000000pt}\definecolor{currentstroke}{rgb}{0.000000,0.000000,0.000000}\pgfsetstrokecolor{currentstroke}\pgfsetstrokeopacity{0.000000}\pgfsetdash{}{0pt}\pgfpathmoveto{\pgfqpoint{0.317568in}{0.464507in}}\pgfpathlineto{\pgfqpoint{2.103119in}{0.464507in}}\pgfpathlineto{\pgfqpoint{2.103119in}{1.297117in}}\pgfpathlineto{\pgfqpoint{0.317568in}{1.297117in}}\pgfpathlineto{\pgfqpoint{0.317568in}{0.464507in}}\pgfpathclose \pgfusepath{fill}\end{pgfscope}\begin{pgfscope}\pgfpathrectangle{\pgfqpoint{0.317568in}{0.464507in}}{\pgfqpoint{1.785551in}{0.832611in}}\pgfusepath{clip}\pgfsetbuttcap \pgfsetroundjoin \definecolor{currentfill}{rgb}{1.000000,0.498039,0.000000}\pgfsetfillcolor{currentfill}\pgfsetlinewidth{0.000000pt}\definecolor{currentstroke}{rgb}{0.000000,0.000000,0.000000}\pgfsetstrokecolor{currentstroke}\pgfsetdash{}{0pt}\pgfpathmoveto{\pgfqpoint{0.398729in}{0.534360in}}\pgfpathlineto{\pgfqpoint{0.398729in}{0.464507in}}\pgfpathlineto{\pgfqpoint{0.398729in}{0.464507in}}\pgfpathlineto{\pgfqpoint{0.405095in}{0.464507in}}\pgfpathlineto{\pgfqpoint{0.405095in}{0.464507in}}\pgfpathlineto{\pgfqpoint{0.417826in}{0.464507in}}\pgfpathlineto{\pgfqpoint{0.417826in}{0.464507in}}\pgfpathlineto{\pgfqpoint{0.443288in}{0.464507in}}\pgfpathlineto{\pgfqpoint{0.443288in}{0.464507in}}\pgfpathlineto{\pgfqpoint{0.494213in}{0.464507in}}\pgfpathlineto{\pgfqpoint{0.494213in}{0.464507in}}\pgfpathlineto{\pgfqpoint{0.596063in}{0.464507in}}\pgfpathlineto{\pgfqpoint{0.596063in}{0.464507in}}\pgfpathlineto{\pgfqpoint{0.799762in}{0.464507in}}\pgfpathlineto{\pgfqpoint{0.799762in}{0.464507in}}\pgfpathlineto{\pgfqpoint{1.207160in}{0.464507in}}\pgfpathlineto{\pgfqpoint{1.207160in}{0.464507in}}\pgfpathlineto{\pgfqpoint{2.021957in}{0.464507in}}\pgfpathlineto{\pgfqpoint{2.021957in}{0.534036in}}\pgfpathlineto{\pgfqpoint{2.021957in}{0.534036in}}\pgfpathlineto{\pgfqpoint{1.207160in}{0.534036in}}\pgfpathlineto{\pgfqpoint{1.207160in}{0.533713in}}\pgfpathlineto{\pgfqpoint{0.799762in}{0.533713in}}\pgfpathlineto{\pgfqpoint{0.799762in}{0.535007in}}\pgfpathlineto{\pgfqpoint{0.596063in}{0.535007in}}\pgfpathlineto{\pgfqpoint{0.596063in}{0.533390in}}\pgfpathlineto{\pgfqpoint{0.494213in}{0.533390in}}\pgfpathlineto{\pgfqpoint{0.494213in}{0.532743in}}\pgfpathlineto{\pgfqpoint{0.443288in}{0.532743in}}\pgfpathlineto{\pgfqpoint{0.443288in}{0.533390in}}\pgfpathlineto{\pgfqpoint{0.417826in}{0.533390in}}\pgfpathlineto{\pgfqpoint{0.417826in}{0.532419in}}\pgfpathlineto{\pgfqpoint{0.405095in}{0.532419in}}\pgfpathlineto{\pgfqpoint{0.405095in}{0.533713in}}\pgfpathlineto{\pgfqpoint{0.398729in}{0.533713in}}\pgfpathlineto{\pgfqpoint{0.398729in}{0.534360in}}\pgfpathlineto{\pgfqpoint{0.398729in}{0.534360in}}\pgfpathclose \pgfusepath{fill}\end{pgfscope}\begin{pgfscope}\pgfpathrectangle{\pgfqpoint{0.317568in}{0.464507in}}{\pgfqpoint{1.785551in}{0.832611in}}\pgfusepath{clip}\pgfsetbuttcap \pgfsetroundjoin \definecolor{currentfill}{rgb}{0.596078,0.305882,0.639216}\pgfsetfillcolor{currentfill}\pgfsetlinewidth{0.000000pt}\definecolor{currentstroke}{rgb}{0.000000,0.000000,0.000000}\pgfsetstrokecolor{currentstroke}\pgfsetdash{}{0pt}\pgfpathmoveto{\pgfqpoint{0.398729in}{0.537594in}}\pgfpathlineto{\pgfqpoint{0.398729in}{0.534360in}}\pgfpathlineto{\pgfqpoint{0.398729in}{0.533713in}}\pgfpathlineto{\pgfqpoint{0.405095in}{0.533713in}}\pgfpathlineto{\pgfqpoint{0.405095in}{0.532419in}}\pgfpathlineto{\pgfqpoint{0.417826in}{0.532419in}}\pgfpathlineto{\pgfqpoint{0.417826in}{0.533390in}}\pgfpathlineto{\pgfqpoint{0.443288in}{0.533390in}}\pgfpathlineto{\pgfqpoint{0.443288in}{0.532743in}}\pgfpathlineto{\pgfqpoint{0.494213in}{0.532743in}}\pgfpathlineto{\pgfqpoint{0.494213in}{0.533390in}}\pgfpathlineto{\pgfqpoint{0.596063in}{0.533390in}}\pgfpathlineto{\pgfqpoint{0.596063in}{0.535007in}}\pgfpathlineto{\pgfqpoint{0.799762in}{0.535007in}}\pgfpathlineto{\pgfqpoint{0.799762in}{0.533713in}}\pgfpathlineto{\pgfqpoint{1.207160in}{0.533713in}}\pgfpathlineto{\pgfqpoint{1.207160in}{0.534036in}}\pgfpathlineto{\pgfqpoint{2.021957in}{0.534036in}}\pgfpathlineto{\pgfqpoint{2.021957in}{0.750387in}}\pgfpathlineto{\pgfqpoint{2.021957in}{0.750387in}}\pgfpathlineto{\pgfqpoint{1.207160in}{0.750387in}}\pgfpathlineto{\pgfqpoint{1.207160in}{0.653045in}}\pgfpathlineto{\pgfqpoint{0.799762in}{0.653045in}}\pgfpathlineto{\pgfqpoint{0.799762in}{0.602919in}}\pgfpathlineto{\pgfqpoint{0.596063in}{0.602919in}}\pgfpathlineto{\pgfqpoint{0.596063in}{0.575431in}}\pgfpathlineto{\pgfqpoint{0.494213in}{0.575431in}}\pgfpathlineto{\pgfqpoint{0.494213in}{0.561848in}}\pgfpathlineto{\pgfqpoint{0.443288in}{0.561848in}}\pgfpathlineto{\pgfqpoint{0.443288in}{0.552793in}}\pgfpathlineto{\pgfqpoint{0.417826in}{0.552793in}}\pgfpathlineto{\pgfqpoint{0.417826in}{0.545355in}}\pgfpathlineto{\pgfqpoint{0.405095in}{0.545355in}}\pgfpathlineto{\pgfqpoint{0.405095in}{0.540181in}}\pgfpathlineto{\pgfqpoint{0.398729in}{0.540181in}}\pgfpathlineto{\pgfqpoint{0.398729in}{0.537594in}}\pgfpathlineto{\pgfqpoint{0.398729in}{0.537594in}}\pgfpathclose \pgfusepath{fill}\end{pgfscope}\begin{pgfscope}\pgfpathrectangle{\pgfqpoint{0.317568in}{0.464507in}}{\pgfqpoint{1.785551in}{0.832611in}}\pgfusepath{clip}\pgfsetbuttcap \pgfsetroundjoin \definecolor{currentfill}{rgb}{0.215686,0.494118,0.721569}\pgfsetfillcolor{currentfill}\pgfsetlinewidth{0.000000pt}\definecolor{currentstroke}{rgb}{0.000000,0.000000,0.000000}\pgfsetstrokecolor{currentstroke}\pgfsetdash{}{0pt}\pgfpathmoveto{\pgfqpoint{0.398729in}{0.540828in}}\pgfpathlineto{\pgfqpoint{0.398729in}{0.537594in}}\pgfpathlineto{\pgfqpoint{0.398729in}{0.540181in}}\pgfpathlineto{\pgfqpoint{0.405095in}{0.540181in}}\pgfpathlineto{\pgfqpoint{0.405095in}{0.545355in}}\pgfpathlineto{\pgfqpoint{0.417826in}{0.545355in}}\pgfpathlineto{\pgfqpoint{0.417826in}{0.552793in}}\pgfpathlineto{\pgfqpoint{0.443288in}{0.552793in}}\pgfpathlineto{\pgfqpoint{0.443288in}{0.561848in}}\pgfpathlineto{\pgfqpoint{0.494213in}{0.561848in}}\pgfpathlineto{\pgfqpoint{0.494213in}{0.575431in}}\pgfpathlineto{\pgfqpoint{0.596063in}{0.575431in}}\pgfpathlineto{\pgfqpoint{0.596063in}{0.602919in}}\pgfpathlineto{\pgfqpoint{0.799762in}{0.602919in}}\pgfpathlineto{\pgfqpoint{0.799762in}{0.653045in}}\pgfpathlineto{\pgfqpoint{1.207160in}{0.653045in}}\pgfpathlineto{\pgfqpoint{1.207160in}{0.750387in}}\pgfpathlineto{\pgfqpoint{2.021957in}{0.750387in}}\pgfpathlineto{\pgfqpoint{2.021957in}{1.257469in}}\pgfpathlineto{\pgfqpoint{2.021957in}{1.257469in}}\pgfpathlineto{\pgfqpoint{1.207160in}{1.257469in}}\pgfpathlineto{\pgfqpoint{1.207160in}{0.877158in}}\pgfpathlineto{\pgfqpoint{0.799762in}{0.877158in}}\pgfpathlineto{\pgfqpoint{0.799762in}{0.704788in}}\pgfpathlineto{\pgfqpoint{0.596063in}{0.704788in}}\pgfpathlineto{\pgfqpoint{0.596063in}{0.623940in}}\pgfpathlineto{\pgfqpoint{0.494213in}{0.623940in}}\pgfpathlineto{\pgfqpoint{0.494213in}{0.585779in}}\pgfpathlineto{\pgfqpoint{0.443288in}{0.585779in}}\pgfpathlineto{\pgfqpoint{0.443288in}{0.566376in}}\pgfpathlineto{\pgfqpoint{0.417826in}{0.566376in}}\pgfpathlineto{\pgfqpoint{0.417826in}{0.554410in}}\pgfpathlineto{\pgfqpoint{0.405095in}{0.554410in}}\pgfpathlineto{\pgfqpoint{0.405095in}{0.546649in}}\pgfpathlineto{\pgfqpoint{0.398729in}{0.546649in}}\pgfpathlineto{\pgfqpoint{0.398729in}{0.540828in}}\pgfpathlineto{\pgfqpoint{0.398729in}{0.540828in}}\pgfpathclose \pgfusepath{fill}\end{pgfscope}\begin{pgfscope}\pgfsetrectcap \pgfsetmiterjoin \pgfsetlinewidth{0.501875pt}\definecolor{currentstroke}{rgb}{0.000000,0.000000,0.000000}\pgfsetstrokecolor{currentstroke}\pgfsetdash{}{0pt}\pgfpathmoveto{\pgfqpoint{0.392364in}{0.464507in}}\pgfpathlineto{\pgfqpoint{0.392364in}{1.297117in}}\pgfusepath{stroke}\end{pgfscope}\begin{pgfscope}\pgfsetrectcap \pgfsetmiterjoin \pgfsetlinewidth{0.501875pt}\definecolor{currentstroke}{rgb}{0.000000,0.000000,0.000000}\pgfsetstrokecolor{currentstroke}\pgfsetdash{}{0pt}\pgfpathmoveto{\pgfqpoint{0.317568in}{0.464507in}}\pgfpathlineto{\pgfqpoint{2.103119in}{0.464507in}}\pgfusepath{stroke}\end{pgfscope}\begin{pgfscope}\definecolor{textcolor}{rgb}{0.000000,0.000000,0.000000}\pgfsetstrokecolor{textcolor}\pgfsetfillcolor{textcolor}\pgftext[x=0.392364in,y=0.367284in,,top]{\color{textcolor}\rmfamily\fontsize{8.000000}{9.600000}\selectfont \(\displaystyle {0}\)}\end{pgfscope}\begin{pgfscope}\definecolor{textcolor}{rgb}{0.000000,0.000000,0.000000}\pgfsetstrokecolor{textcolor}\pgfsetfillcolor{textcolor}\pgftext[x=0.710644in,y=0.367284in,,top]{\color{textcolor}\rmfamily\fontsize{8.000000}{9.600000}\selectfont \(\displaystyle {100}\)}\end{pgfscope}\begin{pgfscope}\definecolor{textcolor}{rgb}{0.000000,0.000000,0.000000}\pgfsetstrokecolor{textcolor}\pgfsetfillcolor{textcolor}\pgftext[x=1.028924in,y=0.367284in,,top]{\color{textcolor}\rmfamily\fontsize{8.000000}{9.600000}\selectfont \(\displaystyle {200}\)}\end{pgfscope}\begin{pgfscope}\definecolor{textcolor}{rgb}{0.000000,0.000000,0.000000}\pgfsetstrokecolor{textcolor}\pgfsetfillcolor{textcolor}\pgftext[x=1.347204in,y=0.367284in,,top]{\color{textcolor}\rmfamily\fontsize{8.000000}{9.600000}\selectfont \(\displaystyle {300}\)}\end{pgfscope}\begin{pgfscope}\definecolor{textcolor}{rgb}{0.000000,0.000000,0.000000}\pgfsetstrokecolor{textcolor}\pgfsetfillcolor{textcolor}\pgftext[x=1.665484in,y=0.367284in,,top]{\color{textcolor}\rmfamily\fontsize{8.000000}{9.600000}\selectfont \(\displaystyle {400}\)}\end{pgfscope}\begin{pgfscope}\definecolor{textcolor}{rgb}{0.000000,0.000000,0.000000}\pgfsetstrokecolor{textcolor}\pgfsetfillcolor{textcolor}\pgftext[x=1.983764in,y=0.367284in,,top]{\color{textcolor}\rmfamily\fontsize{8.000000}{9.600000}\selectfont \(\displaystyle {500}\)}\end{pgfscope}\begin{pgfscope}\definecolor{textcolor}{rgb}{0.000000,0.000000,0.000000}\pgfsetstrokecolor{textcolor}\pgfsetfillcolor{textcolor}\pgftext[x=0.800895in, y=0.135803in, left, base]{\color{textcolor}\rmfamily\fontsize{8.000000}{9.600000}\selectfont Relying parties }\end{pgfscope}\begin{pgfscope}\definecolor{textcolor}{rgb}{0.000000,0.000000,0.000000}\pgfsetstrokecolor{textcolor}\pgfsetfillcolor{textcolor}\pgftext[x=0.904982in, y=0.021605in, left, base]{\color{textcolor}\rmfamily\fontsize{8.000000}{9.600000}\selectfont  \textbf{Passwords}}\end{pgfscope}\begin{pgfscope}\pgfpathrectangle{\pgfqpoint{0.317568in}{0.464507in}}{\pgfqpoint{1.785551in}{0.832611in}}\pgfusepath{clip}\pgfsetbuttcap \pgfsetroundjoin \pgfsetlinewidth{0.803000pt}\definecolor{currentstroke}{rgb}{0.900000,0.900000,0.900000}\pgfsetstrokecolor{currentstroke}\pgfsetdash{{0.800000pt}{1.320000pt}}{0.000000pt}\pgfpathmoveto{\pgfqpoint{0.317568in}{0.464507in}}\pgfpathlineto{\pgfqpoint{2.103119in}{0.464507in}}\pgfusepath{stroke}\end{pgfscope}\begin{pgfscope}\definecolor{textcolor}{rgb}{0.000000,0.000000,0.000000}\pgfsetstrokecolor{textcolor}\pgfsetfillcolor{textcolor}\pgftext[x=0.284724in, y=0.425926in, left, base]{\color{textcolor}\rmfamily\fontsize{8.000000}{9.600000}\selectfont \(\displaystyle {0}\)}\end{pgfscope}\begin{pgfscope}\pgfpathrectangle{\pgfqpoint{0.317568in}{0.464507in}}{\pgfqpoint{1.785551in}{0.832611in}}\pgfusepath{clip}\pgfsetbuttcap \pgfsetroundjoin \pgfsetlinewidth{0.803000pt}\definecolor{currentstroke}{rgb}{0.900000,0.900000,0.900000}\pgfsetstrokecolor{currentstroke}\pgfsetdash{{0.800000pt}{1.320000pt}}{0.000000pt}\pgfpathmoveto{\pgfqpoint{0.317568in}{0.626204in}}\pgfpathlineto{\pgfqpoint{2.103119in}{0.626204in}}\pgfusepath{stroke}\end{pgfscope}\begin{pgfscope}\definecolor{textcolor}{rgb}{0.000000,0.000000,0.000000}\pgfsetstrokecolor{textcolor}\pgfsetfillcolor{textcolor}\pgftext[x=0.225695in, y=0.587623in, left, base]{\color{textcolor}\rmfamily\fontsize{8.000000}{9.600000}\selectfont \(\displaystyle {50}\)}\end{pgfscope}\begin{pgfscope}\pgfpathrectangle{\pgfqpoint{0.317568in}{0.464507in}}{\pgfqpoint{1.785551in}{0.832611in}}\pgfusepath{clip}\pgfsetbuttcap \pgfsetroundjoin \pgfsetlinewidth{0.803000pt}\definecolor{currentstroke}{rgb}{0.900000,0.900000,0.900000}\pgfsetstrokecolor{currentstroke}\pgfsetdash{{0.800000pt}{1.320000pt}}{0.000000pt}\pgfpathmoveto{\pgfqpoint{0.317568in}{0.787901in}}\pgfpathlineto{\pgfqpoint{2.103119in}{0.787901in}}\pgfusepath{stroke}\end{pgfscope}\begin{pgfscope}\definecolor{textcolor}{rgb}{0.000000,0.000000,0.000000}\pgfsetstrokecolor{textcolor}\pgfsetfillcolor{textcolor}\pgftext[x=0.166667in, y=0.749321in, left, base]{\color{textcolor}\rmfamily\fontsize{8.000000}{9.600000}\selectfont \(\displaystyle {100}\)}\end{pgfscope}\begin{pgfscope}\pgfpathrectangle{\pgfqpoint{0.317568in}{0.464507in}}{\pgfqpoint{1.785551in}{0.832611in}}\pgfusepath{clip}\pgfsetbuttcap \pgfsetroundjoin \pgfsetlinewidth{0.803000pt}\definecolor{currentstroke}{rgb}{0.900000,0.900000,0.900000}\pgfsetstrokecolor{currentstroke}\pgfsetdash{{0.800000pt}{1.320000pt}}{0.000000pt}\pgfpathmoveto{\pgfqpoint{0.317568in}{0.949598in}}\pgfpathlineto{\pgfqpoint{2.103119in}{0.949598in}}\pgfusepath{stroke}\end{pgfscope}\begin{pgfscope}\definecolor{textcolor}{rgb}{0.000000,0.000000,0.000000}\pgfsetstrokecolor{textcolor}\pgfsetfillcolor{textcolor}\pgftext[x=0.166667in, y=0.911018in, left, base]{\color{textcolor}\rmfamily\fontsize{8.000000}{9.600000}\selectfont \(\displaystyle {150}\)}\end{pgfscope}\begin{pgfscope}\pgfpathrectangle{\pgfqpoint{0.317568in}{0.464507in}}{\pgfqpoint{1.785551in}{0.832611in}}\pgfusepath{clip}\pgfsetbuttcap \pgfsetroundjoin \pgfsetlinewidth{0.803000pt}\definecolor{currentstroke}{rgb}{0.900000,0.900000,0.900000}\pgfsetstrokecolor{currentstroke}\pgfsetdash{{0.800000pt}{1.320000pt}}{0.000000pt}\pgfpathmoveto{\pgfqpoint{0.317568in}{1.111295in}}\pgfpathlineto{\pgfqpoint{2.103119in}{1.111295in}}\pgfusepath{stroke}\end{pgfscope}\begin{pgfscope}\definecolor{textcolor}{rgb}{0.000000,0.000000,0.000000}\pgfsetstrokecolor{textcolor}\pgfsetfillcolor{textcolor}\pgftext[x=0.166667in, y=1.072715in, left, base]{\color{textcolor}\rmfamily\fontsize{8.000000}{9.600000}\selectfont \(\displaystyle {200}\)}\end{pgfscope}\begin{pgfscope}\pgfpathrectangle{\pgfqpoint{0.317568in}{0.464507in}}{\pgfqpoint{1.785551in}{0.832611in}}\pgfusepath{clip}\pgfsetbuttcap \pgfsetroundjoin \pgfsetlinewidth{0.803000pt}\definecolor{currentstroke}{rgb}{0.900000,0.900000,0.900000}\pgfsetstrokecolor{currentstroke}\pgfsetdash{{0.800000pt}{1.320000pt}}{0.000000pt}\pgfpathmoveto{\pgfqpoint{0.317568in}{1.272992in}}\pgfpathlineto{\pgfqpoint{2.103119in}{1.272992in}}\pgfusepath{stroke}\end{pgfscope}\begin{pgfscope}\definecolor{textcolor}{rgb}{0.000000,0.000000,0.000000}\pgfsetstrokecolor{textcolor}\pgfsetfillcolor{textcolor}\pgftext[x=0.166667in, y=1.234412in, left, base]{\color{textcolor}\rmfamily\fontsize{8.000000}{9.600000}\selectfont \(\displaystyle {250}\)}\end{pgfscope}\begin{pgfscope}\definecolor{textcolor}{rgb}{0.000000,0.000000,0.000000}\pgfsetstrokecolor{textcolor}\pgfsetfillcolor{textcolor}\pgftext[x=0.111111in,y=0.880812in,,bottom,rotate=90.000000]{\color{textcolor}\rmfamily\fontsize{8.000000}{9.600000}\selectfont Auth time (ms)}\end{pgfscope}\end{pgfpicture}\makeatother \endgroup      \begingroup \makeatletter \begin{pgfpicture}\pgfpathrectangle{\pgfpointorigin}{\pgfqpoint{2.105915in}{1.300000in}}\pgfusepath{use as bounding box, clip}\begin{pgfscope}\pgfsetbuttcap \pgfsetmiterjoin \definecolor{currentfill}{rgb}{1.000000,1.000000,1.000000}\pgfsetfillcolor{currentfill}\pgfsetlinewidth{0.000000pt}\definecolor{currentstroke}{rgb}{1.000000,1.000000,1.000000}\pgfsetstrokecolor{currentstroke}\pgfsetdash{}{0pt}\pgfpathmoveto{\pgfqpoint{0.000000in}{0.000000in}}\pgfpathlineto{\pgfqpoint{2.105915in}{0.000000in}}\pgfpathlineto{\pgfqpoint{2.105915in}{1.300000in}}\pgfpathlineto{\pgfqpoint{0.000000in}{1.300000in}}\pgfpathlineto{\pgfqpoint{0.000000in}{0.000000in}}\pgfpathclose \pgfusepath{fill}\end{pgfscope}\begin{pgfscope}\pgfsetbuttcap \pgfsetmiterjoin \definecolor{currentfill}{rgb}{1.000000,1.000000,1.000000}\pgfsetfillcolor{currentfill}\pgfsetlinewidth{0.000000pt}\definecolor{currentstroke}{rgb}{0.000000,0.000000,0.000000}\pgfsetstrokecolor{currentstroke}\pgfsetstrokeopacity{0.000000}\pgfsetdash{}{0pt}\pgfpathmoveto{\pgfqpoint{0.366129in}{0.464507in}}\pgfpathlineto{\pgfqpoint{2.096003in}{0.464507in}}\pgfpathlineto{\pgfqpoint{2.096003in}{1.300000in}}\pgfpathlineto{\pgfqpoint{0.366129in}{1.300000in}}\pgfpathlineto{\pgfqpoint{0.366129in}{0.464507in}}\pgfpathclose \pgfusepath{fill}\end{pgfscope}\begin{pgfscope}\pgfpathrectangle{\pgfqpoint{0.366129in}{0.464507in}}{\pgfqpoint{1.729874in}{0.835493in}}\pgfusepath{clip}\pgfsetbuttcap \pgfsetroundjoin \definecolor{currentfill}{rgb}{0.301961,0.686275,0.290196}\pgfsetfillcolor{currentfill}\pgfsetlinewidth{0.000000pt}\definecolor{currentstroke}{rgb}{0.000000,0.000000,0.000000}\pgfsetstrokecolor{currentstroke}\pgfsetdash{}{0pt}\pgfpathmoveto{\pgfqpoint{0.444759in}{1.112705in}}\pgfpathlineto{\pgfqpoint{0.444759in}{0.464507in}}\pgfpathlineto{\pgfqpoint{0.837913in}{0.464507in}}\pgfpathlineto{\pgfqpoint{1.231066in}{0.464507in}}\pgfpathlineto{\pgfqpoint{1.624219in}{0.464507in}}\pgfpathlineto{\pgfqpoint{2.017372in}{0.464507in}}\pgfpathlineto{\pgfqpoint{2.017372in}{1.196769in}}\pgfpathlineto{\pgfqpoint{2.017372in}{1.196769in}}\pgfpathlineto{\pgfqpoint{1.624219in}{1.182494in}}\pgfpathlineto{\pgfqpoint{1.231066in}{1.165047in}}\pgfpathlineto{\pgfqpoint{0.837913in}{1.159760in}}\pgfpathlineto{\pgfqpoint{0.444759in}{1.112705in}}\pgfpathlineto{\pgfqpoint{0.444759in}{1.112705in}}\pgfpathclose \pgfusepath{fill}\end{pgfscope}\begin{pgfscope}\pgfpathrectangle{\pgfqpoint{0.366129in}{0.464507in}}{\pgfqpoint{1.729874in}{0.835493in}}\pgfusepath{clip}\pgfsetbuttcap \pgfsetroundjoin \definecolor{currentfill}{rgb}{0.901961,0.670588,0.007843}\pgfsetfillcolor{currentfill}\pgfsetlinewidth{0.000000pt}\definecolor{currentstroke}{rgb}{0.000000,0.000000,0.000000}\pgfsetstrokecolor{currentstroke}\pgfsetdash{}{0pt}\pgfpathmoveto{\pgfqpoint{0.444759in}{1.160817in}}\pgfpathlineto{\pgfqpoint{0.444759in}{1.112705in}}\pgfpathlineto{\pgfqpoint{0.837913in}{1.159760in}}\pgfpathlineto{\pgfqpoint{1.231066in}{1.165047in}}\pgfpathlineto{\pgfqpoint{1.624219in}{1.182494in}}\pgfpathlineto{\pgfqpoint{2.017372in}{1.196769in}}\pgfpathlineto{\pgfqpoint{2.017372in}{1.260215in}}\pgfpathlineto{\pgfqpoint{2.017372in}{1.260215in}}\pgfpathlineto{\pgfqpoint{1.624219in}{1.244882in}}\pgfpathlineto{\pgfqpoint{1.231066in}{1.223734in}}\pgfpathlineto{\pgfqpoint{0.837913in}{1.209987in}}\pgfpathlineto{\pgfqpoint{0.444759in}{1.160817in}}\pgfpathlineto{\pgfqpoint{0.444759in}{1.160817in}}\pgfpathclose \pgfusepath{fill}\end{pgfscope}\begin{pgfscope}\pgfsetrectcap \pgfsetmiterjoin \pgfsetlinewidth{0.501875pt}\definecolor{currentstroke}{rgb}{0.000000,0.000000,0.000000}\pgfsetstrokecolor{currentstroke}\pgfsetdash{}{0pt}\pgfpathmoveto{\pgfqpoint{0.366129in}{0.464507in}}\pgfpathlineto{\pgfqpoint{0.366129in}{1.300000in}}\pgfusepath{stroke}\end{pgfscope}\begin{pgfscope}\pgfsetrectcap \pgfsetmiterjoin \pgfsetlinewidth{0.501875pt}\definecolor{currentstroke}{rgb}{0.000000,0.000000,0.000000}\pgfsetstrokecolor{currentstroke}\pgfsetdash{}{0pt}\pgfpathmoveto{\pgfqpoint{0.366129in}{0.464507in}}\pgfpathlineto{\pgfqpoint{2.096003in}{0.464507in}}\pgfusepath{stroke}\end{pgfscope}\begin{pgfscope}\definecolor{textcolor}{rgb}{0.000000,0.000000,0.000000}\pgfsetstrokecolor{textcolor}\pgfsetfillcolor{textcolor}\pgftext[x=0.444759in,y=0.367284in,,top]{\color{textcolor}\rmfamily\fontsize{8.000000}{9.600000}\selectfont \(\displaystyle {20}\)}\end{pgfscope}\begin{pgfscope}\definecolor{textcolor}{rgb}{0.000000,0.000000,0.000000}\pgfsetstrokecolor{textcolor}\pgfsetfillcolor{textcolor}\pgftext[x=0.837913in,y=0.367284in,,top]{\color{textcolor}\rmfamily\fontsize{8.000000}{9.600000}\selectfont \(\displaystyle {40}\)}\end{pgfscope}\begin{pgfscope}\definecolor{textcolor}{rgb}{0.000000,0.000000,0.000000}\pgfsetstrokecolor{textcolor}\pgfsetfillcolor{textcolor}\pgftext[x=1.231066in,y=0.367284in,,top]{\color{textcolor}\rmfamily\fontsize{8.000000}{9.600000}\selectfont \(\displaystyle {60}\)}\end{pgfscope}\begin{pgfscope}\definecolor{textcolor}{rgb}{0.000000,0.000000,0.000000}\pgfsetstrokecolor{textcolor}\pgfsetfillcolor{textcolor}\pgftext[x=1.624219in,y=0.367284in,,top]{\color{textcolor}\rmfamily\fontsize{8.000000}{9.600000}\selectfont \(\displaystyle {80}\)}\end{pgfscope}\begin{pgfscope}\definecolor{textcolor}{rgb}{0.000000,0.000000,0.000000}\pgfsetstrokecolor{textcolor}\pgfsetfillcolor{textcolor}\pgftext[x=2.017372in,y=0.367284in,,top]{\color{textcolor}\rmfamily\fontsize{8.000000}{9.600000}\selectfont \(\displaystyle {100}\)}\end{pgfscope}\begin{pgfscope}\definecolor{textcolor}{rgb}{0.000000,0.000000,0.000000}\pgfsetstrokecolor{textcolor}\pgfsetfillcolor{textcolor}\pgftext[x=0.821617in, y=0.135803in, left, base]{\color{textcolor}\rmfamily\fontsize{8.000000}{9.600000}\selectfont Relying parties }\end{pgfscope}\begin{pgfscope}\definecolor{textcolor}{rgb}{0.000000,0.000000,0.000000}\pgfsetstrokecolor{textcolor}\pgfsetfillcolor{textcolor}\pgftext[x=1.038857in, y=0.021605in, left, base]{\color{textcolor}\rmfamily\fontsize{8.000000}{9.600000}\selectfont  \textbf{TOTP}}\end{pgfscope}\begin{pgfscope}\pgfpathrectangle{\pgfqpoint{0.366129in}{0.464507in}}{\pgfqpoint{1.729874in}{0.835493in}}\pgfusepath{clip}\pgfsetbuttcap \pgfsetroundjoin \pgfsetlinewidth{0.803000pt}\definecolor{currentstroke}{rgb}{0.900000,0.900000,0.900000}\pgfsetstrokecolor{currentstroke}\pgfsetdash{{0.800000pt}{1.320000pt}}{0.000000pt}\pgfpathmoveto{\pgfqpoint{0.366129in}{0.464507in}}\pgfpathlineto{\pgfqpoint{2.096003in}{0.464507in}}\pgfusepath{stroke}\end{pgfscope}\begin{pgfscope}\definecolor{textcolor}{rgb}{0.000000,0.000000,0.000000}\pgfsetstrokecolor{textcolor}\pgfsetfillcolor{textcolor}\pgftext[x=0.166667in, y=0.425926in, left, base]{\color{textcolor}\rmfamily\fontsize{8.000000}{9.600000}\selectfont \(\displaystyle {0.0}\)}\end{pgfscope}\begin{pgfscope}\pgfpathrectangle{\pgfqpoint{0.366129in}{0.464507in}}{\pgfqpoint{1.729874in}{0.835493in}}\pgfusepath{clip}\pgfsetbuttcap \pgfsetroundjoin \pgfsetlinewidth{0.803000pt}\definecolor{currentstroke}{rgb}{0.900000,0.900000,0.900000}\pgfsetstrokecolor{currentstroke}\pgfsetdash{{0.800000pt}{1.320000pt}}{0.000000pt}\pgfpathmoveto{\pgfqpoint{0.366129in}{0.728861in}}\pgfpathlineto{\pgfqpoint{2.096003in}{0.728861in}}\pgfusepath{stroke}\end{pgfscope}\begin{pgfscope}\definecolor{textcolor}{rgb}{0.000000,0.000000,0.000000}\pgfsetstrokecolor{textcolor}\pgfsetfillcolor{textcolor}\pgftext[x=0.166667in, y=0.690281in, left, base]{\color{textcolor}\rmfamily\fontsize{8.000000}{9.600000}\selectfont \(\displaystyle {0.5}\)}\end{pgfscope}\begin{pgfscope}\pgfpathrectangle{\pgfqpoint{0.366129in}{0.464507in}}{\pgfqpoint{1.729874in}{0.835493in}}\pgfusepath{clip}\pgfsetbuttcap \pgfsetroundjoin \pgfsetlinewidth{0.803000pt}\definecolor{currentstroke}{rgb}{0.900000,0.900000,0.900000}\pgfsetstrokecolor{currentstroke}\pgfsetdash{{0.800000pt}{1.320000pt}}{0.000000pt}\pgfpathmoveto{\pgfqpoint{0.366129in}{0.993216in}}\pgfpathlineto{\pgfqpoint{2.096003in}{0.993216in}}\pgfusepath{stroke}\end{pgfscope}\begin{pgfscope}\definecolor{textcolor}{rgb}{0.000000,0.000000,0.000000}\pgfsetstrokecolor{textcolor}\pgfsetfillcolor{textcolor}\pgftext[x=0.166667in, y=0.954636in, left, base]{\color{textcolor}\rmfamily\fontsize{8.000000}{9.600000}\selectfont \(\displaystyle {1.0}\)}\end{pgfscope}\begin{pgfscope}\pgfpathrectangle{\pgfqpoint{0.366129in}{0.464507in}}{\pgfqpoint{1.729874in}{0.835493in}}\pgfusepath{clip}\pgfsetbuttcap \pgfsetroundjoin \pgfsetlinewidth{0.803000pt}\definecolor{currentstroke}{rgb}{0.900000,0.900000,0.900000}\pgfsetstrokecolor{currentstroke}\pgfsetdash{{0.800000pt}{1.320000pt}}{0.000000pt}\pgfpathmoveto{\pgfqpoint{0.366129in}{1.257571in}}\pgfpathlineto{\pgfqpoint{2.096003in}{1.257571in}}\pgfusepath{stroke}\end{pgfscope}\begin{pgfscope}\definecolor{textcolor}{rgb}{0.000000,0.000000,0.000000}\pgfsetstrokecolor{textcolor}\pgfsetfillcolor{textcolor}\pgftext[x=0.166667in, y=1.218991in, left, base]{\color{textcolor}\rmfamily\fontsize{8.000000}{9.600000}\selectfont \(\displaystyle {1.5}\)}\end{pgfscope}\begin{pgfscope}\definecolor{textcolor}{rgb}{0.000000,0.000000,0.000000}\pgfsetstrokecolor{textcolor}\pgfsetfillcolor{textcolor}\pgftext[x=0.111111in,y=0.882253in,,bottom,rotate=90.000000]{\color{textcolor}\rmfamily\fontsize{8.000000}{9.600000}\selectfont Auth time (s)}\end{pgfscope}\end{pgfpicture}\makeatother \endgroup      \caption{On the left, \sys FIDO2 latency decreases as the number of client
    cores increases (latency is independent of the number of relying parties).
    In the center, \sys password latency grows with the number of relying parties,
    with the majority of the time spent on client proof generation.
    On the right, \sys TOTP latency grows with the number of relying parties,
    with the majority of the time spent in an input-independent ``offline'' phase
    as opposed to the input-dependent ``online'' phase (both phases require
    network communication).}
    \label{fig:latency}
\end{figure*}

\subsection{Cost to deploy a \sys service}
\label{sec:eval:cost}

If successful, \sys can become much simpler and more efficient with a
little support from future FIDO specifications (see
Section~\ref{sec:discuss}).  Nonetheless, we show \sys is already
practical by analyzing the cost of deploying a \sys service today
(\cref{tab:eval-overview}).  We expect a \sys log service to perform
many password-based authentications, some FIDO2 authentications, and a
comparatively small number of TOTP authentications.  This is because
the majority of relying parties only support passwords, and relying
parties typically require second-factor authentication only from time
to time.

Throughout this section, we consider password-based authentication with 128
relying parties (based on the fact that the average user has roughly 100
passwords~\cite{nordpass-study}) and TOTP-based authentication with 20 relying parties
(based on the fact that Yubikey's maximum number of TOTP registrations is
32~\cite{yubikey-totp-max}).
The authentication overhead of FIDO2 in \sys is independent of the number of relying
parties the user has registered with.

\paragraph{Storage.}
For each of the three protocols, the log service must store authentication
records (timestamp, ciphertext, and signature).
Authentication records are 104B for FIDO2, 88B for TOTP, and 138B for
passwords (differences are due to ciphertext size).
The FIDO2 protocol additionally requires the client to generate presignatures
for the log, each of which is 192B.
For 10K presignatures, the log service must store 1.83MiB.
In \cref{fig:storage-cost} (left), we show how per-client log storage actually decreases
as presignatures are consumed and replaced by authentication records.
To minimize storage costs, the log service can encrypt its presignatures and
store them at the client.
The log service then simply needs to keep a counter to prevent presignature
re-use.

\paragraph{Throughput.}
In \cref{tab:eval-overview}, we show the number of auths/s a single
log service core can support assuming 128 passwords and 20 TOTP accounts.
We achieve the highest throughput for passwords (47.62 auths/cores/s), which
are the most common authentication mechanism.
For FIDO2, which can be used as either a first or second authentication
factor and is supported by fewer relying parties than passwords, we
achieve 6.18 auths/core/s.
Finally, for TOTP, which is only used as a second factor, we achieve
0.73 auths/core/s.

Our FIDO2 protocol can be instantiated with any NIZK proof system to achieve
a different tradeoff between authentication latency and log service throughput.
For example, we instantiate our system with ZKBoo, but could also use
Groth16~\cite{Groth16} to reduce communication and verifier time (increasing
log throughput).
We measure the performance of Groth16 on our \sys FIDO2 circuit on
the BN-128 curve using ZoKrates~\cite{zokrates} with libsnark~\cite{libsnark}
with a single core (we only measure the overhead of SHA-256, which dominates circuit
cost, to provide a performance lower bound).
While the verifier time is much lower (8ms) and the proof is much smaller
(4.26KiB), (1) the trusted setup requires the client to store 19.86MiB and the
log service to store 9.2MiB per client, and (2) the proving time is 4.07s, meaning that
authentication latency is much higher.

\paragraph{Cost.}
We now quantify the cost of running a \sys log service.
The cost of one core on a c5 instance is \$0.0425-\$0.085/hour depending
on instance size~\cite{ec2-costs}.
Data transfer to AWS instances is free, and data transfer from 
AWS instances costs \$0.05-\$0.09/GB
depending on the amount of data transferred per month~\cite{ec2-costs}.
In \cref{tab:eval-overview}, we show the cost of supporting
10M authentications for each authentication method with \sys.

Supporting 10M authentications requires 450 log core hours for FIDO2,
3,832 log core hours for TOTP, and 59 log core hours for passwords.
Compute for 10M authentications costs \$19.13-\$38.25 for FIDO2,
\$162.86-\$325.72 for TOTP, and \$2.51-\$5.02 for passwords.
Communication for 10M authentications costs \$0.10-\$0.19 for FIDO2,
\$17,923-\$32,262 for TOTP, and \$0.015-\$0.027 for passwords.
The high cost for TOTP is due to the large amount of communication
required at authentication: the log service must send the client 36.8MiB
for every authentication.
In both the FIDO2 and password protocols, the vast majority of the communication
overhead is due to the proof sent from the client to the log service,
which incurs no monetary cost.
We show how cost increases with the number of authentications for each
of the the authentication methods in \cref{fig:storage-cost} (right).

TOTP is substantially more expensive than FIDO2 or passwords.
However, we expect a relatively small fraction of authentication
requests to be for TOTP.

\begin{figure}[t]
\centering
    \begingroup \makeatletter 
\makeatother \endgroup      \caption{On the left, per-client storage overhead at the log
    decreases as presignatures are replaced with authentication records
    (client enrolls with 10K presignatures).
    On the right, minimum cost of supporting more authentications with passwords,
    (128 relying parties), FIDO2, and TOTP (20 relying parties).
    Both axes use a logarithmic scale.}
    \label{fig:storage-cost}
\end{figure}

\begin{figure}[t]
    \centering
    \begin{minipage}{0.6\columnwidth}
    \begingroup \makeatletter %
\makeatother \endgroup      \end{minipage}
    \begin{minipage}{0.37\columnwidth}
    \caption{Communication for \sys with passwords increases logarithmically
    with the number of relying parties (both axes use a logarithmic scale).}
    \label{fig:pw-comm}
    \end{minipage}
\end{figure}

\begin{table}[t]
    {\begin{tabular}[t]{lrrr}
    \toprule
    & \textbf{FIDO2} & \textbf{TOTP} & \textbf{Password} \\
    \midrule
        Online auth time & 150 ms & 91 ms & 74 ms \\
        Total auth time & 150 ms & 1.32 s & 74 ms \\
        \midrule
        Online auth comm. & 1.73 MiB & 201 KiB & 3.25 KiB \\
        Total auth comm. & 1.73 MiB & 65 MiB & 3.25 KiB \\
\midrule
        Auth record & 104 B & 88 B & 138 B\\
        Log presignature & 192 B & $\varnothing$ & $\varnothing$ \\
        \midrule
        Log auths/core/s & 6.18 & 0.73 & 47.62 \\
        \midrule
        10M auths min cost & \$19.19 & \$18,086 & \$2.48  \\
        10M auths max cost & \$38.37 & \$32,588 & \$4.96 \\
    \bottomrule
    \end{tabular}
    }
    \caption{Costs for \sys with FIDO2, TOTP (20 relying parties),
    and passwords (128 relying parties). We take the cost of one core
    on a c5 instance to be \$0.0425-\$0.085/hour (depending on instance
    size) and data transfer out of AWS to cost \$0.05-\$0.09/GB (depending
    on amount of data transferred)~\cite{ec2-costs}.
    For comparison, the Argon2 password hash function should take 0.5s using
    2 cores.}
    \label{tab:eval-overview}
\end{table}

 \section{Discussion}
\label{sec:discuss}

\paragraph{Deployment strategy.}
Because \sys supports passwords, TOTP, and FIDO2, people can use it
with the vast majority of web services.  In addition, \sys offers
users many of the benefits of FIDO2 without a dedicated hardware
security token, particularly FIDO2's protection against phishing.  The
flexibility for users to choose log services can foster an ecosystem of
new security products, such as log services that request login
confirmation via a mobile phone app, apps that monitor the log to
notify users of anomalous behavior, or enterprise security products
that monitor access to arbitrary third-party services that a 
company could contract with.

\paragraph{FIDO improvements.}
\Sys can benefit from enhancements we hope to see considered for
future versions of the FIDO specification.  One simple improvement
would be to support BLS signatures, which are easier to threshold and
so eliminate \sys's need for presignatures~\cite{bonehbls}.

Future versions of FIDO could also directly support secure client-side
logging by allowing the relying party to compute the encrypted log
record itself. The relying party could then ensure that the log
service receives the correct encrypted log record by checking for
the log record in the signing payload.
Specifically, the signature payload could have the form:
$$\Hash(\textsf{log-record-ciphertext},
\Hash(\textsf{remaining-FIDO-data}))~.$$
The log server can then take the outer hash preimage as input without
needing to verify anything else about the log record.

We want to allow the relying party to generate the encrypted log
record without making it possible to link users across relying parties.
Instead of giving the relying party the user's public key directly
at registration, which would link a user's identity across relying parties,
we instead give the relying party a key-private, re-randomizable
encryption of the relying party's identifier (we can achieve this using
ElGamal encryption). At authentication, the
relying party can re-randomize the ciphertext to generate the
encrypted log record.

We also hope that future FIDO revisions standardize and promote
authentication metadata as part of the challenge and hypothetical log
record field.  For users with multiple accounts at one relying party,
it would be useful to include account names as well as relying party
names in signed payloads.  It would furthermore improve security to
allow distinct types of authentication log records for different
security-sensitive operations such as authorizing payments and
changing or removing 2FA on an account.  An app monitoring a user's
log can then immediately notify the user of such operations.

\paragraph{Multiple devices.}
Clients need to authenticate to their accounts across multiple
devices, which requires synchronizing a small amount of dynamic,
secret state across devices.  Cross-device state could be stored
encrypted at the log, or could be disseminated through existing
profile synchronization mechanisms in browsers.  There is a danger of
the synchronization mechanism maliciously convincing two devices to
use the same presignature.  Therefore, presignatures should be
partitioned between devices in advance, and devices should employ
techniques such as fork consistency~\cite{mazieres:sundr-podc} to
detect and deter any rollback attacks.
Existing tools can help a user recover if she loses all
of her devices~\cite{safetypin,apple-key-vault,google-key-vault,signal-recovery}.

\paragraph{Enforcing client-specific policies.}
We can extend \sys in a straightforward way to allow the log to
enforce more complex policies on authentications.
The client could submit a policy at enrollment time, and the log service could then
enforce this policy for subsequent authentications.
If the policy decision is based on public information, the log service 
can apply the policy directly (e.g., rate-limiting, sending push notifications
to a client's mobile device).
Other policies could be based on private information.
For example, if we used \sys for cryptocurrency wallets, the log could
enforce a policy such as ``deny transactions sending more than
\$10K to addresses that are not on the allowlist.''
For policies based on private information, 
the client could send the log service a commitment to the policy at enrollment,
and the log service could then enforce the policy by running a two-party
computation or checking a zero-knowledge proof.

\paragraph{Revocation and migration.}
If a client loses her device or wants to migrate her authentication secrets from
an old device to a new device, she needs a way to easily and remotely
invalidate the secrets on the old device.
\Sys allows her to do this easily.
To migrate credentials to a new device, the client and log simply 
re-share the authentication secrets.
To invalidate the secrets on the old device, the client asks
the log to delete the old secret shares (client must authenticate with
the log first).

\paragraph{Account recovery.}
In the event that a client loses all of her devices, she needs some way
to recover her \sys account. To ensure that she can later recover her
account, the client can encrypt her \sys client state under a key derived
from her password and store the ciphertext with the \sys service.
The security of the backup is only as good as the security of the client's
password.
Alternatively, the client could choose a random key to encrypt her client
state and then back up this key using her password and secure hardware
in order to defend against password-guessing attacks~\cite{safetypin}.

\paragraph{Limitations.}
If an attacker
compromises the client's account with the log, the attacker can access
the client's entire authentication history.
To mitigate this damage, the log could delete old authentication
records (e.g., records older than one week) or re-encrypt them under
a key that the user keeps offline.

 \section{Related work}
\label{sec:rel}

\paragraph{Privacy-preserving single sign-on.}
Prior work has explored how to protect user privacy in single sign-on systems.
BrowserID~\cite{browserID} (previously implemented in Mozilla Persona and Firefox
Accounts), SPRESSO~\cite{spresso}, EL PASSO~\cite{elpasso}, UnlimitID~\cite{UnlimitID},
UPPRESSO~\cite{uppresso}, PseudoID~\cite{pseudoID}, and Hammann et al.~\cite{HSB20}
all aim to hide user login patterns from the single sign-on server. 
Several of these systems~\cite{UnlimitID,pseudoID,elpasso,HSB20} are compatible with existing
single sign-on protocols for incremental deployment.
However, with the exception of UnlimitID~\cite{UnlimitID}, these systems do not protect
user accounts from a malicious attacker that compromises the single sign-on server.
None of these systems privately log the identity of the relying party.

Separately, Privacy Pass allows a user to obtain anonymous tokens for completing CAPTCHAs, which
she can then spend at different relying parties without allowing them
to link her across sites~\cite{privacy-pass}.
Like \sys, Privacy Pass does not link users across accounts, but unlike \sys, Privacy Pass 
does not provide a mechanism for logging authentications.

\paragraph{Threshold signing.}
Our two-party ECDSA with preprocessing protocol builds on prior work
on threshold ECDSA.
MacKenzie and Reiter proposed the first threshold ECDSA protocol for a
dishonest majority specific to the two-party setting~\cite{MR01}.
Genarro et al.\ \cite{GGN16} and Lindell~\cite{lindell-ecdsa} subsequently
improved on this protocol.
Doerner et al.\ show how to achieve two-party threshold ECDSA without 
additional assumptions~\cite{DKLS18}.
Another line of work supports threshold ECDSA using generic multi-party computation over finite
fields~\cite{ST19,dnssec-ecdsa}.
A number of works show how to split ECDSA signature generation into online and offline phases
\cite{DJN+20, CGG+20, GS21, GS22, GKSS20, CMP20, XAXYC21, ecdsa-pcgs}; in many, the offline
phase is signing-key-specific, allowing for a non-interactive online
signing phase, whereas we need an offline phase that is signing-key-independent.
Abram et al.\ show how to reduce the bandwidth of the offline phase via pseudorandom
correlation generators~\cite{ecdsa-pcgs}.
Aumasson et al.\ provide a survey of prior work on threshold ECDSA
\cite{ecdsa-survey}.
Arora et al.\ show how to split trust across a group of FIDO
authenticators to enable account recovery using a new group signature
scheme~\cite{fido-group-sigs}.

\paragraph{Proving properties of encrypted data.}
\Sys's split-secret authentication protocol for FIDO2 and passwords relies on
proving properties of encrypted data, which is also explored in prior work.
Verifiable encryption was first proposed by Stadler~\cite{Sta96}, and Camenisch
and Damgard introduced it as a well-defined primitive~\cite{CD00}.
Subsequent work has designed verifiable encryption schemes for
limited classes of relations (e.g. discrete logarithms)~\cite{CS03, Ate04, YAS+12, LN17, NRSW20}.
Takahashi and Zaverucha introduced a generic compiler for MPC-in-the-head-based verifiable
encryption~\cite{TZ21}.
Lee et al.~\cite{SAVER} contribute a SNARK-based verifiable encryption scheme that
decouples the encryption function from the circuit
by using a commit-and-prove SNARK~\cite{legoSNARK}.
This approach does not work for us for FIDO2 authentication because the ciphertext must be connected to a
SHA-256 digest.

Grubbs et al.\ introduce zero-knowledge middleboxes, which enforce properties on encrypted
data using SNARKs~\cite{zkmbs}.
Wang et al.\ show how to build blind certificate authorities, enabling a certificate
authority to validate an identity and generate a certificate for it without 
learning the identity~\cite{blind-CAs}.
DECO allows users to prove that a piece of data accessed via TLS came from a particular website
and, optionally, prove statements about the data in zero-knowledge~\cite{DECO}.

\paragraph{Transparency logs.}
Like \sys, transparency logs detect attacks rather than prevent them, and they
achieve this by maintaining a log recording sensitive
actions~\cite{CONIKS,trillian,merkle2,seemless,ct,wave,safetypin}.
However, transparency logs traditionally maintain public, global state.
For example, the certificate transparency log records what certificates were issued
and by whom in order to track when certificates were issued incorrectly~\cite{ct}.
In contrast, the \sys log service maintains encrypted, per-user state about
individual users' authentication history.

 \section{Conclusion}

\Sys is an authentication manager that logs every successful
authentication to any of a user's accounts on a third-party log
service.  It guarantees log integrity without trusting clients.  It
furthermore guarantees account security and privacy without trusting
the log service.  \Sys works with any existing service supporting
FIDO2, TOTP, or password-based login.
Our evaluation shows the implementation is
practical and cost-effective.

\paragraph{Acknowledgements.}
We thank the anonymous reviewers and our shepherd Ittay Eyal for
their feedback. We also thank Raluca Ada Popa for her support, 
Albert Kwon for discussions about EdDSA, and
members of the Berkeley Sky security group and MIT PDOS group
for giving feedback that improved the presentation of this paper.
The RISELab and Sky Lab are supported by NSF CISE Expeditions
Award CCF-1730628 and gifts from
the Sloan Foundation, Alibaba, Amazon Web Services, Ant
Group, Ericsson, Facebook, Futurewei, Google, Intel, Microsoft,
Nvidia, Scotiabank, Splunk, and VMware.
Emma Dauterman was supported by an NSF Graduate Research Fellowship and a
Microsoft Ada Lovelace Research Fellowship.
This work was funded in part by the Stanford Future of Digital
Currency Initiative as well as gifts from
Capital One, Facebook, Google, Mozilla, Seagate, and MIT's
FinTech@CSAIL Initiative.
We also received support under NSF Awards CNS-2054869 and CNS-2225441.

{\bibliography{refs}

\begin{thebibliography}{10}

\bibitem{ec2-costs}
{Amazon EC2 On-Demand Pricing}.
\newblock \url{https://aws.amazon.com/ec2/pricing/on-demand/}, accessed
  December 7, 2022.

\bibitem{fido2-interop}
{FIDO Device Onboarding Conformance Server}.
\newblock \url{https://github.com/fido-alliance/iot-fdo-conformance-tools},
  Accessed June 26, 2024.

\bibitem{fido2-server-reqs}
Server requirements and transport binding profile, 2018.
\newblock
  \url{https://fidoalliance.org/specs/fido-v2.0-rd-20180702/fido-server-v2.0-rd-20180702.html},
  Accessed June 26, 2024.

\bibitem{yubikey-totp-max}
{How many accounts can I register my YubiKey with?}, 2020.
\newblock
  \url{https://support.yubico.com/hc/en-us/articles/360013790319-How-many-accounts-can-I-register-my-YubiKey-with-FIDO2}.

\bibitem{passkeys}
{Passkeys}.
\newblock Google, 2022.
\newblock \url{https://developers.google.com/identity/passkeys}.

\bibitem{ecdsa-pcgs}
Damiano Abram, Ariel Nof, Claudio Orlandi, Peter Scholl, and Omer Shlomovits.
\newblock Low-bandwidth threshold {ECDSA} via pseudorandom correlation
  generators.
\newblock In {\em IEEE Security \& Privacy}, 2022.

\bibitem{wave}
Michael~P Andersen, Sam Kumar, Moustafa AbdelBaky, Gabe Fierro, John Kolb,
  Hyung-Sin Kim, David~E Culler, and Raluca~Ada Popa.
\newblock {WAVE: A decentralized authorization framework with transitive
  delegation}.
\newblock In {\em USENIX Security}, 2019.

\bibitem{fido-group-sigs}
Sunpreet~S Arora, Saikrishna Badrinarayanan, Srinivasan Raghuraman, Maliheh
  Shirvanian, Kim Wagner, and Gaven Watson.
\newblock Avoiding lock outs: Proactive fido account recovery using managerless
  group signatures.
\newblock {\em Cryptology ePrint Archive}, 2022.

\bibitem{Ate04}
Giuseppe Ateniese.
\newblock Verifiable encryption of digital signatures and applications.
\newblock {\em ACM Transactions on Information and System Security (TISSEC)},
  7(1):1--20, 2004.

\bibitem{ecdsa-survey}
Jean-Philippe Aumasson, Adrian Hamelink, and Omer Shlomovits.
\newblock A survey of ecdsa threshold signing.
\newblock {\em Cryptology ePrint Archive}, 2020.

\bibitem{Beaver91}
Donald Beaver.
\newblock Efficient multiparty protocols using circuit randomization.
\newblock In {\em CRYPTO}, 1991.

\bibitem{ROM}
Mihir Bellare and Phillip Rogaway.
\newblock Random oracles are practical: A paradigm for designing efficient
  protocols.
\newblock In {\em {CCS}}, pages 62--73, 1993.

\bibitem{BGW88}
M~Ben-Or, S~Goldwasser, and A~Wigderson.
\newblock Completeness theorems for non-cryptographic fault-tolerant
  distributed computing.
\newblock In {\em STOC}, pages 1--10, 1988.

\bibitem{BCCT12}
Nir Bitansky, Ran Canetti, Alessandro Chiesa, and Eran Tromer.
\newblock From extractable collision resistance to succinct non-interactive
  arguments of knowledge, and back again.
\newblock In {\em ITCS}, pages 326--349, 2012.

\bibitem{BFM88}
Manuel Blum, Paul Feldman, and Silvio Micali.
\newblock Non-interactive zero-knowledge and its applications.
\newblock In {\em ACM STOC}. 1988.

\bibitem{bonehbls}
Dan Boneh, Manu Drijvers, and Gregory Neven.
\newblock {BLS} multi-signatures with public-key aggregation.
\newblock \url{https://crypto.stanford.edu/~dabo/pubs/papers/BLSmultisig.html},
  Accessed 23 May 2022, March 2018.

\bibitem{CD00}
Jan Camenisch and Ivan Damg{\aa}rd.
\newblock Verifiable encryption, group encryption, and their applications to
  separable group signatures and signature sharing schemes.
\newblock In {\em {ASIACRYPT}}, pages 331--345. Springer, 2000.

\bibitem{CS03}
Jan Camenisch and Victor Shoup.
\newblock Practical verifiable encryption and decryption of discrete
  logarithms.
\newblock In {\em {CRYPTO}}, pages 126--144. Springer, 2003.

\bibitem{legoSNARK}
Matteo Campanelli, Dario Fiore, and Ana{\"\i}s Querol.
\newblock Legosnark: Modular design and composition of succinct zero-knowledge
  proofs.
\newblock In {\em {CCS}}, pages 2075--2092, 2019.

\bibitem{CGG+20}
Ran Canetti, Rosario Gennaro, Steven Goldfeder, Nikolaos Makriyannis, and Udi
  Peled.
\newblock {UC} non-interactive, proactive, threshold {ECDSA} with identifiable
  aborts.
\newblock In {\em {CCS}}, pages 1769--1787, 2020.

\bibitem{CMP20}
Ran Canetti, Nikolaos Makriyannis, and Udi Peled.
\newblock {UC} non-interactive, proactive, threshold {ECDSA}.
\newblock {\em Cryptology ePrint Archive}, 2020.

\bibitem{zkb++}
Melissa Chase, David Derler, Steven Goldfeder, Claudio Orlandi, Sebastian
  Ramacher, Christian Rechberger, Daniel Slamanig, and Greg Zaverucha.
\newblock Post-quantum zero-knowledge and signatures from symmetric-key
  primitives.
\newblock In {\em {CCS}}, pages 1825--1842, 2017.

\bibitem{seemless}
Melissa Chase, Apoorvaa Deshpande, Esha Ghosh, and Harjasleen Malvai.
\newblock Seemless: Secure end-to-end encrypted messaging with less trust.
\newblock In {\em {CCS}}, pages 1639--1656, 2019.

\bibitem{dnssec-ecdsa}
Anders Dalskov, Claudio Orlandi, Marcel Keller, Kris Shrishak, and Haya
  Shulman.
\newblock Securing {DNSSEC} keys via threshold {ECDSA} from generic {MPC}.
\newblock In {\em European Symposium on Research in Computer Security}, pages
  654--673. Springer, 2020.

\bibitem{DJN+20}
Ivan Damg{\aa}rd, Thomas~Pelle Jakobsen, Jesper~Buus Nielsen, Jakob~Illeborg
  Pagter, and Michael~B{\ae}ksvang Osterg{\aa}rd.
\newblock Fast threshold {ECDSA} with honest majority.
\newblock In {\em {SCN}}, pages 382--400. Springer, 2020.

\bibitem{SPDZ12}
Ivan Damg{\aa}rd, Valerio Pastro, Nigel Smart, and Sarah Zakarias.
\newblock Multiparty computation from somewhat homomorphic encryption.
\newblock In {\em {CRYPTO}}, pages 643--662. Springer, 2012.

\bibitem{safetypin}
Emma Dauterman, Henry Corrigan-Gibbs, and David Mazi{\`e}res.
\newblock {SafetyPin}: Encrypted backups with {Human-Memorable} secrets.
\newblock In {\em {OSDI}}, pages 1121--1138, 2020.

\bibitem{privacy-pass}
Alex Davidson, Ian Goldberg, Nick Sullivan, George Tankersley, and Filippo
  Valsorda.
\newblock Privacy pass: Bypassing internet challenges anonymously.
\newblock {\em Proc. Priv. Enhancing Technol.}, 2018(3):164--180, 2018.

\bibitem{Desmedt87}
Yvo Desmedt.
\newblock Society and group oriented cryptography: A new concept.
\newblock In {\em {EUROCRYPT}}, pages 120--127. Springer, 1987.

\bibitem{DF89}
Yvo Desmedt and Yair Frankel.
\newblock Threshold cryptosystems.
\newblock In {\em {EUROCRYPT}}, pages 307--315. Springer, 1989.

\bibitem{pseudoID}
Arkajit Dey and Stephen Weis.
\newblock {PseudoID}: Enhancing privacy in federated login.
\newblock In {\em {HotPETS} workshop}, 2010.

\bibitem{DKLS18}
Jack Doerner, Yashvanth Kondi, Eysa Lee, and Abhi Shelat.
\newblock Secure two-party threshold {ECDSA} from {ECDSA} assumptions.
\newblock In {\em IEEE Security \& Privacy}, pages 980--997. IEEE, 2018.

\bibitem{faife:okta}
Corin Faife.
\newblock Okta ends lapsus\$ hack investigation, says breach lasted just 25
  minutes.
\newblock {\em The Verge}, April 2022.
\newblock
  \url{https://www.theverge.com/2022/4/20/23034360/okta-lapsus-hack-investigation-breach-25-minutes}.

\bibitem{browserID}
Daniel Fett, Ralf K{\"u}sters, and Guido Schmitz.
\newblock Analyzing the {BrowserID} {SSO} system with primary identity
  providers using an expressive model of the web.
\newblock In {\em European Symposium on Research in Computer Security}, pages
  43--65. Springer, 2015.

\bibitem{spresso}
Daniel Fett, Ralf K{\"u}sters, and Guido Schmitz.
\newblock Spresso: A secure, privacy-respecting single sign-on system for the
  web.
\newblock In {\em {CCS}}, pages 1358--1369, 2015.

\bibitem{FS86}
Amos Fiat and Adi Shamir.
\newblock How to prove yourself: Practical solutions to identification and
  signature problems.
\newblock In {\em {EUROCRYPT}}, pages 186--194. Springer, 1986.

\bibitem{fido2}
{FIDO Alliance}.
\newblock {FIDO} alliance specifications: Overview.
\newblock \url{https://fidoalliance.org/specifications/}, Accessed 20 May 2022.

\bibitem{CoSI}
B.~Ford, N.~Gailly, L.~Gasser, and P.~Jovanovic.
\newblock {Collective Edwards-Curve Digital Signature Algorithm}.
\newblock {Internet-Draft}, June 2017.

\bibitem{CBMC-GC}
Martin Franz, Andreas Holzer, Stefan Katzenbeisser, Christian Schallhart, and
  Helmut Veith.
\newblock {CBMC-GC: An ANSI C Compiler for Secure Two-Party Computations}.
\newblock In {\em Compiler Construction: 23rd International Conference}, volume
  8409, page 244, 2014.

\bibitem{GKSS20}
Adam G{\k{a}}gol and Damian Straszak.
\newblock Threshold ecdsa for decentralized asset custody.
\newblock Technical report, Tech. rep., Cryptology ePrint Archive, Report
  2020/498, 2020.

\bibitem{GGPR13}
Rosario Gennaro, Craig Gentry, Bryan Parno, and Mariana Raykova.
\newblock Quadratic span programs and succinct nizks without pcps.
\newblock In {\em {EUROCRYPT}}, pages 626--645. Springer, 2013.

\bibitem{GG18}
Rosario Gennaro and Steven Goldfeder.
\newblock Fast multiparty threshold {ECDSA} with fast trustless setup.
\newblock In {\em {CCS}}, pages 1179--1194, 2018.

\bibitem{GGN16}
Rosario Gennaro, Steven Goldfeder, and Arvind Narayanan.
\newblock Threshold-optimal {DSA/ECDSA} signatures and an application to
  bitcoin wallet security.
\newblock In {\em {ACNS}}, pages 156--174. Springer, 2016.

\bibitem{zkboo}
Irene Giacomelli, Jesper Madsen, and Claudio Orlandi.
\newblock {ZKBoo}: Faster zero-knowledge for boolean circuits.
\newblock In {\em {USENIX Security}}, pages 1069--1083, 2016.

\bibitem{GMR85}
Shafi Goldwasser, Silvio Micali, and Charles Rackoff.
\newblock The knowledge complexity of interactive proof systems.
\newblock {\em SIAM Journal on computing}, 18(1):186--208, 1989.

\bibitem{trillian}
Trillian.
\newblock \url{https://github.com/google/trillian}.

\bibitem{Groth16}
Jens Groth.
\newblock On the size of pairing-based non-interactive arguments.
\newblock In {\em {EUROCRYPT}}, pages 305--326. Springer, 2016.

\bibitem{GK14}
Jens Groth and Markulf Kohlweiss.
\newblock One-out-of-many proofs: Or how to leak a secret and spend a coin.
\newblock In {\em Annual International Conference on the Theory and
  Applications of Cryptographic Techniques}, pages 253--280. Springer, 2015.

\bibitem{GS22}
Jens Groth and Victor Shoup.
\newblock Design and analysis of a distributed {ECDSA} signing service.
\newblock {\em Cryptology ePrint Archive}, 2022.

\bibitem{GS21}
Jens Groth and Victor Shoup.
\newblock On the security of {ECDSA} with additive key derivation and
  presignatures.
\newblock In {\em {EUROCRYPT}}, 2022.

\bibitem{zkmbs}
Paul Grubbs, Arasu Arun, Ye~Zhang, Joseph Bonneau, and Michael Walfish.
\newblock Zero-knowledge middleboxes.
\newblock {\em Cryptology ePrint Archive}, 2021.

\bibitem{uppresso}
Chengqian Guo, Jingqiang Lin, Quanwei Cai, Fengjun Li, Qiongxiao Wang, Jiwu
  Jing, Bin Zhao, and Wei Wang.
\newblock {UPPRESSO}: Untraceable and unlinkable privacy-preserving single
  sign-on services.
\newblock {\em arXiv preprint arXiv:2110.10396}, 2021.

\bibitem{HSB20}
Sven Hammann, Ralf Sasse, and David Basin.
\newblock Privacy-preserving {OpenID} connect.
\newblock In {\em {ASIACCS}}, pages 277GS22--289, 2020.

\bibitem{merkle2}
Yuncong Hu, Kian Hooshmand, Harika Kalidhindi, Seung~Jin Yang, and Raluca~Ada
  Popa.
\newblock Merkle 2: A low-latency transparency log system.
\newblock In {\em IEEE Security \& Privacy}, pages 285--303. IEEE, 2021.

\bibitem{UnlimitID}
Marios Isaakidis, Harry Halpin, and George Danezis.
\newblock Unlimitid: Privacy-preserving federated identity management using
  algebraic macs.
\newblock In {\em Proceedings of the 2016 ACM on Workshop on Privacy in the
  Electronic Society}, pages 139--142, 2016.

\bibitem{IKOS07}
Yuval Ishai, Eyal Kushilevitz, Rafail Ostrovsky, and Amit Sahai.
\newblock Zero-knowledge from secure multiparty computation.
\newblock In {\em {STOC}}, pages 21--30, 2007.

\bibitem{apple-key-vault}
Ivan Krstic.
\newblock Behind the scenes with {iOS} security, 2016.
\newblock \url{https://www.blackhat.com/docs/us-16/materials/us-16-Krstic.pdf}.

\bibitem{kr-u2f}
Krypton.
\newblock kr-u2f.
\newblock \url{https://github.com/kryptco/kr-u2f}, Accessed 17 May 2022.

\bibitem{libsnark}
{SCIPR} Lab.
\newblock libsnark.
\newblock \url{https://github.com/scipr-lab/libsnark}, Accessed 30 May 2022.

\bibitem{paxos}
Leslie Lamport.
\newblock The part-time parliament.
\newblock In {\em {TOCS}}, pages 133--169, 1998.

\bibitem{ct}
Adam Langley, Emilia Kasper, and Ben Laurie.
\newblock Certificate transparency.
\newblock {\em Internet Engineering Task Force}, 2013.
\newblock \url{https://tools.ietf.org/html/rfc6962}.

\bibitem{SAVER}
Jiwon Lee, Jaekyoung Choi, Jihye Kim, and Hyunok Oh.
\newblock {SAVER}: Snark-friendly, additively-homomorphic, and verifiable
  encryption and decryption with rerandomization.
\newblock {\em Cryptology ePrint Archive}, 2019.

\bibitem{lindell-ecdsa}
Yehuda Lindell.
\newblock Fast secure two-party {ECDSA} signing.
\newblock In {\em {CRYPTO}}, pages 613--644. Springer, 2017.

\bibitem{signal-recovery}
Joshua Lund.
\newblock Technology preview for secure value recovery, 2019.
\newblock \url{https://signal.org/blog/secure-value-recovery/}.

\bibitem{LN17}
Vadim Lyubashevsky and Gregory Neven.
\newblock One-shot verifiable encryption from lattices.
\newblock In {\em {EUROCRYPT}}, pages 293--323. Springer, 2017.

\bibitem{MR01}
Philip MacKenzie and Michael~K Reiter.
\newblock Two-party generation of {DSA} signatures.
\newblock In {\em {CRYPTO}}, pages 137--154. Springer, 2001.

\bibitem{musig}
Gregory Maxwell, Andrew Poelstra, Yannick Seurin, and Pieter Wuille.
\newblock Simple schnorr multi-signatures with applications to bitcoin.
\newblock {\em Designs, Codes and Cryptography}, 87(9):2139--2164, 2019.

\bibitem{mazieres:sundr-podc}
David {Mazi\`eres} and Dennis Shasha.
\newblock Building secure file systems out of {B}yzantine storage.
\newblock In {\em 21st Annual ACM SIGACT-SIGOPS Symposium on Principles of
  Distributed Computing}, pages 108--117, July 2002.

\bibitem{CONIKS}
Marcela~S Melara, Aaron Blankstein, Joseph Bonneau, Edward~W Felten, and
  Michael~J Freedman.
\newblock {CONIKS}: Bringing key transparency to end users.
\newblock In {\em USENIX Security}, 2015.

\bibitem{RSA-PKCS}
Kathleen Moriarty, Burt Kaliski, Jakob Jonsson, and Andreas Rusch.
\newblock {PKCS\# 1: RSA} cryptography specifications version 2.2.
\newblock {\em Internet Engineering Task Force, Request for Comments}, 8017:72,
  2016.

\bibitem{totp-rfc}
D.~M'Raihi, S.~Machani, M.~Pei, and J.~Rydell.
\newblock {TOTP: Time-Based One-Time Password Algorithm}.
\newblock {RFC} 6238, May 2011.

\bibitem{musig2}
Jonas Nick, Tim Ruffing, and Yannick Seurin.
\newblock {MuSig2}: Simple two-round schnorr multi-signatures.
\newblock In {\em {CRYPTO}}, pages 189--221. Springer, 2021.

\bibitem{NRSW20}
Jonas Nick, Tim Ruffing, Yannick Seurin, and Pieter Wuille.
\newblock {MuSig-DN}: Schnorr multi-signatures with verifiably deterministic
  nonces.
\newblock In {\em {CCS}}, pages 1717--1731, 2020.

\bibitem{raft}
Diego Ongaro and John Ousterhout.
\newblock In search of an understandable consensus algorithm.
\newblock In {\em {USENIX ATC}}, pages 305--319, 2014.

\bibitem{pinocchio}
Bryan Parno, Jon Howell, Craig Gentry, and Mariana Raykova.
\newblock Pinocchio: Nearly practical verifiable computation.
\newblock In {\em {IEEE Security \& Privacy}}, pages 238--252. IEEE, 2013.

\bibitem{roth:lastpass}
Emma Roth.
\newblock {LastPass}' latest data breach exposed some customer information.
\newblock {\em The Verge}, November 2022.
\newblock
  \url{https://www.theverge.com/2022/11/30/23486902/lastpass-hackers-customer-information-breach}.

\bibitem{nordpass-study}
Adam Rowe.
\newblock Study reveals average person has 100 passwords, November 2021.
\newblock
  \url{https://tech.co/password-managers/how-many-passwords-average-person}.

\bibitem{sakimura:openid}
N.~Sakimura, J.~Bradley, M.~Jones, B.~de~Medeiros, and C.~Mortimore.
\newblock {OpenID} connect core 1.0 incorporating errata set 1.
\newblock \url{https://openid.net/specs/openid-connect-core-1_0.html}, November
  2014.

\bibitem{S79}
Adi Shamir.
\newblock How to share a secret.
\newblock {\em Communications of the ACM}, 22(11):612--613, 1979.

\bibitem{ST19}
Nigel~P Smart and Younes Talibi~Alaoui.
\newblock Distributing any elliptic curve based protocol.
\newblock In {\em IMA International Conference on Cryptography and Coding},
  pages 342--366. Springer, 2019.

\bibitem{Sta96}
Markus Stadler.
\newblock Publicly verifiable secret sharing.
\newblock In {\em {EUROCRYPT}}, pages 190--199. Springer, 1996.

\bibitem{TZ21}
Akira Takahashi and Greg Zaverucha.
\newblock Verifiable encryption from {MPC-in-the-Head}.
\newblock {\em Cryptology ePrint Archive}, 2021.

\bibitem{dbir}
Verizon.
\newblock {\em {DBIR} Data Breach Investigations Report}, 2022 edition.
\newblock
  \url{https://www.verizon.com/business/resources/T3cd/reports/dbir/2022-data-breach-investigations-report-dbir.pdf}.

\bibitem{webauthn}
{W3C}.
\newblock Web authentication: An api for accessing public key credentials level
  2, April 2021.
\newblock \url{https://www.w3.org/TR/webauthn-2/}, Accessed 20 May 2022.

\bibitem{google-key-vault}
Shabsi Walfish.
\newblock {Google Cloud Key Vault Service}.
\newblock Google, 2018.
\newblock
  \url{https://developer.android.com/about/versions/pie/security/ckv-whitepaper}.

\bibitem{blind-CAs}
Liang Wang, Gilad Asharov, Rafael Pass, Thomas Ristenpart, and Abhi Shelat.
\newblock Blind certificate authorities.
\newblock In {\em {IEEE Security \& Privacy}}, pages 1015--1032. IEEE, 2019.

\bibitem{emp-toolkit}
Xiao Wang, Alex~J. Malozemoff, and Jonathan Katz.
\newblock {EMP-toolkit: Efficient MultiParty computation toolkit}.
\newblock \url{https://github.com/emp-toolkit}, 2016.

\bibitem{WRK17}
Xiao Wang, Samuel Ranellucci, and Jonathan Katz.
\newblock Authenticated garbling and efficient maliciously secure two-party
  computation.
\newblock In {\em Proceedings of the 2017 ACM SIGSAC conference on computer and
  communications security}, pages 21--37, 2017.

\bibitem{XAXYC21}
Haiyang Xue, Man~Ho Au, Xiang Xie, Tsz~Hon Yuen, and Handong Cui.
\newblock Efficient online-friendly two-party {ECDSA} signature.
\newblock In {\em {CCS}}, pages 558--573, 2021.

\bibitem{YAS+12}
Shota Yamada, Nuttapong Attrapadung, Bagus Santoso, Jacob~CN Schuldt, Goichiro
  Hanaoka, and Noboru Kunihiro.
\newblock Verifiable predicate encryption and applications to {CCA} security
  and anonymous predicate authentication.
\newblock In {\em {PKC}}, pages 243--261. Springer, 2012.

\bibitem{Yao82}
Andrew~C Yao.
\newblock Protocols for secure computations.
\newblock In {\em 23rd annual symposium on foundations of computer science
  (sfcs 1982)}, pages 160--164. IEEE, 1982.

\bibitem{Yao86}
Andrew Chi-Chih Yao.
\newblock How to generate and exchange secrets.
\newblock In {\em {FOCS}}, pages 162--167. IEEE, 1986.

\bibitem{DECO}
Fan Zhang, Deepak Maram, Harjasleen Malvai, Steven Goldfeder, and Ari Juels.
\newblock {DECO}: Liberating web data using decentralized oracles for {TLS}.
\newblock In {\em {CCS}}, pages 1919--1938, 2020.

\bibitem{elpasso}
Zhiyi Zhang, Michal Kr{\'o}l, Alberto Sonnino, Lixia Zhang, and Etienne
  Rivi{\`e}re.
\newblock {EL PASSO}: Efficient and lightweight privacy-preserving single sign
  on.
\newblock {\em PoPETS}, 2021(2):70--87, 2021.

\bibitem{zokrates}
{ZoKrates}.
\newblock {ZoKrates}.
\newblock \url{https://github.com/Zokrates/ZoKrates}, Accessed 30 May 2022.

\end{thebibliography}
\bibliographystyle{plain} 
}

\iffull
\appendix
\section{ECDSA with additive key derivation and presignatures}
\label{app:secdefs}
\label{app:secdefs:gs}

We now define security for a variant of ECDSA where (1) the
adversary can choose ``tweaks'' to the signing key, and (2) the
adversary can request presignatures that are later used to generate
presignatures. In the rest of this section, we will show that the
advantage of the adversary in this modified version of ECDSA is
negligible. In \cref{app:two-ecdsa}, we will argue that a two-party
protocol that achieves the ideal functionality of this modified
version of ECDSA is secure.

Throughout all algorithms implicitly run in time polynomial in the
security parameter.

We define a security experiment for ECDSA with preprocessing and
additive key derivation in Experiment~\ref{exp:ecdsa2} in \cref{fig:mod-exp}.
The security experiment models the fact that the client's and log's secret key
shares are not authenticated, and so the adversary can query for signatures
under a signing key with adversarially chosen ``tweaks''. However, the final
signature must verify under a fixed set of tweaks (in our setting, these
correspond to public keys that the client generated at registration).
The presignature queries allow us to capture the client preprocessing that
we take advantage of.

\begin{theorem}
Let $\ECDSAModAdv[\calA,\G,\ell,N]$ denote the adversary $\calA$'s advantage
in Experiment~\ref{exp:ecdsa2} with group $\G$ of prime order $q$,
$\ell$ $\KeyGen$ queries, and $N$ total $\PreSign$ and $\Sign$ queries. Then
    $$ \ECDSAModAdv[\calA,\G,\ell,N] \leq O(N \cdot ((\ell+N)/q + 1/\sqrt{q}))~. $$
\end{theorem}

\begin{proof}
In order to prove the above theorem, we use an additional experiment,
Experiment~\ref{exp:ecdsa} in \cref{fig:gs-exp}.
Let $\ECDSAAdv[\calB,\G,N,\calE]$ denote the advantage of adversary
$\calB$ in Experiment~\ref{exp:ecdsa} with a set of tweaks $\calE$ of
size $\ell$. By Lemma~\ref{lemma:gs},
    $$ \ECDSAAdv[\calA,\G,N,\calE] \leq O(N \cdot ((|\calE|+N)/q + 1/\sqrt{q}))~,$$
and by Lemma~\ref{lemma:ecdsa-mod},
$$ \ECDSAModAdv[\calA,\G,\ell,N] \leq \ECDSAAdv[\calB,\G,N,\calE]~, $$
and because $|\calE| = \ell$,
    $$ \ECDSAModAdv[\calA,\G,\ell,N] \leq O(N \cdot ((\ell+N)/q + 1/\sqrt{q}))~.$$

\end{proof}

For the intermediate security experiment, we use the security game from Groth and
Shoup, which we include for reference as Experiment~\ref{exp:ecdsa} in \cref{fig:gs-exp}.
The intermediate security experiment allows us to
leverage Groth and Shoup's security analysis for ECDSA
with additive key derivation and presignatures~\cite{GS21}. They define a security
game that is similar to but slightly different from Experiment~\ref{exp:ecdsa2}
that we define to match our setting.
Experiment~\ref{exp:ecdsa} models the variant of additive key derivation where the
signing tweak is not constrained to lie in the set of tweaks that the
adversary must produce a forgery for (only the
forging tweak is constrained; see Note 1 in Section 6 of Groth and Shoup~\cite{GS21}).

At a high level, the differences between Experiment~\ref{exp:ecdsa2} and
Experiment~\ref{exp:ecdsa} are:
\begin{itemize}
    \item In Experiment~\ref{exp:ecdsa2}, the adversary makes $\KeyGen$ queries
        to receive public keys corresponding to the set of tweaks, whereas
        in Experiment~\ref{exp:ecdsa}, the adversary simply receives the set of
        tweaks directly at the beginning of the experiment (these are an experiment
        parameter).
    \item Signing queries in Experiment~\ref{exp:ecdsa2} take as input a share of the tweak
        rather than the entire signing tweak.
    \item The Experiment~\ref{exp:ecdsa2} challenger enforces an order on $\KeyGen$,
        $\PreSign$, and $\Sign$ queries. The Experiment~\ref{exp:ecdsa} challenger does not
        enforce an order.
\end{itemize}

\begin{lemma}
Let $\ECDSAAdv[\calA,\G,N,\calE]$ denote adversary $\calA$'s advantage in
Experiment~\ref{exp:ecdsa} with group $\G$ of prime order $q$, number of presignature and
signing queries $N \in \N$, and a set of tweaks $\calE$ of size polynomial in the
security parameter. Then if $\Hash$ is
collision resistant and $\G$ is a random oracle,
    $$ \ECDSAAdv[\calA,\G,N,\calE] \leq O(N \cdot( (|\calE| + N)/q + 1/\sqrt{q}))~.$$
\label{lemma:gs}
\end{lemma}

We refer the reader to Groth and Shoup's Theorem 4 for the analysis proving
Lemma~\ref{lemma:gs} in the generic group model. In our setting, the number of
public keys and therefore size of $\calE$ is polynomial.

\begin{figure}
\begin{framed}
{\small
\Experiment{ECDSA with presignatures and additive key derivation}\label{exp:ecdsa2}
The experiment is parameterized by a number of $\KeyGen$ queries $\ell \in \N$,
a number of $\PreSign$ and $\Sign$ queries $N \in \N$, a group $\G$ of
prime order $q$ with generator $g$, message space $\calM$, a hash
function $\Hash: \calM \to \Z_q$, a conversion function
$f: \G \to \Z_q$.
\begin{itemize}
    \item The challenger initializes state $\initdone = 0$, $\presigdone = 0$.
    \item The adversary can make $\ell$ $\KeyGen$ queries and $N$ $\PreSign$
        and $\Sign$ queries.
    \item $\KeyGen() \to \pk$:
        \begin{itemize}
            \item If $\initdone = 0$, $k=1$; otherwise $k \gets k + 1$.
            \item Sample $\sk_k \getsr \Z_q$.
            \item Set $\presigctr, \authctr \gets 0$.
            \item Set $\initdone \gets 1$, $\presigdone \gets 0$.
            \item Output $g^{\sk_k}$.
        \end{itemize}
    \item $\PreSign() \to R$:
        \begin{itemize}
            \item If $\initdone = 0$ or $\presigdone = 1$, output $\bot$.
            \item Sample $r_\presigctr \getsr \Z_q^\ast$.
            \item Output $g^{r_\presigctr}$ and set $\presigctr \gets \presigctr + 1$.
        \end{itemize}
    \item $\Sign(m, \tweak, j) \to \sigma$:
        \begin{itemize}
            \item If $\initdone = 0$, $\presigctr < \authctr$, or $j > k$ or $j < 1$, output $\bot$.
            \item Let $R \gets g^{r_\authctr}$
            \item Let $s \gets r_\authctr^{-1} \cdot (\Hash(m) + (\sk_j + \tweak) \cdot f(R))$.
            \item Set $\presigdone \gets 1$, $\authctr \gets \authctr + 1$.
            \item Output $(s,f(R))$.
        \end{itemize}
\end{itemize}
The output of the experiment is ``1'' if:
    \begin{itemize}
        \item the signature $\sigma^\ast$ on $m^\ast$ verifies under $\pk$,
        \item $m^\ast$ was not an input to $\Sign$, and
        \item $\pk$ was an output of $\KeyGen$.
    \end{itemize}
The output is``0'' otherwise.
}
\end{framed}
\caption{Our experiment for security of ECDSA with additive key derivation
    and presignatures.}
\label{fig:mod-exp}
\end{figure}

\begin{figure}[t]
    \begin{framed}
    {\small
        \Experiment{Groth-Shoup ECDSA with presignatures and additive key derivation~\cite{GS21}}\label{exp:ecdsa}

        We recall the security experiment from Groth and Shoup~\cite{GS21}
        in Figure 4 for ECDSA with presignatures with the modifications
        described for additive key derivation.

        The experiment is parameterized by the number of total queries the adversary
        can make $N \in \N$, a group $\G$ of prime order
        $q$ with generator $g$, message space $\calM$, a hash function $\Hash: \calM \to \Z_q$, a conversion
        function $f: \G \to \Z_q$, and a set of tweaks $\calE \subseteq \Z_q$.
        \begin{itemize}
            \item The challenger initializes state:
                \begin{itemize}
                    \item $k \gets 0$, $K \gets \emptyset$
                    \item $d \getsr \Z_q$, $D \gets g^d \in \G$.
                \end{itemize}
            \item The adversary can make $N$ total queries (presign or sign).
            \item Presignature query:
                \begin{itemize}
                    \item $k \gets k + 1$, $r_k \getsr \Z_q^\ast$
                    \item $R_k \gets g^{r_k} \in \G$
                    \item $t_k \gets f(R_k) \in \Z_q$. If $t_k = 0$, output $\bot$
                    \item Return $R_k$
                \end{itemize}
            \item Signing request for message $m$ with presignature $k \in K$ and tweak $\tweak \in \Z_q$:
                \begin{itemize}
                    \item $K \gets K \backslash \{k\}$
                    \item $h_k \gets \Hash(m) \in \Z_q$
                    \item If $h_k + t_kd + t_k \tweak = 0$, output $\bot$
                    \item $s_k \gets r_k^{-1}(h_k + t_kd + t_k \tweak) \in \Z_q$
                    \item Return $(s_k, t_k, R_k)$
                \end{itemize}
            \item After making $N$ queries, the adversary must output
                $(m^\ast, \sigma^\ast, \omega^\ast)$.
        \end{itemize}
        The output of the experiment is ``1'' if:
        \begin{itemize}
          \item the signature $\sigma^\ast$ on $m^\ast$ verifies under $D \cdot g^{\tweak^\ast}$,
          \item $m^\ast$ was not an input to a previous signing query,
        and 
      \item $\tweak^\ast \in \calE$.
\end{itemize}
        The output is ``0'' otherwise.
    }
    \end{framed}
    \caption{Security experiment for ECDSA with additive key derivation and
    presignatures from Groth and Shoup~\cite{GS21}.}
    \label{fig:gs-exp}
\end{figure}

All that remains is to prove that the the adversary in the modified
Experiment~\ref{exp:ecdsa2} does not have a greater advantage than the
adversary in the original Experiment~\ref{exp:ecdsa}.

\begin{lemma}
    Let $\ECDSAModAdv[\calA,\G,\ell,N]$ denote the adversary $\calA$'s advantage
    in Experiment~\ref{exp:ecdsa2} with group $\G$ and $\ell,N \in \N$.
    Let $\ECDSAAdv[\calB,\G,N,\calE]$ denote the adversary $\calB$'s advantage in
    Experiment~\ref{exp:ecdsa} with a set of tweaks $\calE$ of size $\ell$.
    Then given an adversary $\calA$ in Experiment~\ref{exp:ecdsa2}, we construct
    an adversary $\calB$ for Experiment~\ref{exp:ecdsa} that runs in time
    linear in $\calA$ such that for all groups $\G$, $N \in \N$, and
    randomly sampled set of tweaks $\calE$ of size $\ell$,
    $$ \ECDSAModAdv[\calA,\G,\ell,N] \leq \ECDSAAdv[\calB,\G,N,\calE]~. $$
    \label{lemma:ecdsa-mod}
\end{lemma}

\begin{proof}
We prove the above theorem by using an adversary $\calA$ in
Experiment~\ref{exp:ecdsa2} to construct an adversary $\calB$ in
Experiment~\ref{exp:ecdsa}. We then show that the adversary $\calB$
has an advantage greater than or equal to the adversary $\calA$.

We construct $\calB$ in the following way:
\begin{itemize}
    \item Rather than sending $\calA$ the set of tweaks $\calE$ immediately,
        $\calB$ keeps the set of tweaks to use to respond to $\KeyGen$ queries.
        On the $i$th invocation of $\KeyGen$, $\calB$ sends
        $D \cdot g^{\tweak_i}$ where $\tweak_i \in \calE$.
    \item $\calB$ simply forwards presignature requests from $\calA$ to the
        Experiment~\ref{exp:ecdsa} challenger and sends the responses back
        to $\calA$.
    \item $\calB$ takes signing requests from $\calA$ with an index $j$ and a
        tweak $\omega$. $\calB$ then computes $\omega_j + \omega = \omega'$
        where $\omega_j$ was the value returned by the $j$th call to $\KeyGen$
        (if $\calA$ has not made $j$ calls to $\KeyGen$, $\calB$ outputs $\bot$).
        $\calB$ then forwards the signing request with $\omega'$ to
        the Experiment~\ref{exp:ecdsa} challenger.
    \item $\calB$ additionally enforces that all presignature queries must be made
        before signing queries and that presignatures must be used in order.
        If $\calA$ sends queries that do not meet these requirements, $\calB$
        outputs $\bot$.
\end{itemize}

The adversary $\calA$ cannot distinguish between interactions with $\calB$
and the Experiment~\ref{exp:ecdsa2} challenger, and so the advantage of
$\calA$ is less than or equal to that of $\calB$, completing the proof.

\end{proof}

\paragraph{Zero-knowledge proof of preimage}
In \sys, the log takes as input a hash of the message rather than
the message itself.
It is important for security that the log has a zero-knowledge
proof of the preimage of the signing digest, as ECDSA with
presignatures is completely insecure if the signing oracle signs
arbitrary digests directly instead of messages~\cite{GS21}.
Because the log checks a zero-knowledge proof certifying that
the digest preimage is correctly encrypted
before signing, it will not sign arbitrary purported hashes
generated by a malicious client (the party submitting the hash
must know the preimage for the proof to verify).

 \section{Two-party ECDSA with preprocessing}
\label{app:two-ecdsa}

In this section, we describe the construction of our two-party
ECDSA with preprocessing protocol and argue its security.
In \cref{app:two-ecdsa:syntax}, we describe the syntax and construction
of our protocol. In \cref{app:secdefs:halfmul}, we explain a sub-protocol and
argue that it is secure. Finally in \cref{app:secdefs:proof}, we prove
that our overall protocol is secure by showing that the protocol achieves
the ideal functionality captured by the challenger in Experiment~\ref{exp:ecdsa2}
from \cref{app:secdefs:gs}.

\subsection{Syntax and construction}
\label{app:two-ecdsa:syntax}
For our purposes, a two-party ECDSA signature scheme consists of the
following algorithms:
\begin{itemize}
    \item $\LogKeyGen() \to (\sksiglog, \pksiglog)$: Generate a log secret
        key $\sksiglog \in \Z_q$ and corresponding public key $\pksiglog \in \G$.
        The log runs this routine at enrollment.
    \item $\PreSign() \to (\presig_0, \presig_1)$: Generate presignature
        $(\presig_0, \presig_1)$ where each presignature should be used to sign exactly once.
        The client runs this routine many times at enrollment to generate a batch of presignatures.
    \item $\ClientKeyGen(\pksiglog) \to (\sksigclient, \pk)$: Given the log public key
        $\pksiglog \in \G$, output a secret key share $\sksigclient \in \Z_q$ and
        corresponding public key $\pk \in \G$.
        The client runs this routine during registration with each relying party.
\end{itemize}
We additionally define the following signing protocol:
\begin{itemize}
    \item $\Pi_\Sign$: Both parties take as input the message
        $m \in \calM$, the log takes as input log 
        secret key $\sksiglog$ and presignature~$\presig_0$, and the client takes
        as input a secret-key share~$\sksigclient$ and presignature~$\presig_1$.
        The joint output is a signature~$\sigma$ on message $m$ or $\bot$
        (if either misbehaved).
\end{itemize}
The signing protocol outputs ECDSA signatures that verify under $\pk$ output
by $\ClientKeyGen$, and 
so the signature-verification algorithm is exactly as in ECDSA.

We include the constructions for the above algorithms and signing protocols in
\cref{fig:sigs}. The construction for the $\PiHalfMul$ is in \cref{app:secdefs:halfmul}
and the opening protocol $\PiOpen$ is the same opening protocol used in SPDZ~\cite{SPDZ12}.

\begin{figure}[t]
    \begin{framed}
    {\small
    The ECDSA signature scheme for message space $\calM$ uses a group
    $\G$ of prime order $q$
    and is parameterized by a hash function $\Hash: \calM \to \Z_q$ and
    a conversion function $f: \G \to \Z_q$.
    An ECDSA keypair is a pair $(y, g^y) \in \Z_q \times \G$ for
    $y \getsr \Z_q$. 
    We use a secure multiplication protocol $\PiMul$ (\cref{fig:halfmul}
    in \cref{app:secdefs:halfmul})
    and a secure opening protocol $\PiOpen$ that returns the result or $\bot$
    (``output'' step in Figure 1 of SPDZ~\cite{SPDZ12}).

    \medskip
    
    \underline{$\LogKeyGen() \to (\sksiglog, \pksiglog)$:}
            \begin{itemize}
                \item Sample $x \getsr \Z_q$ and output $(x, g^{x})$.
            \end{itemize}
    \smallskip
    \underline{$\PreSign() \to (\presig_0, \presig_1)$:}
            \begin{itemize}
                \item Sample $r \getsr \Z_q^\ast$ and compute $R \gets f(g^{r})$.
                \item Sample $\alpha \getsr \Z_q$ and compute
                    $\hat{r} \gets \alpha \cdot r^{-1}$.
                \item Split $r^{-1}$ into secret shares $r_0, r_1$;
                    $\hat{r}$ into $\hat{r_0}, \hat{r_1}$;
                    $\alpha$ into $\alpha_0, \alpha_1$;
                \item Output $(R, r_0, \hat{r_0}, \alpha_0)$,
                    $(R, r_1, \hat{r_1}, \alpha_1)$.
            \end{itemize}
    \smallskip
    \underline{$\ClientKeyGen(\pksiglog) \to (\sksigclient, \pk)$:}
            \begin{itemize}
                \item Sample $y \getsr \Z_q$.
                \item Output $(y, \pksiglog \cdot g^{y})$.
            \end{itemize}

    \smallskip
    
    \underline{$\Pi_\Sign$:}

    \smallskip
    
    We refer to the log as party $0$ and the client as party $1$.
    The input of party $i \in \zo$ is $(m, \sk_i, \presig_i)$,
    and the output is a signature $\sigma$ on $m$ or $\bot$.
  
    \smallskip

        For each party $i \in \zo$:
        \begin{itemize}
            \item Party $i$ parses $\presig_i$ as $(R, r_i, \hat{r_i}, \alpha_i)$.
            \item Given input $(r_i, \hat{r_i}, \sk_i, \alpha_i)$ for party $i$,
                run $\PiHalfMul$ to compute shares $(x_i, \hat{x_i}, v_i, \hat{v_i})$
                where $\hat{v_i}$ authenticates any intermediate 
                values.
            \item Party $i$ computes $s_i \gets r_i \cdot \Hash(m) + x_i \cdot R $.
            \item Party $i$ computes $\hat{s_i} \gets \hat{r_i} \cdot \Hash(m) + \hat{x_i} \cdot R $.
            \item Parties run $\PiOpen$ with party $i$ input $\alpha_i, s_i, \hat{s_i}, v_i, \hat{v_i}$
                to get $s$ or $\bot$ (if returns $\bot$, output $\bot$).
            \item Output $(R,s)$.
        \end{itemize}
        }
\end{framed}
     \caption{Two-party ECDSA signing protocol with preprocessing.}
    \label{fig:sigs}
\end{figure}

\subsection{Malicious security with half-authenticated secure multiplication}
\label{app:secdefs:halfmul}

As part of our signing protocol, we use a half-authenticated secure multiplication
sub-protocol. We describe the protocol in \cref{fig:halfmul}.

We need to ensure that by deviating from the protocol, neither the client nor 
the log can learn secret information (i.e. the other party's share of the 
secret key or signing nonce) or produce a signature for a different message.
To use tools for malicious security (e.g. information-theoretic
MACs~\cite{SPDZ12}) in a black-box way, we need authenticated shares of
the signing nonce and the secret key.
We can easily generate authenticated shares of the signing nonce as part of the 
presignature, but generating authenticated shares of the secret key 
poses several problems: (1) presignatures are generated at enrollment
(before the client has secret-key shares), and (2) we
don't want which presignature the client uses to leak which
relying party the client is authenticating to. 

\paragraph{Ideal functionality.}
At a very high level, the ideal functionality takes as input additive
shares of $x,y$ and outputs shares of $x \cdot y$.
In order to perform the multiplication, we use Beaver triples.
To authenticate $x$, we use information-theoretic MAC tags (because
there is no MAC tag for $y$, each party can adjust its share by an
arbitrary additive shift without detection).
More precisely then, the ideal functionality takes as inputs additive shares of $(a,b,c)$,
$(f,g,h)$, $(x,\hat{x},y)$, and $\alpha$ such that $a \cdot b = c$,
$f \cdot g = h$, $(f,g,h) = \alpha \cdot (a,b,c)$, and
$\hat{x} = \alpha \cdot x$. 
Each party outputs intermediate value
$d$ and additive shares of $\hat{d},z,\hat{z}$ where $x \cdot y = z$
and $\hat{d} = \alpha \cdot d$.

\paragraph{Protocol.}
Our protocol uses information-theoretic MACs for only
one of the inputs (the signing nonce).
We call this protocol $\PiHalfMul$.
Our construction uses authenticated Beaver triples and follows naturally
from the SPDZ protocol~\cite{SPDZ12}.
We also use a secure opening protocol $\PiOpen$ for checking MAC tags, which
we can instantiate using the SPDZ protocol directly~\cite{SPDZ12}.
It is safe to not authenticate one of the key shares due to the fact that
the signature scheme is secure if the adversary can request signatures
for arbitrary ``tweaks'' of the secret key (\cref{app:secdefs:gs}).

\begin{figure}[t]
    {\small
    \begin{framed}
        \underline{$\PiHalfMul$}: \\
        The protocol is parameterized by a prime $q$.

      \itpara{Inputs:} 
        Party $i \in \zo$ takes as input an additive share 
        of each of the $\Z_q$ values: $(a, b, c)$, $(f, g, h)$, $(x, \hat x, y)$, and $\alpha$,
        such that: 
\begin{align*}
  a \cdot b &= c &\in \Z_q\\
          f \cdot g &= h&\in \Z_q\\
  (f,g,h) &= \alpha \cdot (a, b, c) &\in \Z_q^3\\
          \hat{x} &= \alpha \cdot x &\in \Z_q
          \intertext{
        \itpara{Outputs:}
        Each party outputs intermediate value $d \in \Z_q$ and additive shares of $\hat d$, $z$, and $\hat z$,
where:}
          x \cdot y &= z &\in \Z_q\\
\hat d &= \alpha \cdot d &\in \Z_q\rlap{.}
        \end{align*}
\hrule

\smallskip

\paragraph{Protocol.}
      Each party $i \in \zo$ computes:
        \begin{align*}
d_i &\gets x_i - a_i &\in \Z_q\\
                    e_i &\gets y_i - b_i &\in \Z_q\\
                    \hat{d}_i &\gets \hat{x}_i - f_i &\in \Z_q
            \intertext{and sends $d_i,e_i$ to the other party.
            Each party $i$ then computes:}
                    d &\gets d_0 + d_1 &\in \Z_q\\
                    e &\gets e_0 + e_1 &\in \Z_q\\
                    z_i &\gets de + db_i + ea_i + c_i &\in \Z_q\\
                    \hat{z_i} &\gets de \cdot \alpha_i + dg_i + ef_i + h_i &\in \Z_q
           \end{align*}
                and outputs $d \in \Z_q$ and shares $\hat{d_i}, z_i, \hat{z_i} \in \Z_q$.
\end{framed}
    }
    \caption{$\PiMul$ protocol.}
    \label{fig:halfmul}
\end{figure}

We present a slightly modified version of the SPDZ protocol~\cite{SPDZ12} for
multiplication on authenticated secret-shared inputs in \cref{fig:halfmul}.
The only difference is that, in our protocol, only one of the inputs
is authenticated.
This requirement means that we only authenticate one of the intermediate
values in the Beaver triple multiplication.
We allow the attacker to add arbitrary shifts to the unauthenticated
input, but the attacker cannot shift the authenticated input
without detection (Claim \ref{claim:halfmul}).

The protocol $\PiHalfMul$ allows us to model authenticating the signing
nonce in ECDSA signing. The signing key is unrestricted (we discuss
why this is secure in \cref{app:secdefs:gs}).

We first show the security of our $\PiHalfMul$ protocol for secure multiplication
where only one of the inputs is authenticated, which is very
similar to the multiplication protocol in SPDZ~\cite{SPDZ12}.
This allows us to ensure that both parties use
the correct signing nonce from the presignature.

\begin{claim}
Let $x,y,\alpha,\hat{x} \in \Z_q$ be inputs to
$\PiMul$ secret-shared across the parties where $\hat{x} = \alpha x$.
Then, the probability that an adversary
that has statically corrupted one of the parties can cause the protocol
$\PiMul$ to output shares of $z, \hat{z}, \hat{d}$ where $\hat{z} = \alpha \cdot z$ and
$z = (x + \Delta)y$ for some $\Delta \neq 0$ is $1/q$.
\label{claim:halfmul}
\end{claim}

\begin{proof}
Let the input Beaver triple be $(a,b,c) \in \Z_q$ such that $a \cdot b = c$
where to multiply values $x,y$, we use intermediate values
$d = x-a$ and $e = y-b$, where $\hat{d}$ is the MAC tag for $d$.

To avoid detection, the adversary needs to ensure that $\hat{d} = \alpha \cdot d$,
so the adversary needs to
find some $\Delta_1, \Delta_2 \in \Z_q$ such that
$$ \alpha(x + \Delta_1 - a) = \hat{x} + \Delta_2 - \alpha \cdot a \quad \in \Z_q$$
which we can reduce to
$$ \alpha (x + \Delta_1) = \hat{x} + \Delta_2 \quad \in \Z_q$$
The probability of the adversary choosing $\Delta_1, \Delta_2 \in \Z_q$ that
satisfies this equation is the probability of guessing $\alpha$, or
  $1/q$.
The value $e$ does not depend on $x$.
Therefore, since the remainder of the operations are additions and
multiplications by public values, the attacker
can only shift the final output by $\Delta_3, \Delta_4$ and needs to ensure
the following:
$$ \alpha(xy + \Delta_3) = \alpha xy + \Delta_4 \quad \in \Z_q$$
The probability of finding such $\Delta_3, \Delta_4$ is the probability
of guessing $\alpha$, which is $1/q$.
\end{proof}
 
\subsection{Security proof for our construction}
\label{app:secdefs:proof}

Recall our construction of our two-party signing scheme in \cref{fig:sigs}.

As the construction in \cref{fig:sigs} contains separate algorithms for
key generation and generating presignatures, we define $\PiGen$ and
$\PiPreSign$ in terms of the algorithms in \cref{fig:sigs} below:

\begin{itemize}
    \item $\PiGen$:
        \begin{itemize}
            \item If $\pk_0$ is not initialized, the log runs
                $(\sk_0,\pk_0) \gets \LogKeyGen()$ and sends
                $\pk_0$ to the client.
            \item If $k$ is not initialized, set to 1; otherwise
                $k \gets k + 1$.
            \item The client runs $(\sk_{1,k}, \pk_{k}) \gets
                \ClientKeyGen(\pk_0)$
        \end{itemize}
    \item $\PiPreSign$:
        \begin{itemize}
            \item Client runs $(\presig_0,\presig_1) \gets \PreSign()$
                and sends $\presig_0$ to the log.
        \end{itemize}
\end{itemize}

Additive key derivation models the fact that the secret key shares
are unauthenticated private inputs to $\PiHalfMul$,
and so the adversary can run the signing protocol
with any secret key share as input. However, to produce a forgery, the
adversary must generate a signature that verifies under a small, fixed
set of public keys (corresponding to the public keys generated at registration
before compromise).

We define the ideal functionality $\FECDSA$ as simply the routines the challenger
runs to respond to $\KeyGen$,$\PreSign$, and $\Sign$ queries in Experiment~\ref{exp:ecdsa2}.

We prove security using the ideal functionality $\FOpen$ for opening values and
checking MAC tags from SPDZ~\cite{SPDZ12}. At a high level, $\FOpen$ takes as input
the party's output and intermediate shares and MAC tags for output and intermediate shares
and outputs the combined output or abort if the MAC tags are not correct.
The corresponding simulator $\SimOpen$ takes as input one party's shares of the output
and intermediate values and their corresponding MAC tags, as well as the combined output
value.

\begin{theorem}
The two-party ECDSA signing protocol $\Pi$ securely realizes (with abort)
$\FECDSA$ in the $\FOpen$-hybrid
model in the presence of a single statically corrupted malicious party
(if the client is the compromised party, it can only be compromised after
presigning is complete).
Specifically, let $\View^\Real_\Pi$ denote the adversary $\Adv$'s
view in the real world.
Then there exists a probabilistic polynomial-time algorithm $\Sim$
where $\View^\Ideal_\Sim$ denotes the view simulated by $\Sim$
given the outputs of $\FECDSA$ where $\Adv$ can adaptively choose which
procedures to run and the corresponding inputs
such that
    $$ \View^\Real_\Pi \approx \View^\Ideal_\Sim~. $$
\end{theorem}

\begin{proof}
Our goal is to construct a simulator where the simulator takes as input the outputs of
the ideal functionality $\calF$ (as well at the public input message for signing).
The adversary $\Adv$ should then not be able to distinguish
between the real world (interaction with the protocol) and the ideal world (interaction
with the simulator where the simulator is given the compromised party's inputs and
the outputs of $\FECDSA$).

Let $i$ be the index of the compromised party ($i=0$ for compromised client, $i=1$ for compromised log).
The simulator always generates the presignatures and outputs the presignature share to the
adversary, in order to model the fact that we only provide security if the client is
malicious at signing time.
We construct the simulator as follows:
\begin{itemize}
    \item $\KeyGen(\pk)$:
        \begin{itemize}
            \item If $i=0$:
                \begin{itemize}
                    \item If $\sk_0$ is not initialized, sample $\sk_0 \getsr \Z_q$ and send $\pk_0 \gets g^{\sk_0}$ to
                                $\Adv$.
                    \item Otherwise, send nothing to $\Adv$.
                \end{itemize}
            \item Otherwise if $i=1$:
                \begin{itemize}
                    \item If $\pk_0$ is not initialized, receive $\pk_0$ from $\calA$.
                    \item Output $\pk$.
                    \item Set $\initdone \gets 1$, $\presigdone \gets 0$.
                \end{itemize}
            \item Set $\presigctr,\authctr \gets 0$.
            \item Set $\initdone \gets 1$.
        \end{itemize}
    \item $\PreSign(R)$:
        \begin{itemize}
            \item If $\initdone = 0$ or $\presigdone = 1$, output $\bot$.
            \item Sample $\alpha_0, \alpha_1 \getsr \Z_q$, $r_0, r_1 \getsr \Z_q^\ast$.
            \item Set $\alpha^{(\presigctr)} \gets \alpha_0 + \alpha_1, r^{(\presigctr)} \gets r_0 + r_1$.
            \item Sample $\hat{r_0},\hat{r_1}$ such that $\alpha^{(\presigctr)} \cdot r^{(\presigctr)} = \hat{r_0} + \hat{r_1}$.
            \item Sample shares of $(a,b,c)$,$(f,g,h) \in \Z_q^3$ such that $a \cdot b = c$,
                $f \cdot g = h$, $(f,g,h) = \alpha^{(\presigctr)} (a,b,c)$,
                $\hat{x} = \alpha \cdot x$.
            \item Let $T^{\presigctr}_j = (a_j,b_j,c_j,f_j,g_j,h_j)$ for $j \in \{1,2\}$.
            \item Let $\presig^{(\presigctr)}_0 = (R, r_0, \hat{r_0}, \alpha_0, T^{(\presigctr)}_0)$ and
                $\presig^{(\presigctr)}_1 = (R, r_1, \hat{r_1}, \alpha_1, T^{(\presigctr)}_1)$.
            \item Send $\presig^{(\presigctr)}_{1-i}$ to $\Adv$.
            \item Set $\presigctr \gets \presigctr + 1$.
        \end{itemize}
    \item $\Sign(m, \sigma)$:
        \begin{itemize}
            \item If $\initdone = 0$ or $\presigctr < \authctr$, output $\bot$.
            \item Parse $\sigma$ as $(s, \_)$.
            \item Parse $\presig^{(\authctr)}_i$ as $(R, r_i, \hat{r_i}, \alpha_i, T_i)$.
            \item Parse $T_i^{(\authctr)}$ as $(a_i,b_i,c_i,f_i,g_i,h_i)$.
            \item Let $x \gets \frac{s - r^{(\authctr)} \cdot \Hash(m)}{R}$
            \item Let $\sk_i \gets x / r_i$
            \item Let $\hat{d_i} = \hat{r_i} - f_i$.
            \item Send $d_i = r_i - a_i$ and $e_i = \sk_i - b_i$ to $\calA$.
            \item Receive $d_{1-i}$ and $e_{1-i}$ from $\calA$ and compute $d = d_0 + d_1$
                and $e = e_0+e_1$.
            \item Let $x_i \gets d e + d b_i + e a_i + c_i$
            \item Let $\hat{x_i} \gets d e \alpha_i + d g_i + e f_i + h_i$
            \item Compute $s' \gets r^{(\presigctr)} \cdot \Hash(m) + x \cdot R$
            \item Compute $s_i \gets r_i \cdot \Hash(m) + x_i \cdot R + s - s'$
            \item Compute $\hat{s_i} \gets \hat{r_i} \cdot \Hash(m) + \hat{x_i} \cdot R + \alpha^{(\authctr)} (s - s')$
            \item Run $\SimOpen$ on $(\alpha_i, s_i, \hat{s_i}, d, \hat{d_i}, s)$; if $\SimOpen$
                aborts, also abort. 
            \item Set $\presigdone \gets 1$ and $\authctr \gets \authctr + 1$.
        \end{itemize}
\end{itemize}
We now prove that the view generated by $\Sim$ in the ideal world is indistinguishable
from the real world.

We start with the real world (\cref{fig:sigs}).
We then replace calls to $\PiOpen$ with $\SimOpen$.
Because we are in the $\FOpen$-hybrid model, the adversary cannot distinguish between these.

Every other message sent to $\Adv$ is either (1) a value that is random or information theoretically
indistinguishable from random, or
(2) a value generated by $\FECDSA$ ($R$ in $\PreSign$).

The last step is to show that if $\calA$ deviates from the protocol when sending $d_{1-i}$
to the simulator, the simulator can detect this and abort. By \cref{claim:halfmul}
and the guarantees of $\FOpen$, the adversary cannot send an incorrect value for
$d_{1-i}$ without detection except with probability $1/q$. The adversary \emph{can}
send any value for $e_{1-i}$; this is equivalent to the adversary being allowed to
choose any signing tweak $\tweak$, which the attacker can do in Experiment~\ref{exp:ecdsa2}.

Therefore, $\Adv$ cannot distinguish between the real world and the ideal world except with
probability $1/q$, completing the proof.
\end{proof}

 \section{\Sys for passwords}
\label{app:pw}

\subsection{Zero-knowledge proofs for discrete log relations}
\label{sec:pw-app:dlog}

The protocol of \cref{sec:pw} requires the client to prove to 
the log service that the ElGamal decryption of a ciphertext
decrypts to one of $n$ values in a set.
To do so, the client uses a zero-knowledge proof of discrete-log
relations, whose syntax and construction we describe here.
The proof system uses a cyclic group $\G$ of prime order $q$.

The proof system consists of two algorithms:
\begin{itemize}
\item $\DLProof.\Prove(\idx, x, h, \cm_1, \dots, \cm_n) \to \pi$:\\
      Output a proof $\pi$ asserting that $\cm_\idx = h^x \in \G$ for $\idx \in [n]$
\item $\DLProof.\Verify(\pi, h, \cm_1, \dots, \cm_n) \to \zo$:\\
      Check the prover's claim that it knows some $x \in \Z_q$ where 
      $\cm_\idx = h^x \in \G$ for $\idx \in [n]$.
\end{itemize}
We require the standard notions of \textit{completeness},
\textit{soundness} (against computationally bounded provers),
and \textit{zero knowledge}~\cite{BFM88,GMR85}
(in the random-oracle model~\cite{ROM}). 
We instantiate $\DLProof$ using proof techniques from
Groth and Kohlweiss~\cite{GK14}.

\subsection{Protocol for passwords}

We now describe the syntax of our $\LarchPW$ scheme.

\itpara{Step \#1: Enrollment with log service.}
At enrollment, the client and log generate cryptographic
keys and exchange public keys.

\begin{description}[leftmargin=1.5em,labelindent=0em]
    \item $\LarchPW.\ClientGen() \to (x,X)$: The client
        outputs a secret key $x \in \Z_q$ and a
        public key $X \in \G$.
    \item $\LarchPW.\LogGen() \to (k,K)$: The log
        service outputs a secret key $k \in \Z_q$ and
        a public key $K \in \G$.
\end{description}

\vfill\eject

\itpara{Step \#2: Registration with relying party.}
Once the client has enrolled with a log service, it can register
with a relying party by interacting with the log service.

\begin{description}[leftmargin=1.5em,labelindent=0em]
    \item $\LarchPW.\ClientRegister() \to (\id, k_\id)$:
        The client outputs an identifier $\id \in \zo^\lambda$
        and a key $k_\id \in \G$.
    \item $\LarchPW.\LogRegister(k,\id) \to y$:
        Given the log's secret key $k$ and $\id$ produced
        by $\ClientRegister$, the log outputs $y \in \G$.
    \item $\LarchPW.\FinishRegister(k_\id, y) \to \pw_\id$:
        Given the key $k_\id$ generated by $\ClientRegister$
        and the value $y$ generated by $\LogRegister$, the
        client outputs the password $\pw_\id \in \G$.
\end{description}

\itpara{Step \#3: Authentication with relying party.}
After registration, the client and log service perform authentication
together.

\begin{description}[leftmargin=1.5em,labelindent=0em]
    \item $\LarchPW.\ClientAuth(\idx, x, \id_1, \dots, \id_n) \to
        (r, \ct, \pi_1, \pi_2)$:
        Given an index $\idx \in \{1, \dots, n\}$, the client's
        secret key $x \in \Z_q$, and identifier values $\id_1,
        \dots, \id_n$ output by $\ClientRegister$, the client
        outputs $r \in \Z_q^\ast$, an ElGamal ciphertext
        $\ct \in \G^2$, and proofs $\pi_1$ and $\pi_2$.
    \item $\LarchPW.\LogAuth(\ct,\pi_1,\pi_2,\id_1,\dots,\id_n) \to y$:
        Given a ciphertext $\ct \in \G^2$, proofs $\pi_1$ and $\pi_2$,
        and identifiers $\id_1, \dots, \id_n$ output by $\ClientRegister$,
        the log service outputs $y \in \G$.
    \item $\LarchPW.\FinishAuth(x,K,r,k_\id,y) \to \pw_\id$:
        Given the client's secret key $x \in \Z_q$, the log's public
        key $K \in \G$, the nonce $r \in \Z_q^\ast$ generated by
        $\ClientAuth$, the key $k_\id \in \G$ generated by $\ClientRegister$,
        and the value $y \in \G$ from $\LogAuth$, output the password
        $\pw_\id \in \G$.
\end{description}

\itpara{Step \#4: Auditing with log service.}
Given a ciphertext, the client runs ElGamal decryption to recover the
corresponding $\Hash(\id)$ value.

We give a detailed description of the \sys password-based
authentication protocol in \cref{fig:pw}.

\begin{figure}
\begin{framed}
\textbf{Larch password-based authentication scheme.}
The protocol is parameterized by: a cyclic group $\G$ 
of prime order $q$ with generator $g \in \G$,
a hash function $\Hash\colon \zo^* \to \G$, 
and a zero-knowledge discrete-log proof system
$\DLProof$ with syntax as in \cref{sec:pw-app:dlog}.

\medskip

\underline{$\LarchPW.\ClientGen() \to (x, X)$}
\begin{itemize}
    \item Sample $x \getsr \Z_q$.
    \item Output $(x,g^x)$.
\end{itemize}

\underline{$\LarchPW.\LogGen() \to (k, K)$}:
\begin{itemize}
    \item Sample $k \getsr \Z_q$.
    \item Output $(k,g^k)$.
\end{itemize}

\underline{$\LarchPW.\ClientRegister() \to (\id, k_\id)$}:
\begin{itemize}
    \item Sample $\id \getsr \zo^\lambda$.
    \item Sample $k_\id \getsr \G$.
    \item Output $(\id, k_\id)$.
\end{itemize}

\underline{$\LarchPW.\LogRegister(k, \id) \to y$}:
\begin{itemize}
    \item Output $\Hash(\id)^k$.
\end{itemize}

\underline{$\LarchPW.\FinishRegister(k_\id, y) \to \pw_\id$}:
\begin{itemize}
    \item Output $k_\id \cdot y$.
\end{itemize}

\underline{$\LarchPW.\ClientAuth(\idx, x, \id_1, \dots, \id_n) \to (r,\ct,\pi_1,\pi_2)$}:
\begin{itemize}
    \item Sample $r \getsr \Z^\ast_q$
    \item Compute $c_1 = g^r, c_2 = \Hash(\id_\idx) \cdot g^{xr}$.
    \item Let $h_i = c_2 / \Hash(\id_i)$ for $i \in \{1,\dots,n\}$.
    \item Let $\pi_1 \gets \DLProof.\Prove(\idx, r, X, h_1, \dots, h_n)$.
    \item Let $\pi_2 \gets \DLProof.\Prove(\idx, x, c_1, h_1, \dots, h_n)$.
    \item Let $\ct = (c_1,c_2)$.
    \item Output $(r,\ct, \pi_1, \pi_2)$.
\end{itemize}

\underline{$\LarchPW.\LogAuth(\ct,\pi_1,\pi_2,\id_1,\dots, \id_n) \to y$}:
\begin{itemize}
    \item Parse $\ct$ as $(c_1,c_2)$.
    \item Let $h_i = c_2 / \Hash(\id_i)$ for $i \in \{1,\dots,n\}$.
    \item Let $b_1 \gets \DLProof.\Verify(\pi_1, X, h_1, \dots, h_n)$
    \item Let $b_2 \gets \DLProof.\Verify(\pi_2, c_1, h_1, \dots, h_n)$
    \item If $b_1 \neq 1$ or $b_2 \neq 1$, output $\bot$
    \item Output $c_2^k$
\end{itemize}

\underline{$\LarchPW.\FinishAuth(x,K,r,k_\id,y) \to \pw_\id$}:
\begin{itemize}
    \item Output $k_\id \cdot y \cdot K^{-xr}$.
\end{itemize}

\end{framed}
\caption{The details of the \sys protocol
for password-based authentication.}
\label{fig:pw}
\end{figure}
 \fi

\end{document}